\documentclass[11pt]{article}
\input{style.sty}
\usepackage[T1]{fontenc}
\usepackage{lmodern,microtype,booktabs,longtable,array,enumitem}
\usepackage[colorlinks=true,linkcolor=blue!45!black,citecolor=blue!45!black,urlcolor=blue!45!black]{hyperref}
\usepackage{bookmark}
\hypersetup{pdftitle={Quantum Locally Testable CSS Codes},pdfsubject={Cubical construction, expansion analysis, and simultaneous realization}}
\numberwithin{equation}{section}
\newcommand{\PP}{\mathbb P}

\DeclareMathOperator{\rank}{rank}

\DeclareMathOperator{\mult}{mult}
\DeclareMathOperator{\length}{length}
\DeclareMathOperator{\reg}{reg}

\DeclareMathOperator{\lcm}{lcm}
\allowdisplaybreaks[2]
\usepackage{csquotes}
\usepackage{xurl}
\usepackage[
  backend=biber,
  bibencoding=utf8,
  texencoding=ascii,
  style=alphabetic,
  sorting=anyt,
  backref=true,
  maxbibnames=3,
  minbibnames=1,
  maxalphanames=3,
  minalphanames=1,
  giveninits=false,
  doi=true,
  url=true
]{biblatex}
\title{Asymptotically Good Quantum Locally Testable Codes}
\author{
William Gay\thanks{
University of Illinois Urbana-Champaign.
Email: \texttt{\{whgay2,granha\}@illinois.edu}
} \and 
Fernando Granha Jeronimo\footnotemark[1]
}
\date{September 17, 2026}

\begin{document}
\maketitle

\vskip 2cm
\begin{abstract}
We construct explicit families of asymptotically good quantum locally testable codes over qubits. More precseily, we construct explicit quantum LDPC CSS codes over qubits having constant rate, constant relative distance and constant weight local testers with constant soundness.
\end{abstract}
\thispagestyle{empty}

\newpage
\pagenumbering{roman}
\tableofcontents
\newpage

\pagenumbering{arabic}
\setcounter{page}{1}
\section*{Introduction}
\addcontentsline{toc}{section}{Introduction}
\label{sec:introduction}
\begingroup
\clubpenalties 2 10000 0
\widowpenalty=10000


Error-correcting are a fundamental object in computer science, in both practical and theoretical application. A good code stores a constant fraction of its block
length in alphabet symbols while ensuring that distinct codewords
differ in a constant fraction of their coordinates. Rate and distance are global
properties of a code, they constrain the information carried by an entire word and
the separation between entire words. A complementary requirement is that
the code admit a sparse description by local constraints. Gallager's
low-density parity-check (LDPC) codes have parity-check matrices with
bounded row and column weights. That is, each check involves only a bounded
number of symbols, and each symbol participates in only a bounded number
of checks \cite{Gallager1962}. This sparsity limits the number of
interactions needed to compute a syndrome and provides the structure
used by iterative decoding algorithms.

Tanner generalized this viewpoint by placing small constituent codes
on a bipartite graph and imposing their constraints on overlapping
subsets of coordinates \cite{Tanner1981}. Sipser and Spielman used
expander graphs to construct asymptotically good LDPC codes with
linear-time decoding \cite{SipserSpielman1996}. Their work illustrates
how the expansion of a sparse graph can turn local constraints into
global distance and decoding guarantees. The constructions considered
here follow this general approach, with higher-dimensional complexes
replacing graphs and stronger properties required of the local codes.

Local testability asks whether the global coding structure can be
recognized by examining only a few coordinates.
A locally testable code (LTC) has a randomized test that always accepts
codewords and rejects any other word with probability at least a constant
multiple of its relative distance from the code. When the test reads
only a constant number of coordinates, it can reject words at constant
relative distance from the code with constant probability, independently
of the block length. Sparse parity checks provide inexpensive tests,
but sparsity alone does not establish this quantitative relation
between rejection probability and distance.

Unlike many other desirable code property, 
being locally testable is not a property of random codes. 
Goldreich and Sudan developed a systematic theory of locally testable
codes and constructed constant-query binary LTCs of almost-linear
length \cite{GoldreichSudan2006}. Ben-Sasson and Sudan obtained
quasilinear-length PCPs and codes testable with polylogarithmically
many queries at any fixed positive distance threshold, using
univariate Reed--Solomon techniques \cite{BenSassonSudan2008}.
These works brought the block length close to linear in the message
length, while leaving a loss in rate or query complexity. The problem
of simultaneously obtaining constant rate, relative distance, query
complexity, and soundness therefore required further ideas.

Panteleev and Kalachev (PK) \cite{PanteleevKalachev2022} and,
independently, Dinur, Evra, Livn\'e, Lubotzky, and Mozes (DELLM)
\cite{DELLM2022} resolved the constant-rate LTC problem for classical codes. Both constructions give LTCs with positive
constant rate, relative distance, and soundness, using a constant
number of queries. PK combine Tanner codes through a lifted-product
construction and prove expansion properties of the resulting chain
complexes; their paper also establishes asymptotically good quantum
LDPC codes. DELLM place local tensor-code constraints on left-right
Cayley square complexes and use the expansion of the complex together
with local agreement properties to prove global testability. The
journal treatment of DELLM streamlines this argument and supplies
the local codes through expander-code constructions \cite{DELLM2022}.
These results are direct predecessors of the present work, as they show
how local tensor-code properties can be combined with expanding
two-dimensional geometry, while the quantum testability argument here
requires corresponding control in more directions.

Local testing is closely connected to probabilistically checkable proofs
(PCPs), in that one can view the latter as a computational variant of the former information theoretic object. The work of Arora and Safra \cite{AroraSafra1998} and of
Arora, Lund, Motwani, Sudan, and Szegedy \cite{ALMSS1998}
established that every language in NP has polynomial-length proofs
checkable using logarithmically many random bits and constantly many
proof-bit queries, with perfect completeness and constant soundness
error. Dinur subsequently gave a proof through gap amplification
\cite{Dinur2007}. Tests of algebraic encodings and tests of consistency
between local views are central components of these developments.
The connection is particularly concrete for PCPs of
proximity, in which a verifier checks that a given word is close to a
specified language with the aid of an additional proof. Such proximity
proofs can be used to construct locally testable codes, and coding
techniques in turn enter their construction \cite{BGHSV2006}.

Quantum error-correcting codes encode a space of quantum states into a
larger Hilbert space. For a stabilizer code, membership in the codespace
is specified by commuting checks, each of which is a measurement with
accept and reject outcomes. The CSS construction of Calderbank and
Shor \cite{CalderbankShor1996} and Steane \cite{Steane1996}
builds such quantum codes from classical linear codes, with an
orthogonality condition ensuring that the two types of checks commute.
For a quantum LDPC code, both the weight of each check and the number
of checks involving each qubit are bounded independently of the block
length. Achieving sparsity together with the commutation requirement
is a substantial additional constraint on quantum constructions.

This sparsity is especially important for quantum error correction,
where checks must be measured on the encoded state. Bounded check
weight permits syndrome-extraction circuits using bounded-size
ancillary systems for individual checks, and bounded qubit incidence
limits the number of such measurements involving any data qubit.
Gottesman showed how suitable quantum LDPC families, together with
decoding and circuit assumptions, can support fault-tolerant quantum
computation with constant qubit overhead \cite{Gottesman2014}.

Explicit quantum LDPC codes are also a fairly recent development. 
Tillich and Z\'emor's hypergraph-product construction
gives constant rate and distance proportional to the square root of
the block length \cite{TillichZemor2014}. Leverrier, Tillich, and
Z\'emor developed linear-time decoding for quantum expander codes
from this framework \cite{LeverrierTillichZemor2015}. Fiber-bundle
codes of Hastings, Haah, and O'Donnell broke the square-root distance
barrier by a polynomial factor \cite{HastingsHaahODonnell2021},
and Panteleev and Kalachev obtained almost-linear distance through
lifted products \cite{PanteleevKalachev2022Distance}.
PK's subsequent construction achieved both
\cite{PanteleevKalachev2022} constant rate and linear distance. 
Leverrier and Z\'emor's quantum Tanner
codes give another asymptotically good family, formulated using
local codes on left-right Cayley complexes \cite{LeverrierZemor2022}.
The existence of these good quantum LDPC families establishes the
compatibility of sparsity, rate, and distance. Global local-testability
soundness is an additional requirement.

Sampling a stabilizer check gives a natural local test, and quantum
local testability asks that its rejection probability control distance
from the codespace.
Aharonov and Eldar initiated its systematic study
\cite{AharonovEldar2015}. We use the quantum Hamming-distance
formulation of Eldar and Harrow \cite{EldarHarrow2017}: the expected
fraction of violated checks must control the normalized distance from
the codespace. Within a joint eigenspace of the stabilizer checks, this
distance is the minimum support size of a Pauli error producing that
syndrome; for a general state, it is the expectation of the corresponding
distance observable. This asks the test to detect errors according to
the number of qubits affected, rather than according to Hilbert-space
or trace distance. The precise convention appears in
Section~\ref{sec:css-preliminaries}.

Quantum locally testable codes are motivated in part by the quantum PCP
conjecture, which concerns the hardness of approximating the ground
energy of a local Hamiltonian to constant additive accuracy after
normalization \cite{AharonovAradVidick2013}. A stabilizer code gives a
particularly structured local Hamiltonian: its ground space is the
codespace, and its energy is the expected fraction of violated checks.
Local testability relates that energy to distance from the ground
space. This is a quantum local-verification problem closely related in
motivation to PCPs, although constructing good quantum locally testable codes
does not by itself provide the reductions required by the quantum PCP
conjecture \cite{AharonovEldar2015,DLV2024}.
There is also a connection to the complexity of low-energy states.
Eldar and Harrow related quantum local testability to obstructions to
preparing such states by shallow circuits \cite{EldarHarrow2017}.
Anshu, Breuckmann, and Nirkhe subsequently proved the no-low-energy
trivial-state (NLTS) conjecture using a particular family of good
quantum LDPC codes \cite{AnshuBreuckmannNirkhe2023}. 

Several constructions have obtained quantum local testability with
different parameter tradeoffs. Leverrier, Londe, and Z\'emor's
hemicubic codes have inverse-logarithmic soundness and logarithmic
check weight, while encoding one logical qubit
\cite{LeverrierLondeZemor2022}. Cross, He, Natarajan, Szegedy, and
Zhu constructed constant-soundness qLTCs with tradeoffs among rate,
distance, and locality \cite{CrossEtAl2024}. In their first
construction the distance scales with the locality, their second
has constant average locality, which does not bound the largest
check weight. Dinur, Lin, and Vidick (DLV) developed a
higher-dimensional cubical construction whose bounded-locality
instantiation has inverse-polylogarithmic losses in relative distance
and soundness \cite{DLV2024}. Their abstract construction and
local-to-global expansion analysis are the starting point of the present work.

\subsection*{Main Results}

Our main result is the existence of asymptotically good quantum locally
testable codes:
\begin{theorem*}
There exists a family of binary CSS codes of unbounded length with constant rate and relative distance. 
Moreover, the codes are both LDPC and qLTC with constant soundness. 
\end{theorem*}
The formal statement is Theorem~\ref{thm:main-qltc}. The construction
uses the middle degree of a four-dimensional cubical complex, its
abstract formulation and analysis allow a general number of directions.
The proof also establishes a uniform product-expansion theorem for a
family of Reed--Solomon codes on unequal evaluation
sets, including the dual codes needed in the quantum construction
(Theorem~\ref{exloc:finite-grid-pe}).

\subsection*{Tools}

The use of higher-dimensional expansion to study testing has a broader
history. Kaufman and Lubotzky identified the connection between
coboundary expansion and property testing \cite{KaufmanLubotzky2014}.
On the coding side, Ben-Sasson and Sudan developed robust local
testability and tensor-code tests \cite{BenSassonSudan2006}.
These viewpoints relate the global distance from satisfying a system
of constraints to information obtained from its local restrictions.
Product expansion is the particular local tensor-code property needed
for the construction considered here.

Product expansion describes how local codewords can combine on a
Cartesian grid. Fix a code for each coordinate direction, and consider
arrays that are sums of codewords placed on coordinate lines. Such a
sum may have extensive cancellation. Product expansion says that the
same array has a decomposition whose total line cost is controlled by
its nonzero entries, where a used line is charged its length. It thus
turns small support in the resulting array into a quantitatively small
description by local codewords. This is the local property
needed in the cubical construction of~\cite{DLV2024}.

Panteleev and Kalachev used product expansion of pairs of local codes
in their constructions \cite{PanteleevKalachev2022}. In higher
dimensions, robust and agreement testability do not imply product
expansion, as Kalachev exhibited Reed--Solomon tensor powers with at
least three factors separating these notions \cite{Kalachev2026}.
This counterexample concerns repeated factors and does not rule out
the unequal evaluation sets used in this work. Kalachev and
Panteleev also established product-expanding tuples with any
fixed number of factors over sufficiently large fields, with an
expansion constant independent of the common block length
\cite{KalachevPanteleev2025}. Their result also permits simultaneous
expansion of a tuple and its tuple of dual codes. 
In the present construction, the local codes must additionally
respect symmetries imposed by the underlying cubical geometry, so an unconstrained
random choice of local codes does not supply the required input.

Reed--Solomon codes encode a polynomial by evaluating it at a finite
set of field elements \cite{ReedSolomon1960}. Their dimension and
distance follow from polynomial interpolation and the root bound for
a nonzero polynomial. For our purposes, their symmetries are equally
important. In projective form, these codes evaluate homogeneous
polynomials on a projective line. An invertible linear change of the two
homogeneous coordinates over the field defining that projective line
permutes the evaluation points and rescales the resulting coordinates
by nonzero scalars. This provides a natural way to match
local code coordinates to projective incidence data.
Here we prove product expansion for a particular family of projective
evaluation sets whose symmetries also support the local coefficient
maps.

The geometric input is the Ramanujan property. For a regular graph,
this bounds the nontrivial adjacency eigenvalues by the spectral radius
of its covering tree, the trivial eigenvalues include both signs of
the degree when the graph is bipartite. Lubotzky, Phillips, and Sarnak 
constructed arithmetic families
of graphs with this property \cite{LPS1988}. Jordan and Livn\'e
developed its cubical analogue using quotients of products of trees
\cite{JordanLivne2000}. Rungtanapirom, Stix, and Vdovina constructed
further arithmetic lattices in products of trees and cubical Ramanujan
complexes from their congruence subgroups \cite{RungtanapiromStixVdovina2019}.
The relevant operators now move a face to a
parallel face in one direction. Their spectral bounds control how a
small set of faces can meet its transported copies, which is the
geometric estimate used in the coding analysis.

Hsieh, Lubotzky, Mohanty, Reiner, and Zhang (HLMRZ) constructed
Ramanujan cubical complexes from the LPS framework in their work on
lossless vertex expansion \cite{HLMRZ2025}. These give a natural
noncommuting counterpart to the cubical geometry in DLV, directions
can be exchanged around a square even though the individual generators
need not commute. Their construction motivates replacing the commuting
geometry by arithmetic cubical complexes whose expansion bounds remain
uniform as the quotient grows. Our realization uses the related
Jordan--Livn\'e and Livn\'e arithmetic constructions
\cite{JordanLivne2000,Livne2001}.

\subsection*{Proof Overview}

The abstract construction builds on DLV's
abstract framework \cite{DLV2024}. Vector spaces and incidence maps are
assigned to cubical faces, with simultaneous local tensor coordinates
coming from the base-code encoders. Compatible incidence maps give a
square-zero differential and hence commuting CSS checks. We allow
unequal directional degrees, endpoint-dependent code dimensions, and
local coordinates in place of global generator labels. The expansion
analysis is nearly identical to DLV's, product expansion gives local
fillings, which force a cocycle resistant to local weight reduction
to have many neighbors in its support. Face-walk mixing rules out a
small such support. Applying this mechanism to syndromes gives
soundness, DLV's double-complex argument transfers the estimates to
homology using the dual base codes. Section~\ref{comb:section}
assembles the resulting abstract criterion.

The main issue is to realize the geometric and local-code hypotheses
simultaneously. DLV's geometric small-set expansion range causes the
inverse-polylogarithmic loss. Ramanujan cubical complexes of the kind
suggested by HLMRZ offer a geometry without this loss, but their
noncommuting labels obstruct a naive substitution, the two routes
around a square can change the local code coordinates, so the two
coded incidence compositions need not agree. To obtain the differential,
we choose base codes respecting these changes of coordinates.

Section~\ref{sec:arithmetic} uses the Jordan--Livn\'e/Livn\'e
framework to construct a tower of finite quotients of a product of
trees. Smaller congruence subgroups give larger quotients with fixed
local degrees, sufficiently deep quotients retain every full upward
star. The arithmetic input gives uniform spectral bounds for all
required parallel-face walks, whose connected components occupy fixed
positive fractions of the relevant face classes. This removes the
geometric loss. This particular arithmetic realization allows the
characteristic-two residue fields and local stabilizer actions needed.

At a tree vertex, incident edges form a projective line over the
residue field. Projective Reed--Solomon evaluation respects the
stabilizer's action, including its coordinate rescalings. Across an
edge, the two endpoint coordinate systems exchange the stabilizer's
diagonal entries, so their scalar actions must be matched. We pair
two nearby polynomial degrees and reindex one evaluation code by a
Frobenius automorphism, acting on evaluation arguments rather than
coefficients. Section~\ref{sec:equivariance} thereby builds equivariant
incidence maps on each tree, tensors them, and descends them to the
quotients. The two routes around a square evaluate different tensor
factors and therefore agree, giving square zero in characteristic two
and the simultaneous local charts required by the analysis.

For product expansion, we identify each projective evaluation set
with a norm-one subgroup of a quadratic finite-field extension, up
to coordinate rescaling. The field sizes are distinct powers of two
with bounded ratios. Sections~\ref{exinc:chapter}--\ref{exloc:chapter}
convert the problem to interpolation on their Cartesian product.
A small support has a polynomial closure of controlled size. On these
norm-one sets, the corresponding Frobenius powers act as inversion,
combining the unequal powers with inversion and coefficient Frobenius
transforms curves without changing their grid intersections. An irreducible curve meeting the grid
and varying in two coordinates cannot be preserved, since its
coordinate degrees would scale differently. Intersections with a
distinct transformed curve, together with induction, control the
points of a closed set lying on no full coordinate line in that set.

This incidence estimate gives interpolation for the points indexing
maximal coordinate flats in the closure. Regularity bounds for unions
of these flats then extend compatible local polynomial data with the
required individual degree bounds. By linear duality, this extension
gives a decomposition into codewords on lines contained in the closure.
Its controlled size bounds the line cost, proving product expansion
uniformly in the local lengths, also for the required dual codes and
after extension of the coefficient field.

Finally, we fix the number of directions, the local-length ratio bound,
and rate intervals away from zero and one. The Reed--Solomon theorem
fixes a positive product-expansion constant before the residue fields
are enlarged to meet the Ramanujan threshold. All local data then
remain fixed as the congruence level, and hence the quantum block
length, grows. Using endpoint codes in two directions and their duals
in the other two gives positive middle-dimensional rate by an
Euler-characteristic calculation and bounds on the other cohomology
groups. Binary restriction of scalars incurs only constant losses.
Section~\ref{real:section} combines these choices with the abstract
distance and soundness estimates.

\endgroup

\section{Preliminaries}
\label{sec:preliminaries}
For an integer $a\geq0$, write $[a]=\{1,\ldots,a\}$.
All tensor products of indexed factors are taken in increasing index
order; an empty tensor product over a field $F$ is $F$.
For a field $F$ and a finite set $\Omega$, equip $F^\Omega$ with
the bilinear pairing $\langle x,y\rangle=\sum_{\omega\in\Omega}x_\omega y_\omega$.
For a subspace $C\leq F^\Omega$, its \emph{dual code}, or
\emph{orthogonal complement}, is
\[
 C^\perp=\{y\in F^\Omega:\langle x,y\rangle=0\text{ for every }x\in C\}.
\]
The notation $C^*=\operatorname{Hom}_F(C,F)$ denotes the vector-space dual.
For a matrix $T$, $T^{\mathsf T}$ denotes
the transpose in the specified bases. Hamming weight counts nonzero
coordinates. When coordinates are grouped into vector-valued blocks,
$|x|_b$ counts the nonzero blocks. We work in characteristic two for
the cubical complexes; coding statements not using this restriction
are stated over arbitrary fields. Write $[z^j]P$ for a polynomial
coefficient, with value zero outside its degree range, and
\[
 e_j(a_1,\ldots,a_t)=\sum_{|I|=j}\prod_{i\in I}a_i,
 \qquad e_0=1,\qquad e_j=0\quad(j\notin[0,t]).
\]
\label{bin:elementary-symmetric}

\subsection{Classical Codes and Product Expansion}
\label{cert:encoder-data}\label{cert:line-data}
Let $F$ be a field, $r\geq1$, and, for $i\in[r]$, let
$\Omega_i$ be a finite nonempty set of size $n_i$. An encoder for
an $m_i$-dimensional code $C_i\leq V_i:=F^{\Omega_i}$ is an
injective map $G_i:U_i:=F^{m_i}\to V_i$ with image $C_i$.
A parity-check matrix $H_i$ has rows forming a basis of $C_i^\perp$;
thus
\[
 C_i^\perp=\ker G_i^{\mathsf T},\qquad
 \ker H_i=C_i,\qquad \rank H_i=n_i-m_i.
\]
For $\Omega=\prod_i\Omega_i$, identify $\bigotimes_i V_i=F^\Omega$ and put
\[
 A_i=C_i\otimes\bigotimes_{j\ne i}V_j,\qquad A=\sum_i A_i.
\]
Factors in this and subsequent formulas are reordered into increasing
index order. An $i$-line is obtained by fixing all coordinates except
$i$. For $x_i\in A_i$, let $\ell_i(x_i)$ count its nonzero $i$-lines.
Thus each $i$-line of $x_i$ belongs to $C_i$.

\begin{definition}[Product expansion]\label{loc:weighted-product-expansion}
The tuple $(C_i)_{i=1}^r$ is $\rho$-product-expanding, for a real
$\rho>0$, if
\[
 \forall x\in A\quad\exists x_i\in A_i:\qquad
 x=\sum_i x_i,\qquad \rho\sum_i n_i\ell_i(x_i)\leq |x|.
\]
It is \emph{universally} $\rho$-product-expanding if this holds for
$(E\otimes_F C_i)_i$ over every field extension $E/F$, with the same
$\rho$. For a finite-dimensional $E$-space $R$, set
$A_i(E;R)=R\otimes_E A_i(E)$ and $A(E;R)=\sum_i A_i(E;R)$.
The corresponding inequality for $R$-valued entries, counting a
coordinate or line when its vector is nonzero, is called collective
product expansion. Abbreviate these spaces to $A_i(R),A(R)$ when
the field is understood.
\end{definition}

\begin{lemma}\label{cert:product-check}\label{cert:product-kernel}
For the preceding encoders and checks,
\[
 \ker\Bigl(\bigotimes_i H_i\Bigr)=A,
 \qquad A^\perp=\bigotimes_i C_i^\perp.
\]
\end{lemma}
\begin{proof}
Choose $V_i=C_i\oplus Q_i$. Expanding the tensor product, the summands
containing a $C_i$ factor form exactly $A$. The remaining summand is
$\bigotimes_iQ_i$, on which $\bigotimes_iH_i$ is an isomorphism onto
its target. This proves the kernel formula. Taking orthogonal complements gives
$A^\perp=\im(\bigotimes_iH_i^{\mathsf T})=\bigotimes_iC_i^\perp$.
\end{proof}

\begin{lemma}\label{cert:scalar-extension}\label{exloc:scalar-extension}\label{cert:monomial-invariance}\label{exloc:monomial-isometry}
Ranks, kernels, and images of matrices over $F$ commute with extension
of the coefficient field. Independent permutations and nonzero
rescalings of the coordinates in each $\Omega_i$ preserve product
expansion, including universal and collective expansion. If $T_i$
is such a monomial map, then
$(T_iC_i)^\perp=T_i^{-\mathsf T}C_i^\perp$.
\end{lemma}
\begin{proof}
A rank normal form $PBQ=\operatorname{diag}(I_a,0)$ remains a rank
normal form after extending the field. This also identifies the
extended kernel and image. The map $\bigotimes_iT_i$ bijects grid
points and directional lines, preserving their zero or nonzero
status, so it transports each expanding decomposition without changing
any of its weights. Finally
$\langle T_ix,T_i^{-\mathsf T}y\rangle=\langle x,y\rangle$.
\end{proof}

\subsection{CSS Codes and Local Testability}
\label{sec:css-preliminaries}\label{bin:pauli}\label{bin:css}
For $N\geq1$, the $N$-qubit space is $\mathscr H_N=\C^{\F_2^N}$,
with orthonormal basis $(e_a)_a$. For $x,z\in\F_2^N$, set
\[
 X^xe_a=e_{a+x},\qquad Z^ze_a=(-1)^{z\cdot a}e_a,
 \qquad \operatorname{wt}(X^xZ^z)=|\supp x\cup\supp z|.
\]
Let $H_X\in\F_2^{m_X\times N}$ and
$H_Z\in\F_2^{m_Z\times N}$ satisfy $H_XH_Z^{\mathsf T}=0$,
and suppose $m=m_X+m_Z>0$. The CSS code $\mathscr Q$ is the common
$+1$ eigenspace of $X^h$ for the rows $h$ of $H_X$ and $Z^h$
for the rows of $H_Z$. Repetitions and zero rows are retained in
this specified check list. Its dimension and distance are
\begin{align}
 \log_2\dim_\C\mathscr Q
 &=N-\rank H_X-\rank H_Z,\label{eq:css-dimension}\\
 d(\mathscr Q)
 &=\min\left\{
 \min_{z\in\ker H_X\setminus\im H_Z^{\mathsf T}}|z|,
 \min_{x\in\ker H_Z\setminus\im H_X^{\mathsf T}}|x|
 \right\},\label{eq:css-distance}
\end{align}
where the minimum of the empty set is $+\infty$.
The rate is $(\log_2\dim\mathscr Q)/N$.

For an integer $a\geq0$, let $\mathscr Q_{\leq a}$ be the span
of $P\mathscr Q$ over Pauli operators $P$ of weight at most $a$,
and let $\Pi_{\leq a}$ be its orthogonal projector; put
$\Pi_{\leq-1}=0$. Define
\[
 \mathsf D_{\mathscr Q}=\sum_{a=0}^N a(\Pi_{\leq a}-\Pi_{\leq a-1}),
 \qquad
 \mathsf H=\frac1m\left(
 \sum_{h\text{ row of }H_X}\frac{I-X^h}{2}
 +\sum_{h\text{ row of }H_Z}\frac{I-Z^h}{2}\right).
\]
\begin{definition}[Quantum local testability]\label{def:qltc}
The specified CSS code has soundness at least $\sigma>0$ if
$\mathsf H\succeq (\sigma/N)\mathsf D_{\mathscr Q}$, where
$A\succeq B$ means that $A-B$ is positive semidefinite.
Its check weight is the maximum row weight of the two matrices;
its qubit incidence is the maximum, over columns, of the sum of
their two column weights.
\end{definition}
The formulas above and the syndrome criterion for this operator
inequality are proved in Appendix~\ref{sec:css-appendix}. In particular,
the normalization uses the full listed number $m$ of checks.

\section{Construction}

\subsection{Cubical Complexes}
\label{sec:cubical-complexes}

Fix an integer $t \ge 1$, and write $[t] := \set{1,\ldots,t}$.
The face poset of the $t$-dimensional cube is
\[
\mcC_t := \set{0,1,\star}^{[t]}, \qquad
c \le d \quad\Longleftrightarrow\quad
\forall i \in [t],\quad c_i=d_i \ \text{or}\ d_i=\star.
\]
For a poset $P$ and $f,g \in P$, write
\[
P_{\le f} := \set*{h\in P:h\le f}, \qquad
P_{\ge f} := \set*{h\in P:f\le h}, \qquad
[f,g] := P_{\ge f}\cap P_{\le g}.
\]

\begin{definition}[Cubical complexes]
\label{def:cubical-complex}\label{cube:finite-cube-data}\label{cube:intervals}
A $t$-dimensional cubical complex is a finite poset
$( X,\le)$ equipped with a map
\[
\phi: X\to\mcC_t, \qquad f\mapsto\wt f,
\]
whose image contains $(\star,\ldots,\star)$, such that the following hold:
\begin{enumerate}
\item For every $f\in X$, the restriction
$\phi|_{ X_{\le f}}: X_{\le f}\to(\mcC_t)_{\le\wt f}$
is a poset isomorphism.
\item For every $f,g\in X$, the intersection
$ X_{\le f}\cap X_{\le g}$ is either empty or equal to
$ X_{\le h}$ for some $h\in X$.
\end{enumerate}
\end{definition}

The elements of $ X$ are its faces. For $f\in X$, set
\[
\dir(f):=\set*{i\in[t]:\wt f_i=\star}, \qquad
\codir(f):=[t]\setminus\dir(f), \qquad
\dim(f):=|\dir(f)|.
\]
The set $\dir(f)$ is also called the type of $f$.
Faces of dimensions zero and one are called vertices and edges,
respectively. No empty face is included. For $I\subseteq[t]$ and
$k\in\Z$, write
\[
 X(I):=\set*{f\in X:\dir(f)=I}, \qquad
 X(k):=\set*{f\in X:\dim(f)=k},
\]
and set $ X_{\le f}(k):= X_{\le f}\cap X(k)$ and
$ X_{\ge f}(k):= X_{\ge f}\cap X(k)$.
Thus $ X(k)=\varnothing$ outside $0\le k\le t$.
If $ X_{\le f}\cap X_{\le g}= X_{\le h}$, then $h$ is
unique: equality of the lower sets of $h$ and $h'$ implies
$h\le h'$ and $h'\le h$. We therefore write $f\cap g=h$ in this
case and $f\cap g=\varnothing$ otherwise.

The lower-cube isomorphism gives, for every $f\le g$,
\[
\dir(f)\subseteq\dir(g), \qquad
[f,g]\xrightarrow{\sim}2^{\dir(g)\setminus\dir(f)}, \qquad
h\longmapsto\dir(h)\setminus\dir(f),
\]
where the set of subsets is ordered by inclusion. Indeed, a word
between $\wt f$ and $\wt g$ is determined by which coordinates in
$\dir(g)\setminus\dir(f)$ are changed to stars. Similarly, for
$g\in X(k)$ and $0\le\ell\le k$,
\begin{equation}
\label{eq:lower-face-count}
| X_{\le g}(\ell)|=\binom{k}{\ell}2^{k-\ell}.
\end{equation}
There are $\binom{k}{\ell}$ choices of the stars that remain,
and two choices for each of the other $k-\ell$ coordinates.
In particular, every face has a vertex, and every interval whose
endpoints differ in dimension by two has exactly two intermediate faces.

For $x\in X(0)$ and $i\in[t]$, let
\[
E_i(x):=\set*{e\in X(\{i\}):x\le e}.
\]
Fix integers $n_1,\ldots,n_t\ge2$, and write
$n_-:=\min_{i\in[t]}n_i$ and $n_+:=\max_{i\in[t]}n_i$.
We call $ X$ $(n_1,\ldots,n_t)$-regular if $|E_i(x)|=n_i$ for every
$x\in X(0)$ and $i\in[t]$, and if, for every $x\in X(0)$,
$I\subseteq[t]$, and $(e_i)_{i\in I}\in\prod_{i\in I}E_i(x)$,
\begin{equation}
\label{eq:cubical-regularity}
\left|\set*{f\in X(I):x\le f,\quad e_i\le f\ \text{for all }i\in I}\right|=1.
\end{equation}
The product over an empty index set is a singleton; for
$I=\varnothing$, the unique face in \eqref{eq:cubical-regularity} is $x$.

\begin{lemma}
\label{lem:cubical-stars}\label{cube:full-stars}
Suppose $ X$ is a $t$-dimensional $(n_1,\ldots,n_t)$-regular
cubical complex. Fix $I\subseteq[t]$, $f\in X(I)$, and
$x\in X_{\le f}(0)$.
There is a poset isomorphism
\[
\theta_{f,x}:\bigsqcup_{J\subseteq\codir(f)}
\left(\{J\}\times\prod_{j\in J}E_j(x)\right)
\xrightarrow{\sim} X_{\ge f},
\]
where $(J,e)\le(J',e')$ means $J\subseteq J'$ and $e'|_J=e$.
Its value $\theta_{f,x}(J,e)$ is the unique face $g\in X(I\cup J)$
containing $f$ and every $e_j$ for $j\in J$.
Consequently, for every $S$ with $I\subseteq S\subseteq[t]$,
\begin{equation}
\label{eq:upper-face-count}
| X_{\ge f}\cap X(S)|=\prod_{i\in S\setminus I}n_i.
\end{equation}
\end{lemma}
\begin{proof}
For every $i\in I$, the lower cube of $f$ contains a unique
edge $a_i\in E_i(x)$. Given $J\subseteq\codir(f)$ and
$e\in\prod_{j\in J}E_j(x)$, regularity supplies a unique
$g\in X(I\cup J)$ through $x$ containing the edges
$(a_i)_{i\in I}$ and $(e_j)_{j\in J}$. The type-$I$ face of $g$
through $x$ contains all the $a_i$, so the uniqueness clause of
\eqref{eq:cubical-regularity} identifies it with $f$.
This defines $\theta_{f,x}$.

Conversely, every $g\ge f$ contains $x$, and its lower cube
contains a unique edge through $x$ in each direction of $\dir(g)$.
Its edges in $I$ are the $a_i$, while the remaining edges recover
$J=\dir(g)\setminus I$ and $e$. Regularity then recovers $g$,
proving bijectivity. If $g\le g'$, their recovered edge lists
agree on every direction of $g$. Conversely, if one edge list
extends the other, the corresponding face of the larger cube
through $x$ has the smaller edge list and hence equals the
smaller face by regularity. Thus $\theta_{f,x}$ preserves and
reflects the order. For type $S$, its domain has
$\prod_{i\in S\setminus I}|E_i(x)|=\prod_{i\in S\setminus I}n_i$
elements, proving \eqref{eq:upper-face-count}.
\end{proof}

For $I\subseteq[t]$, abbreviate $D_I:=\prod_{i\in I}n_i$, with
$D_\varnothing=1$.

\begin{lemma}
\label{lem:cubical-face-counts}\label{cube:face-counts}
Suppose $ X$ is a $t$-dimensional $(n_1,\ldots,n_t)$-regular
cubical complex, and put $V:=2^{-t}| X(0)|$. Then $V$ is a
positive integer, and for every $I\subseteq[t]$ and
$c\in\{0,1\}^{[t]\setminus I}$,
\begin{equation}
\label{eq:parity-face-count}
\left|\set*{f\in X(I):\wt f|_{[t]\setminus I}=c}\right|=VD_I.
\end{equation}
In particular, for every $I\subseteq[t]$ and every integer $0\le k\le t$,
\begin{equation}
\label{eq:cubical-face-counts}
| X(I)|=2^{t-|I|}VD_I, \qquad
| X(k)|=2^{t-k}V\sum_{\substack{I\subseteq[t]\\|I|=k}}D_I.
\end{equation}
\end{lemma}
\begin{proof}
For $a\in\{0,1\}^{[t]}$, let
$v_a:=|\{x\in X(0):\wt x=a\}|$. For $i\in[t]$, let $a^{(i)}$
be obtained by flipping coordinate $i$. Every color-$i$ edge with an endpoint
of label $a$ has its other endpoint labeled $a^{(i)}$.
Counting these edges from either endpoint gives
\[
n_i v_a=n_i v_{a^{(i)}}.
\]
Since $n_i>0$, all the $v_a$ are equal by successive coordinate
flips. Their common value is $V$, which is a positive integer
because $ X(0)\ne\varnothing$.

Fix $I,c$, and extend $c$ to $a\in\{0,1\}^{[t]}$.
By Lemma~\ref{lem:cubical-stars}, each of the $V$ vertices
labeled $a$ belongs to exactly $D_I$ faces of type $I$.
Each such face has exactly one vertex labeled $a$, by its
lower-cube isomorphism. Moreover, a type-$I$ face has such a
vertex precisely when its fixed coordinates equal $c$.
Counting the pairs $(x,f)$ with $\wt x=a$ and $x\le f\in X(I)$
therefore proves \eqref{eq:parity-face-count}. Summing over the
$2^{t-|I|}$ choices of $c$, and then over all $I$ of size $k$,
gives \eqref{eq:cubical-face-counts}.
\end{proof}

For an $(n_1,\ldots,n_t)$-regular cubical complex $ X$,
$I\subsetneq[t]$, and $j\in[t]\setminus I$, define the
parallel-face graph $G_{I,j}$ by
\[
V(G_{I,j})= X(I), \qquad E(G_{I,j})= X(I\cup\{j\}).
\]
An edge $q\in X(I\cup\{j\})$ joins its two type-$I$ faces
$f_j^0(q),f_j^1(q)$, determined by
\[
f_j^\epsilon(q)\le q, \qquad
\wt{f_j^\epsilon(q)}_j=\epsilon
\quad(\epsilon\in\{0,1\}).
\]
These are the opposite facets obtained by fixing coordinate $j$.
By \eqref{eq:upper-face-count}, every vertex belongs to exactly
$n_j$ such edges. Its random walk chooses one of these edges
uniformly and moves to the opposite face. The transition operator
$M_{I,j}:\R^{ X(I)}\to\R^{ X(I)}$ is therefore
\begin{equation}
\label{eq:parallel-face-walk}
(M_{I,j}z)(f)=\frac1{n_j}
\sum_{\substack{q\in X(I\cup\{j\})\\f\le q}}
z\!\left(f_j^{1-\wt f_j}(q)\right)
\qquad(z\in\R^{ X(I)},\ f\in X(I)).
\end{equation}

Each step preserves the coordinates outside $I\cup\{j\}$.
For $c\in\{0,1\}^{[t]\setminus(I\cup\{j\})}$, set
\[
C_{I,j,c}:=\set*{f\in X(I):\wt f|_{[t]\setminus(I\cup\{j\})}=c}.
\]
The walk restricts to $C=C_{I,j,c}$; denote its transition
operator there by $M_C$. The graph induced on $C$ is
$n_j$-regular and bipartite, with parts distinguished by
$\wt f_j\in\{0,1\}$. Both parts are nonempty by
\eqref{eq:parity-face-count}. On $\R^C$, use the counting
inner product $\langle z,w\rangle:=\sum_{f\in C}z(f)w(f)$
and its norm $\|z\|_2:=\langle z,z\rangle^{1/2}$. Write
$\mathbf{1}_C(f):=1$ and $\sigma_C(f):=(-1)^{\wt f_j}$ for $f\in C$.
Since every edge has opposite endpoints and is chosen with
probability $1/n_j$ in either direction,
\[
M_C^{\mathsf T}=M_C, \qquad
M_C\mathbf{1}_C=\mathbf{1}_C, \qquad M_C\sigma_C=-\sigma_C.
\]

\begin{definition}
\label{def:expanding-cubical-complex}
Let $0\le\lambda<1$. An $(n_1,\ldots,n_t)$-regular cubical
complex $ X$ is $\lambda$-expanding if, for every
$I\subsetneq[t]$, $j\in[t]\setminus I$, and
$c\in\{0,1\}^{[t]\setminus(I\cup\{j\})}$, the induced graph
$G_{I,j}|_{C_{I,j,c}}$ is connected and
\begin{equation}
\label{eq:cubical-spectral-expansion}
\|M_C z\|_2\le\lambda\|z\|_2
\qquad
\text{for every }z\in\operatorname{span}\{\mathbf{1}_C,\sigma_C\}^{\perp},
\qquad C=C_{I,j,c}.
\end{equation}
\end{definition}
Thus the prescribed sets $C_{I,j,c}$ are exactly the connected
components of $G_{I,j}$, and \eqref{eq:cubical-spectral-expansion}
bounds the absolute values of all eigenvalues other than the
two bipartite eigenvalues $1$ and $-1$ on each component.

\subsection{Cosheaves and Compatible Base Codes}
\label{sec:cosheaves}
Fix a finite field $F=\F_{2^s}$ and an $(n_1,\ldots,n_t)$-regular
cubical complex $X$.

\begin{definition}\label{diag:ascending-system}
A cosheaf $\mathcal F$ on $X$ assigns a coordinate space
$\mathcal F(f)=F^{d_f}$ to each face and a linear map
$R_{gf}:\mathcal F(f)\to\mathcal F(g)$ to each $f\leq g$, with
\[
 R_{ff}=I,\qquad R_{hg}R_{gf}=R_{hf}\quad(f\leq g\leq h).
\]
Zero-dimensional stalks are allowed.
\end{definition}

Fix integers $0\leq m_{i,b}\leq n_i$ for $i\in[t]$, $b\in\{0,1\}$.
For $f\in X$, put $J_f=\codir(f)$ and
$U_{f,i}=F^{m_{i, \wt{f}_i}}$ for $i\in J_f$.

\begin{definition}\label{def:compatible-base-codes}
\label{cube:star-chart-certificate}\label{hom:array-star}
Compatible base-code data for $\mathcal F$ consist, for each face $f$, of
an anchor vertex $x_f\leq f$, injective linear encoders
\[
 E_{f,i}:U_{f,i}\hookrightarrow F^{\Omega_{f,i}},
 \qquad \Omega_{f,i}=E_i(x_f)\quad(i\in J_f),
\]
and isomorphisms
\[
 \Psi_{f,g}:\mathcal F(g)\xrightarrow{\sim}
 \bigotimes_{i\in\codir(g)}U_{f,i}\quad(g\geq f).
\]
They must simultaneously satisfy the following identities. If
$T\subseteq J_f$, $\omega_T\in\prod_{i\in T}\Omega_{f,i}$,
$j\in J_f\setminus T$, and $\omega_j\in\Omega_{f,j}$, set
\[
 g=\theta_{f,x_f}(T,\omega_T),\qquad
 h=\theta_{f,x_f}(T\cup\{j\},(\omega_T,\omega_j)).
\]
Then
\begin{equation}\label{eq:compatible-charts}
 \Psi_{f,h}R_{hg}\Psi_{f,g}^{-1}
 =\left(\operatorname{ev}_{\omega_j}\circ E_{f,j}\right)
       \otimes\bigotimes_{i\in J_f\setminus(T\cup\{j\})}I_{U_{f,i}}.
\end{equation}
The scalar output is removed by the canonical tensor identification.
Write $C_{f,i}=\im E_{f,i}$ for these base codes.
\end{definition}
The geometric indexing in this definition is supplied by regularity;
only the encoders and their simultaneous linear charts are additional
structure. No agreement between charts with different anchors is
required beyond their intertwining the same incidence maps.
In particular,
\begin{equation}\label{cube:chart-dimensions}
 d_f=\prod_{i\in\codir(f)}m_{i, \wt{f}_i},\qquad
 \mathcal F(f)\cong F\quad(\dim f=t).
\end{equation}

\subsection{Cosheaf Complexes and Their CSS Codes}
\label{diag:cubical-cochains}
For $j\in\mathbb Z$, define
\[
 C^j=C^j(X;\mathcal F)=\bigoplus_{f\in X(j)}\mathcal F(f),\qquad
 (\delta^j x)(g)=\sum_{f\in X_{\leq g}(j)}R_{gf}x(f),\qquad
 |x|_b=|\{f:x(f)\ne0\}|.
\]
The spaces are zero outside $0\leq j\leq t$. When discussing
transposes, we also use $C_j=C^j$, $D_j=\delta^j$, and
$B_j=(\delta^{j-1})^{\mathsf T}$; the subscript on $C_j$ still
specifies the cochain degree.
For a face $v$, restriction to faces containing $v$ gives its
upward complex, denoted $(C_v^\bullet,\delta_v)$, with local degree
$j-\dim v$ on global $j$-faces. The simultaneous charts identify
this complex with the tensor evaluation complex of Section~\ref{ten:section}.
They act separately in each face and hence preserve block support.

\begin{lemma}\label{cube:differential-square}
For every $j$, $\delta^{j+1}\delta^j=0$.
\end{lemma}
\begin{proof}
For $f\in X(j)$ and $h\in X(j+2)$ with $f\leq h$, the interval
$[f,h]$ has exactly two intermediate faces. Each contributes $R_{hf}$
to the composite matrix block. Their sum is zero in characteristic
two. If $f\not\leq h$ there is no contribution.
\end{proof}

Suppose henceforth that $0<m_{i,b}<n_i$ for every $i,b$.
Set $M_j=\dim_F C^j$, $Z_j=|X(j)|$ and
$h_j=\dim_F(\ker\delta^j/\im\delta^{j-1})$.
For an ordered binary basis $\beta$ of $F$, let $R_\beta(T)$ be
the matrix of an $F$-linear map $T$ after expressing every field
coordinate in $\beta$. In degree $1\leq k<t$ the construction is
\begin{equation}\label{con:matrix-construction}
 H_X=R_\beta(\delta^k),\qquad
 H_Z=R_\beta(\delta^{k-1})^{\mathsf T},\qquad N=sM_k.
\end{equation}
Since restriction of scalars respects composition,
$H_XH_Z^{\mathsf T}=0$. The listed numbers of checks are $sM_{k+1}$
and $sM_{k-1}$, and the number of encoded qubits is $sh_k$.
Every degree has positive dimension by regularity and the strict
endpoint inequalities.

Put
\begin{equation}\label{comb:endpoint-data}
 A_0=V\prod_i n_i,\qquad \mu_{i,b}=m_{i,b}/n_i,\qquad
 S_i=\mu_{i,0}+\mu_{i,1}.
\end{equation}
Counting faces of each type and parity, as in
Lemma~\ref{lem:cubical-face-counts}, gives
\begin{equation}\label{eq:cochain-count}
 M_j=\sum_{|I|=j}V\prod_{i\in I}n_i
                       \prod_{i\notin I}(m_{i,0}+m_{i,1})
     =A_0e_{t-j}(S).
\end{equation}

\begin{lemma}\label{bin:locality}
Let $n=\max_i n_i$ and $1\leq k<t$. In arbitrary bases of the
individual stalks, each row and column of $\delta^{k-1}$ and
$\delta^k$ has weight at most
\[
 w(t,k,n)=\max\{2(k+1),t-k+1\}n^{t-k+1}.
\]
Every middle stalk has dimension at most $n^{t-k}$. The binary
code has check weight at most $sw(t,k,n)$ and qubit incidence at
most $2sw(t,k,n)$.
\end{lemma}
\begin{proof}
A $(j+1)$-face has $2(j+1)$ facets, each with stalk dimension at
most $n^{t-j}$. A $j$-face has at most $(t-j)n$ immediate cofaces,
each with stalk dimension at most $n^{t-j-1}$. Thus the row and
column bounds for $\delta^j$ are $2(j+1)n^{t-j}$ and
$(t-j)n^{t-j}$, respectively. Apply these with $j=k-1,k$.
Each scalar entry becomes an $s\times s$ binary block, multiplying
row and column weights by at most $s$; a qubit belongs to two check
lists. The stalk bound is \eqref{cube:chart-dimensions}.
\end{proof}

\subsection{Main Results}
\label{sec:main-results}
Our abstract criterion separates the dimension condition from the
expansion conditions. In addition to $S_i$, put
\[
 p_i=\min_b\mu_{i,b},\qquad q_i=1-\max_b\mu_{i,b},\qquad
 \mathcal B_j=[z^j]\prod_i(p_i+q_i z),
\]
\[
 \Gamma_k=(-1)^k\prod_i(S_i-1)
       -\sum_{\substack{j\equiv k\ (2)\\0\leq j\leq t,\ j\ne k}}
                   \mathcal B_j.
\]
The explicit positive function $\mathcal L(t,k,\rho,K)$ is defined
in \eqref{geo:abstract-constants}; it is independent of the local lengths $n_i$.

\begin{theorem}\label{thm:abstract}
Fix integers $t\geq4$, $2\leq k\leq t-2$, $s\geq1$, degrees
$n_i\geq2$, and endpoint dimensions $0<m_{i,b}<n_i$. Fix reals
$K\geq1$, $0<\rho\leq1$, $0\leq\lambda<1$, with
\[
 \frac{\max_i n_i}{\min_i n_i}\leq K,\qquad
 \Gamma_k>0,\qquad \lambda\mathcal L(t,k,\rho,K)<1.
\]
Suppose $(X_\nu,\mathcal F_\nu)_{\nu\geq1}$ have compatible base
codes over $F=\F_{2^s}$ with these fixed dimensions, each $X_\nu$
is $(n_1,\ldots,n_t)$-regular and $\lambda$-expanding, and
$|X_\nu(0)|\to\infty$. Suppose every nonempty subtuple of the base
codes in each chart, and of their dual codes, is
universally $\rho$-product-expanding.
Then \eqref{con:matrix-construction} produces binary CSS codes of lengths
$N_\nu\to\infty$, rate at least $\Gamma_k/e_{t-k}(S)>0$, and
positive constant relative distance and soundness. Their check weight
and qubit incidence have the bounds of Lemma~\ref{bin:locality}.
All constants are independent of $\nu$.
\end{theorem}

\begin{theorem}\label{thm:simultaneous}
For every $t\geq4$ and $2\leq k\leq t-2$, there are fixed
parameters and a family $(X_\nu,\mathcal F_\nu)$ satisfying every
hypothesis of Theorem~\ref{thm:abstract}.
\end{theorem}

\begin{theorem}\label{thm:main-qltc}
There are positive constants $R,\Delta,\sigma$ and an integer $w$
and a family of binary CSS codes with lengths $N_\nu\to\infty$
such that their rates are at least $R$, their distances are at
least $\Delta N_\nu$, their specified check lists have soundness
at least $\sigma$, and both check weight and qubit incidence are
at most $w$.
\end{theorem}

Theorem~\ref{thm:abstract} is proved in Section~\ref{comb:section}.
Sections~\ref{exar:section}--\ref{real:section} realize its hypotheses
and prove Theorems~\ref{thm:simultaneous} and~\ref{thm:main-qltc}.
Section~\ref{real:section} also gives the fixed specialization with
rate at least $2^{-5}$.
The local-code input is a uniform product-expansion theorem for
Reed--Solomon codes on products of distinct-exponent norm-one sets,
proved in Section~\ref{exloc:chapter}. Its constant is uniform in
the set sizes and over extensions of the coefficient field; these
two uniformities are used separately in the realization and in the
local filling argument.

\section{Rate and Dimension}\label{rankcomp:section}\label{sec:rate}
Throughout this section, $\mathcal F$ has compatible base codes with
$0<m_{i,b}<n_i$. No expansion hypothesis is imposed. We prove
\begin{equation}\label{eq:rate-headline}
 \frac{h_k}{M_k}\geq\frac{\Gamma_k}{e_{t-k}(S)}
 \qquad(0\leq k\leq t),
\end{equation}
with the quantities of Section~\ref{sec:main-results}.

\subsection{Directional Differentials and Partial-Star Quotients}
\label{rankcomp:data}
Let $\delta_i$ retain the incidences adding direction $i$.
Cubical intervals and functoriality give
\begin{equation}\label{rankcomp:directional-identities}
 \delta=\sum_i\delta_i,\qquad \delta_i^2=0,\qquad
 \delta_i\delta_j=\delta_j\delta_i\quad(i\ne j).
\end{equation}
Indeed, a direction cannot be added twice, and the two paths adding
$i,j$ from $f$ to $h$ both have composite $R_{hf}$.

For $f\in X(J)$ and $T\subseteq[t]\setminus J$, put
\begin{align*}
 W_T(f)&=\bigoplus_{g\geq f,\ \dir(g)=J\cup T}\mathcal F(g),\\
 B_T(f)&=\sum_{i\in T}\im\bigl(D_{T,i}(f):W_{T\setminus\{i\}}(f)
                                      \longrightarrow W_T(f)\bigr),\\
 Q_T(f)&=W_T(f)/B_T(f),
\end{align*}
where $D_{T,i}$ uses the actual incidences adding $i$.
\label{rankcomp:partial-quotients}

\begin{lemma}\label{rankcomp:quotient-charts}
In the single chart at $f$, write $U_i=U_{f,i}$,
$V_i=F^{\Omega_{f,i}}$, $C_i=\im E_{f,i}$. Then
\[
 Q_T(f)\cong\bigotimes_{i\in T}(V_i/C_i)
                  \otimes\bigotimes_{i\notin J\cup T}U_i.
\]
For $j\notin J\cup T$, the induced map
\[
 Q_T(f)\longrightarrow\bigoplus_{h\geq f,\ \dir(h)=J\cup\{j\}}Q_T(h)
\]
is injective, with cokernel $Q_{T\cup\{j\}}(f)$.
\end{lemma}
\begin{proof}
The simultaneous chart identifies $W_T(f)$ with
$\bigotimes_{i\in T}V_i\otimes\bigotimes_{i\notin J\cup T}U_i$.
For each $i\in T$, the image of $D_{T,i}$ consists of tensors whose
$i$th factor is in $C_i$. Choose complements of all $C_i$ in $V_i$;
the tensor expansion shows that quotienting by the sum of these
images leaves precisely the displayed quotient tensor.
For the new direction $j$, collect all factors other than $j$ into
$R$. Restricting the same chart to each coface $h$ identifies the
induced map with $I_R\otimes E_{f,j}:R\otimes U_j\to R\otimes V_j$.
This map is injective. Its cokernel is $R\otimes(V_j/C_j)$.
The quotient is intrinsic: adding the $j$-relations after the
$T$-relations imposes exactly the $T\cup\{j\}$ relations, since
all directional maps commute.
\end{proof}

\subsection{Cohomology-Preserving Directional Elimination}
\begin{lemma}\label{rankcomp:one-direction}
Let a finite graded complex over a field of characteristic two have
$d=\sum_i d_i$, $d_i^2=0$, and $d_id_j=d_jd_i$. Suppose
$D=H\oplus L\oplus E$ is a graded decomposition such that
$d_jH,d_jL\subseteq E$, $d_jE=0$, each $d_i$ for $i\ne j$
preserves the three summands, and $d_j|_H$ is injective.
Then $N=H\oplus d_jH$ is an acyclic subcomplex, every $d_i$
descends to $D/N$, and $H^q(D)\cong H^q(D/N)$ for all $q$.
\end{lemma}
\begin{proof}
Commutation shows that $N$ is preserved by every $d_i$.
Define $s|_H=0$ and $s(d_jh)=h$. Injectivity makes this a
well-defined map of degree $-1$. On $H$ and on $d_jH$,
respectively, direct calculation gives $ds+sd=I$; the terms
$d_ih$ for $i\ne j$ cancel in pairs. Hence $N$ is acyclic.
A closed quotient vector represented by $x$ has $dx\in N$;
subtracting a primitive in $N$ gives a closed lift. If a closed
$x$ maps to $d\bar y$, then $x-dy$ is closed in $N$ and therefore
exact there. These two observations prove surjectivity and
injectivity of the induced cohomology map. The directional
identities descend because $N$ is invariant.
\end{proof}

Fix eliminated endpoint bits $a_i\in\{0,1\}$. For $P=[p]$,
$T\subseteq P$ and $J\subseteq[t]\setminus P$, define
\[
 \mathscr A_P(T,J)=\{f\in X(J): \wt{f}_i=a_i\ (i\in T),\
                   \wt{f}_i=1-a_i\ (i\in P\setminus T)\}.
\]
\begin{proposition}\label{rankcomp:successive-elimination}
There is a complex $C_P$ with $H^q(C_P)\cong H^q(C)$ and
\[
 C_P^q\cong\bigoplus_{\substack{T\subseteq P,\ J\subseteq[t]\setminus P\\
                              |T|+|J|=q}}
                     \ \bigoplus_{f\in\mathscr A_P(T,J)}Q_T(f).
\]
Its unprocessed directional maps are induced by the original incidences.
\end{proposition}
\begin{proof}
Induct on $p$. For $p=0$, $Q_\varnothing(f)=\mathcal F(f)$,
so take $C_\varnothing=C$. Suppose the statement holds for $P=[p]$
and put $j=p+1$. Divide the displayed summands into layers $H,L,E$
according as the $j$-state is $a_j$, $1-a_j$, or $\star$.
Every other directional map preserves these states. The $j$-map
sends $H,L$ into $E$ and vanishes on $E$.
For fixed $T,J_0$ not containing $j$, its restriction from $H$ is
the direct sum of the injective maps in
Lemma~\ref{rankcomp:quotient-charts}. Their target lists partition
$E$: each face with active direction $j$ has exactly one facet
with $j$-bit $a_j$. Thus the full map $d_j|_H$ is injective.
Apply Lemma~\ref{rankcomp:one-direction}. The quotient retains
the $L$-summands and replaces each corresponding list in $E$ by
$Q_{T\cup\{j\}}(f)$, anchored at its $a_j$ facet. These are
exactly the claimed summands for $P\cup\{j\}$, including the
increase of degree by one in the second case. The remaining
maps still come from incidences modulo the imposed relations;
commutation guarantees preservation of those relations. This
completes the induction.
\end{proof}

\begin{theorem}\label{rankcomp:betti-bound}\label{rank:compression}
For $0\leq q\leq t$ and every choice of $a_i$,
\[
 h_q\leq V\sum_{|T|=q}\prod_{i\in T}(n_i-m_{i,a_i})
                           \prod_{i\notin T}m_{i,1-a_i}.
\]
In particular $h_q\leq A_0\mathcal B_q$.
\end{theorem}
\begin{proof}
At $P=[t]$, the preceding proposition has $J=\varnothing$.
For each $T$, its anchors are the $V$ vertices of one prescribed
parity, and their quotient stalks have the dimensions in the
first displayed bound. A cohomology dimension is at most the
dimension of its cochain space. To obtain the last bound, choose
$a_i$ maximizing $m_{i,a_i}$, divide by $A_0=V\prod_i n_i$,
and expand $[z^q]\prod_i(p_i+q_i z)$.
\end{proof}

\subsection{Euler Characteristic and the Rate Bound}
\begin{proposition}\label{bin:euler-rate}
For all $0\leq j\leq t$, $M_j=A_0e_{t-j}(S)$ and
$0\leq h_j\leq A_0\mathcal B_j$. For every $k$, $h_k\geq A_0\Gamma_k$.
\end{proposition}
\begin{proof}
The first two conclusions are \eqref{eq:cochain-count} and
Theorem~\ref{rank:compression}. Put $r_j=\rank\delta^j$,
with $r_{-1}=r_t=0$. As $h_j=M_j-r_{j-1}-r_j$, the ranks
cancel in the alternating sum:
\[
 \sum_j(-1)^jh_j=\sum_j(-1)^jM_j=A_0\prod_i(S_i-1).
\]
After multiplying by $(-1)^k$ and isolating $h_k$, the terms
with parity opposite to $k$ have positive sign and may be discarded.
Bound the remaining $h_j$ by $A_0\mathcal B_j$ to obtain the claim.
Dividing by $M_k>0$ proves \eqref{eq:rate-headline}.
\end{proof}

For $1\leq k<t$ and $0<\tau<1/2$, put
\begin{equation}\label{bin:rate-functions}
 \begin{gathered}
 R_{t,k,j}(\tau)=[z^j](\tau+z)^k(1+\tau z)^{t-k},\\
 E_{t,k,j}(\tau)=[z^{t-j}](1+2\tau z)^k(1+2z)^{t-k},\\
 G_{t,k}(\tau)=(1-2\tau)^t-
 \sum_{\substack{0\leq j\leq t\\j\equiv k\ (2),\ j\ne k}}R_{t,k,j}(\tau).
 \end{gathered}
\end{equation}
\begin{corollary}\label{bin:rate-numbers}
Let $L\subseteq[t]$ have size $k$, and suppose
$0\leq\alpha_i\leq\beta_i\leq\tau$ and
\[
 \{\mu_{i,0},\mu_{i,1}\}=
 \begin{cases}\{\alpha_i,\beta_i\},&i\in L,\\
 \{1-\alpha_i,1-\beta_i\},&i\notin L.
 \end{cases}
\]
Then $\Gamma_k\geq G_{t,k}(\tau)$,
$M_j\leq A_0 E_{t,k,j}(\tau)$ for every $j$, and
$M_k\geq A_0[2(1-\tau)]^{t-k}$. If $G_{t,k}(\tau)>0$, then
\[
 \frac{h_k}{M_k}\geq\frac{G_{t,k}(\tau)}{E_{t,k,k}(\tau)},\qquad
 \frac{M_k}{M_{k-1}+M_{k+1}}\geq
 \frac{[2(1-\tau)]^{t-k}}{E_{t,k,k-1}(\tau)+E_{t,k,k+1}(\tau)}.
\]
\end{corollary}
\begin{proof}
Writing $s_i=\alpha_i+\beta_i\leq2\tau$, we have $S_i=s_i$
on $L$ and $S_i=2-s_i$ off $L$. Thus
$(-1)^k\prod_i(S_i-1)=\prod_i(1-s_i)\geq(1-2\tau)^t$.
Coefficientwise, $p_i+q_i z\leq\tau+z$ on $L$ and
$p_i+q_i z\leq1+\tau z$ off $L$, giving
$\mathcal B_j\leq R_{t,k,j}(\tau)$. Likewise
$\prod_i(1+S_i z)\leq(1+2\tau z)^k(1+2z)^{t-k}$
coefficientwise. The single term in $e_{t-k}(S)$ indexed by
$[t]\setminus L$ is at least $[2(1-\tau)]^{t-k}$.
Apply Proposition~\ref{bin:euler-rate} and divide the resulting
positive bounds.
\end{proof}

\begin{lemma}\label{bin:rate-positive}
If $t\geq2$, $1\leq k<t$, and $0<\tau\leq(8t2^t)^{-1}$,
then $G_{t,k}(\tau)>1/2$.
\end{lemma}
\begin{proof}
In the polynomial defining $R_{t,k,j}$, a choice of $z$ in positions
$I$ contributes $\tau^{|I\triangle L|}z^{|I|}$. For
$|I|\equiv k\pmod2$ and $|I|\ne k$, the exponent of $\tau$
is at least two. Hence the subtracted sum is at most $2^t\tau^2$.
Bernoulli's inequality gives
\[
 G_{t,k}(\tau)\geq1-2t\tau-2^t\tau^2
 \geq1-\frac1{4\cdot2^t}-\frac1{64t^2 2^t}>\frac12.
\]
\end{proof}

\section{Local Filling from Product Expansion}\label{ten:section}
We prove that universal product expansion supplies filling inequalities
on every upward star, with constants independent of the local lengths.
If $0<\rho\leq1$ and $K\geq1$, define
\begin{equation}\label{ten:filling-constants}
 \kappa_{r,u}(\rho,K)=
 \frac{\rho^{\binom{r-1}{u}}}{(K+1)^{\binom{r-1}{u}-1}}
 \quad(r\geq1,\ 0\leq u<r).
\end{equation}
These numbers lie in $(0,1]$. Their boundary values are $\rho$ and
Pascal's identity gives
$\kappa_{r,u}=\kappa_{r-1,u}\kappa_{r-1,u-1}/(K+1)$ for
$1\leq u\leq r-2$.
\label{ten:filling-constant-bounds}

\begin{theorem}\label{geo:general-filling-table}
Suppose $X$ is a $t$-dimensional $(n_1,\ldots,n_t)$-regular
cubical complex with $n_+/n_-\leq K$, and $\mathcal F$
has compatible base codes. Suppose every nonempty subtuple in every
chart is universally $\rho$-product-expanding. For $v\in X(\ell)$,
$\ell\leq k<t$, and $z\in C_v^{k-\ell}$,
\begin{equation}\label{geo:local-filling}
 |\delta_v z|_b\geq n_-\kappa_{t-\ell,k-\ell}(\rho,K)
                \min_{y\in C_v^{k-\ell-1}}|z+\delta_v y|_b.
\end{equation}
The assertion also holds after any field extension and tensoring
all stalks with any finite-dimensional passive vector space.
\end{theorem}

\subsection{Vector-Valued Product Expansion}
\begin{lemma}\label{cert:collective-expansion}\label{ten:collective-pe}
A universally $\rho$-product-expanding tuple has collective product
expansion with the same constant over every extension $E/F$ and
every finite-dimensional $E$-space $R$.
\end{lemma}
\begin{proof}
The cases $R=0$ or $x=0$ are immediate. Otherwise choose a basis
$r_1,\ldots,r_d$ of $R$ and embed $R$ linearly into $E(T)$ by
$r_j\mapsto T^{j-1}$. Applied entrywise to $x\in A(E;R)$, this
gives $\widetilde x\in A(E(T))$ with identical support.
Universal expansion gives $\widetilde x=\sum_i y_i$, with
$\rho\sum_i n_i\ell_i(y_i)\leq|x|$.
Let $M\leq E(T)$ be the finite-dimensional $E$-span of the image
of $R$, all entries of the $y_i$, and entries of messages encoding
each directional line of each $y_i$. Choose an $E$-linear retraction
$M\to R$ of the chosen embedding. Applying it entrywise preserves
the sum and every line-encoding equation, because the matrices have
entries in $F\subseteq E$. A zero line remains zero. The resulting
$x_i\in A_i(E;R)$ therefore have sum $x$ and satisfy the required
bound.
\end{proof}

\subsection{Tensor Evaluation Complexes and Local Exactness}
\label{ten:array-spaces}\label{ten:tensor-complex}\label{ten:array-differential}
Fix a field $F$ of characteristic two and an integer $r\geq0$.
For every $i\in[r]$, fix a finite nonempty set $\Omega_i$ of size
$n_i$, an integer $0\leq m_i\leq n_i$, and an injective $F$-linear
map $G_i:U_i=F^{m_i}\to V_i=F^{\Omega_i}$.
Fix also a finite-dimensional $F$-space $R$. Set
$\Omega=\prod_{i\in[r]}\Omega_i$ and put
\[
 T^u=\bigoplus_{|I|=u}R\otimes\bigotimes_{i\in I}V_i
                              \otimes\bigotimes_{i\notin I}U_i,
 \qquad
 (\delta x)_J=\sum_{j\in J}(I\otimes G_j)x_{J\setminus\{j\}}.
\]
Here the indicated map applies $G_j$ in its own factor and the
identity in every other factor. Spaces outside $0\leq u\leq r$
are zero. A block in the $I$-summand is indexed by
$\omega_I\in\prod_{i\in I}\Omega_i$ and has values in
$R\otimes\bigotimes_{i\notin I}U_i$; $|x|_b$ counts these blocks.
\label{ten:tensor-block-weight}

\begin{lemma}\label{ten:exactness}\label{ten:dual-exactness}
One has $\delta^2=0$ and $\ker\delta^u=\im\delta^{u-1}$ for
$0\leq u<r$. Moreover $\im\delta^{r-1}=A(R)$.
For $T_u=(T^u)^*$ and $\partial_u=(\delta^{u-1})^*$,
\[
 \ker\partial_u=\im\partial_{u+1}\quad(0\leq u<r),\qquad
 \ker\partial_r=R^*\otimes\bigotimes_i\ker G_i^{\mathsf T}.
\]
\end{lemma}
\begin{proof}
Maps in different factors commute, so every term in $\delta^2$
occurs twice. Choose $V_i=\im G_i\oplus Q_i$. The tensor complex
splits into tensor products of $[U_i\xrightarrow{I}U_i]$ and
$[0\to Q_i]$. Any summand with an identity factor is contractible:
the map removing that factor satisfies $\delta h+h\delta=I$;
mixed terms cancel in characteristic two. The only remaining
summand is $R\otimes\bigotimes_iQ_i$ in degree $r$.
For $r=0$, the complex is $R$ in degree zero, its top image is
$A(R)=0$ (the empty directional sum), and its top dual kernel is $R^*$.
For $r\geq1$, the image in top degree is the sum of the directional code spaces,
hence $A(R)$. Dualizing proves exactness below the top degree;
at the top, Lemma~\ref{cert:product-kernel} identifies the orthogonal
complement of $A(R)$ in $(R\otimes F^\Omega)^*$, under the natural
evaluation pairing, with the asserted tensor of coordinate kernels.
\end{proof}
The transposed equations apply $G_i^{\mathsf T}$ in an active
factor and sum over the possible factors to remove.
\label{ten:dual-arrays}

\subsection{Controlled Primitives}
\begin{proposition}\label{ten:filling-recursion}
Let $r\geq1$, and in this proposition write
$n_-:=\min_{i\in[r]}n_i$ and $n_+:=\max_{i\in[r]}n_i$.
Suppose all nonempty subtuples of $(\im G_i)_{i\in[r]}$ are
universally $\rho$-product-expanding and $n_+/n_-\leq K$.
For every field extension $E/F$, every finite-dimensional $E$-space
$R$, every $0\leq u<r$, and $b\in\im\delta^u$, there is $y\in T^u$
such that
\[
 \delta y=b,\qquad n_-\kappa_{r,u}|y|_b\leq|b|_b.
\]
\end{proposition}
\begin{proof}
All spaces and encoder matrices in this proof are extended to $E$,
and all tensor products are over $E$. Induct on $r\geq1$,
simultaneously for all coefficient extensions and passive spaces. If $b=0$ take $y=0$. In degree zero, exactness gives
a unique nonzero primitive with one block. Applying the singleton
case of Lemma~\ref{cert:collective-expansion} in any factor, with
all other factors passive, gives $|b|_b\geq\rho n_-$.
For $u=r-1$, collective expansion of the full tuple writes
$b=\sum_i b_i$ with $\rho\sum_i n_i\ell_i(b_i)\leq|b|_b$.
The unique messages of these directional code lines give $y$ with
$|y|_b=\sum_i\ell_i(b_i)$, proving the claim in this degree.

For $1\leq u\leq r-2$, put $a=\kappa_{r-1,u}$ and
$c=\kappa_{r-1,u-1}$. Write the complex as
\[
 T^j=T_{r-1}^j(R\otimes U_r)\oplus
                    \bigoplus_{\omega\in\Omega_r}T_{r-1}^{j-1}(R),
 \qquad \delta(z_0,z_1)=(\delta'z_0,G_rz_0+\delta'z_1).
\]
For $b=(b_0,b_1)$ choose an old primitive $y^{\rm old}$.
The induction hypothesis gives $y_0$ with
\[
 \delta'y_0=b_0,\qquad |y_0|_b\leq |b_0|_b/(n_-a).
\]
Set $b'=b-\delta(y_0,0)$. Then
\[
 b'_0=0,\qquad |b'|_b\leq|b_1|_b+n_+|y_0|_b
                         \leq|b_1|_b+(K/a)|b_0|_b.
\]
Each active slice of $b'$ is in $\im\delta'$, as follows.
Let $y'=y^{\rm old}-(y_0,0)$. Its inactive component is closed,
and $u<r-1$, so exactness gives $y'_0=\delta'z_0$. Subtracting
$\delta(z_0,0)$ from $y'$ produces a primitive of $b'$ with
inactive component zero. Thus $b'_\omega=\delta'\widetilde y_{1,\omega}$.
Apply induction in degree $u-1$ separately in every active slice;
this yields $y_1$ with $\delta'y_1=b'_1$ and
$|y_1|_b\leq|b'|_b/(n_-c)$. Consequently
\[
 |(y_0,y_1)|_b\leq\frac1{n_-}
 \left(\frac{K+c}{ac}|b_0|_b+\frac1c|b_1|_b\right)
 \leq\frac{K+1}{n_-ac}|b|_b.
\]
Since $a,c\leq1$, the last inequality is valid, and its constant
is $1/(n_-\kappa_{r,u})$. This proves the induction.
\end{proof}

\begin{proof}[Proof of Theorem~\ref{geo:general-filling-table}]
At $v\in X(\ell)$, its simultaneous chart has $r=t-\ell$ factors,
minimum length at least $n_-$, and ratio at most $K$. In degree
$u=k-\ell<r$, Proposition~\ref{ten:filling-recursion} gives a
primitive $p$ of $\delta_vz$ with
$|p|_b\leq|\delta_vz|_b/(n_-\kappa_{r,u})$.
By Lemma~\ref{ten:exactness}, $z+p=\delta_vy$ for some local $y$.
Thus the minimum in \eqref{geo:local-filling} is at most $|p|_b$.
The charts preserve blocks and intertwine all incidences, so this
is the claimed inequality on the actual star. The same argument
works over every extension with every passive space.
\end{proof}

\section{Mixing of Cubical Face Walks}\label{geo:section}
\label{geo:geometry-data}
Throughout this section $X$ is an $(n_1,\ldots,n_t)$-regular,
$\lambda$-expanding cubical complex, with $n_+/n_-\leq K$.
Recall $D_I=\prod_{i\in I}n_i$ and
\begin{equation}\label{geo:counts}
 |X(I)|=2^{t-|I|}VD_I,\qquad Z_j=2^{t-j}V e_j(n_1,\ldots,n_t).
\end{equation}
All operators below act on real-valued functions. For a positive
probability measure $\pi$, use
$\langle x,y\rangle_\pi=\sum_v\pi(v)x(v)y(v)$; an inner product
without subscript uses counting measure.
We shall prove the mixing estimate \eqref{geo:mixing-bound} for a
walk obtained by descending to a lower face, moving to a parallel
face, and ascending again.

\subsection{Parallel Components and Incidence Measures}
\label{geo:parallel-components}\label{geo:incidence-data}
Fix $0\leq\ell\leq k<t$. Every $k$-face has $c_{k\ell}$ lower
$\ell$-faces, and an $\ell$-face of type $I$ has $h_I$ upper
$k$-faces, where
\begin{equation}\label{geo:c-h}
 c_{k\ell}=\binom{k}{\ell}2^{k-\ell},\qquad
 h_I=\sum_{\substack{S\supseteq I\\|S|=k}}D_{S\setminus I}.
\end{equation}
The first identity counts faces of a cube; the second follows from
Lemma~\ref{lem:cubical-stars}. The parallel graph $G_{I,j}$ has
components $C_{I,j,c}$, each with two parts of size $VD_I$ and
regular degree $n_j$, by \eqref{eq:parity-face-count}.

Define the incidence measures and averaging maps by
\begin{equation}\label{geo:pi}
 \begin{gathered}
 \pi_k(f)=Z_k^{-1},\qquad \pi_\ell(v)=\frac{h_{\dir(v)}}{c_{k\ell}Z_k},\\
 (Rx)(v)=h_{\dir(v)}^{-1}\sum_{f\geq v,\ \dim f=k}x(f),\qquad
 (Uz)(f)=c_{k\ell}^{-1}\sum_{v\leq f,\ \dim v=\ell}z(v).
 \end{gathered}
\end{equation}
\begin{lemma}\label{geo:adjoint}
The measure $\pi_\ell$ has total mass one, $U=R^*$,
$R1=U1=1$, and $R$ preserves means and is a contraction in the
indicated weighted Euclidean norms.
\end{lemma}
\begin{proof}
Double-counting $(v,f)$ with $v\leq f$ gives
$\sum_vh_{\dir(v)}=c_{k\ell}Z_k$. Also
\[
 \langle z,Rx\rangle_{\pi_\ell}
 =\frac1{c_{k\ell}Z_k}\sum_{v\leq f}z(v)x(f)
 =\langle Uz,x\rangle_{\pi_k}.
\]
The identities on constants follow from the incidence counts, and
then give preservation of means. Finally Jensen's inequality
$(Rx)^2\leq R(x^2)$, summed against $\pi_\ell$, proves contraction.
\end{proof}

\subsection{Transported-Face Walks}
\label{geo:walk-data}
For $|I|=\ell$, put
\[
 M_{I,j}=n_j^{-1}\operatorname{Adj}(G_{I,j}),\qquad
 M_I=\frac1{t-\ell}\sum_{j\notin I}M_{I,j},\qquad
 M=\bigoplus_{|I|=\ell}M_I,
\]
\begin{equation}\label{geo:walk}
 W_{k\ell}=UMR=R^*MR.
\end{equation}
The operator $M$ is self-adjoint for $\pi_\ell$, which is constant
on each type. All three operators in $UMR$ are nonnegative and
preserve constants. Thus $W_{k\ell}$ is a reversible transition
matrix for the uniform measure on $X(k)$.
\label{geo:walk-reversible}
More explicitly, let $\mathcal W_{k\ell}(f,g)$ consist of tuples
$(v,j,q,v')$ in which $v\leq f$ and $v'\leq g$ are opposite
$\ell$-facets of $q\in X(\dir(v)\cup\{j\})$. Then
\begin{equation}\label{eq:walk-paths}
 W_{k\ell}(f,g)=\sum_{(v,j,q,v')\in\mathcal W_{k\ell}(f,g)}
          \frac1{c_{k\ell}(t-\ell)n_j h_{\dir(v)}}.
\end{equation}
Set
\begin{equation}\label{geo:beta}
 \mu_I=\frac{2VD_Ih_I}{c_{k\ell}Z_k},\qquad
 \beta_{k\ell}=\max_{|I|=\ell}\mu_I^{-1}.
\end{equation}
Here $\mu_I$ is the $\pi_\ell$-mass of each component of every
$G_{I,j}$, independently of $j$.

\subsection{The Mixing Estimate}
\begin{lemma}\label{geo:mixing}
For every $A\subseteq X(k)$,
\begin{equation}\label{geo:mixing-bound}
 \langle1_A,W_{k\ell}1_A\rangle
 \leq\lambda|A|+\beta_{k\ell}\frac{|A|^2}{Z_k},
\end{equation}
where
\begin{equation}\label{geo:beta-ratio}
 \beta_{k\ell}\leq\binom t\ell2^{t-\ell-1}K^k.
\end{equation}
\end{lemma}
\begin{proof}
For $z\geq0$ on $X(\ell)$, restrict to a component $C$ of
$G_{I,j}$ and write
$z|_C=a_C1_C+b_C\sigma_C+z_C^0$, with
$z_C^0\perp1_C,\sigma_C$. The two distinguished eigenvalues are
$1,-1$, while the remaining norm is at most $\lambda$.
Since both distinguished vectors have squared norm $\mu_I$,
\[
 \langle z|_C,M_Cz|_C\rangle_{\pi_\ell}
 \leq\mu_Ia_C^2-\mu_Ib_C^2+\lambda\|z_C^0\|_{\pi_\ell}^2
 \leq\lambda\|z|_C\|_{\pi_\ell}^2
       +\mu_I^{-1}\left(\int_Cz\,d\pi_\ell\right)^2.
\]
The second inequality uses $\lambda\geq0$. Nonnegativity of $z$
gives $\sum_C(\int_Cz)^2\leq(\int_{X(I)}z)^2$, and the same
inequality applies when summing over types. Summing over components,
averaging over $j$, and summing over $I$ therefore yields
\[
 \langle z,Mz\rangle_{\pi_\ell}
 \leq\lambda\|z\|_{\pi_\ell}^2+
          \beta_{k\ell}\left(\int z\,d\pi_\ell\right)^2.
\]
Take $z=R1_A$. Lemma~\ref{geo:adjoint} bounds its squared norm
by $|A|/Z_k$ and identifies its mean with $|A|/Z_k$.
Multiplying the last display by $Z_k$ proves \eqref{geo:mixing-bound}.
Finally
\[
 \mu_I=\frac{2}{c_{k\ell}2^{t-k}}
 \frac{\sum_{S\supseteq I,\ |S|=k}D_S}{\sum_{|S|=k}D_S}
 \geq\frac{2\binom{t-\ell}{k-\ell}n_-^k}
 {c_{k\ell}2^{t-k}\binom tk n_+^k}
 \geq\frac1{\binom t\ell2^{t-\ell-1}K^k}.
\]
Use $\binom tk\binom k\ell=\binom t\ell\binom{t-\ell}{k-\ell}$
and take reciprocals.
\end{proof}

\section{Cosystolic Distance and Cocycle Expansion}\label{sec:cosystole}
Let $X$ have the geometry of Section~\ref{geo:section}, and let
$\mathcal F$ be a cosheaf over a field of characteristic two.
For a subspace $U\leq C^j$, put
$\operatorname{dist}_b(x,U)=\min_{u\in U}|x-u|_b$. Write
\begin{equation}\label{hom:block-weight}
 \mu^j=\min_{x\in\ker\delta^j\setminus\im\delta^{j-1}}|x|_b,
 \qquad
 \varepsilon^j=\inf_{x\notin\ker\delta^j}
       \frac{|\delta^jx|_b}{\operatorname{dist}_b(x,\ker\delta^j)}.
\end{equation}
Empty minima and infima are $+\infty$.
The conclusion of this section is a linear lower bound on $\mu^j$
and a positive lower bound on $\varepsilon^j$, under local filling
and a sufficiently small geometric expansion parameter.

\subsection{Local Cominimality and Support Cancellation}
\begin{definition}\label{geo:cominimal}\label{hom:local-cominimality}
A cochain $x\in C^j$ is locally cominimal if
$|x+\delta y|_b\geq|x|_b$ whenever $y\in C^{j-1}$ is supported
on faces containing a single vertex. Put
\[
 m^j(\mathcal F)=\min\{|x|_b:0\ne x\in\ker\delta^j,
                                    \ x\text{ locally cominimal}\}.
\]
\label{geo:cominimal-distance-definition}\label{hom:cominimal-distance}
\end{definition}
For $v\in X(\ell)$ and $\ell\leq k<t$, consider facets
$g\in X(k)$ of $(k+1)$-faces containing $v$, with $v\not\leq g$.
Let $O_v(k)$ contain those with $v\cap g=\varnothing$, and
$N_v(k)$ those with nonempty intersection.
\label{geo:external-face-sets}

\begin{lemma}\label{cube:external-witness}\label{geo:external-characterization}
For each such pair $(v,g)$ there is at most one $(k+1)$-face
containing both. If it exists, then either $v\cap g$ is a facet
of $v$, or $v\cap g=\varnothing$. In the latter case there is
within the witness a unique $(\ell+1)$-face $q$ whose opposite
$\ell$-facets are $v$ and a face $v'\leq g$.
\end{lemma}
\begin{proof}
Two witnesses would have intersection containing $g$ and $v$.
Its dimension is at least $k$, and equality would make it $g$,
contradicting $v\not\leq g$. Thus its dimension is $k+1$ and
the witnesses coincide. In this cube write $g$ as the facet
fixing direction $i$ to $b$. If $v$ is active in $i$, intersecting
with $g$ fixes this one coordinate and gives a facet of $v$.
Otherwise $v$ must have $i$-bit $1-b$, so the intersection is
empty. Replacing that bit by $\star$, or by $b$, respectively,
gives the unique faces $q,v'$ in the second alternative.
In particular $N_v(k)=\varnothing$ when $\ell=0$.
\end{proof}

Suppose positive numbers $\kappa_{r,u}$ satisfy
\eqref{geo:local-filling} for every upward star in degrees below
its top. This hypothesis alone on the coefficients is used until
the final specialization to product expansion.
\label{geo:coefficient-data}

\begin{lemma}\label{geo:local-cancellation}
If $x\in C^k$ is a locally cominimal cocycle with support $A$, then
for $v\in X(\ell)$, $0\leq\ell\leq k<t$,
\begin{equation}\label{geo:local-support}
 n_-\kappa_{t-\ell,k-\ell}|A\cap X_{\geq v}(k)|
 \leq |A\cap O_v(k)|+|A\cap N_v(k)|.
\end{equation}
\end{lemma}
\begin{proof}
A local correction $y$ at $v$ extends by zero to a global cochain.
Its coboundary is supported entirely on faces containing $v$ and
restricts there to $\delta_vy$. Choosing any vertex of $v$ shows
that this correction is allowed by local cominimality. Hence
$x_v$ is minimum weight modulo local coboundaries. Apply
\eqref{geo:local-filling} to obtain
$|\delta_vx_v|_b\geq n_-\kappa_{t-\ell,k-\ell}|x_v|_b$.
For each $(k+1)$-face where this local coboundary is nonzero,
the global equation $\delta x=0$ supplies a nonzero external
facet contribution. Choose one such facet $g\in A$.
Lemma~\ref{cube:external-witness} makes the assignment to $g$
injective, and puts its image in $O_v(k)\cup N_v(k)$.
This proves the opposite inequality and hence the claim.
\end{proof}

\subsection{Neighbor Covering and the Support Recursion}
Let $P_{k\ell}(f,g)=|\mathcal W_{k\ell}(f,g)|$ count the paths in
\eqref{eq:walk-paths}, with all listed edges retained. Put
\begin{equation}\label{geo:path-bound}
 L_{k\ell}=c_{k\ell}(t-\ell)\binom{t-\ell}{k-\ell}n_+^{k-\ell+1}.
\end{equation}
\begin{lemma}\label{geo:path-domination}\label{geo:neighbor-covering}\label{geo:neighbor-bound}
For $f\in X(k)$ and $A\subseteq X(k)$,
\begin{align}
 \sum_{v\leq f,\ \dim v=\ell}|A\cap O_v(k)|
 &\leq(P_{k\ell}1_A)(f)\leq L_{k\ell}(W_{k\ell}1_A)(f),
                         \label{geo:opposite-path}\\
 \sum_{v\leq f,\ \dim v=\ell+1}|A\cap N_v(k)|
 &\leq(k-\ell)\sum_{u\leq f,\ \dim u=\ell}|A\cap X_{\geq u}(k)|
                         \quad(\ell<k).\label{geo:neighbor-cover}
\end{align}
\end{lemma}
\begin{proof}
Each pair $(v,g)$ counted in the first sum gives, by
Lemma~\ref{cube:external-witness}, a path $(v,j,q,v')$ from $f$
to $g$. The path retains $(v,g)$, so this assignment is injective.
Every path in \eqref{eq:walk-paths} has probability at least
$L_{k\ell}^{-1}$, since
$h_I\leq\binom{t-\ell}{k-\ell}n_+^{k-\ell}$.
This proves the first line.
For the second, map $(v,g)$ to $(v\cap g,g)$. Its first entry
is an $\ell$-face $u\leq f,g$. For fixed $(u,g)$, the possible
$v$ lie among the $k-\ell$ immediate cofaces of $u$ inside $f$.
This bounds every fiber and proves the inequality.
\end{proof}

For $0\leq\ell\leq k<t$, define
\begin{equation}\label{geo:Btilde}
 b_{k\ell}=
 \frac{K^{k-\ell+1}\binom k\ell\binom{t-\ell}{k-\ell}
       (t-\ell)2^{k-\ell}(k-\ell)!}
      {\prod_{j=\ell}^k\kappa_{t-j,k-j}}.
\end{equation}
\begin{proposition}\label{geo:support-recursion}
Every locally cominimal cocycle in $C^k$, $k<t$, has support $A$
satisfying
\begin{equation}\label{geo:support-walk}
 |A|\leq\sum_{\ell=0}^k b_{k\ell}
                  \langle1_A,W_{k\ell}1_A\rangle.
\end{equation}
\end{proposition}
\begin{proof}
For fixed $f\in X(k)$, let $H_j,O_j,N_j$ be the sums over
$j$-faces $v\leq f$ of $|A\cap X_{\geq v}(k)|$,
$|A\cap O_v(k)|$, and $|A\cap N_v(k)|$, respectively.
Set $r_j=n_-\kappa_{t-j,k-j}$. Lemma~\ref{geo:local-cancellation}
gives $r_jH_j\leq O_j+N_j$. Also $H_k=1_A(f)$,
$N_0=0$, and \eqref{geo:neighbor-cover} gives
$N_j\leq(k-j+1)H_{j-1}$. Substitution from $j=0$ upwards yields
\[
 1_A(f)\leq\sum_{\ell=0}^k
 \frac{(k-\ell)!}{\prod_{j=\ell}^kr_j}\,O_\ell(f).
\]
Multiply by $1_A(f)$, sum over $f$, and use
\eqref{geo:opposite-path}. Since
$L_{k\ell}/n_-^{k-\ell+1}$ is at most
$K^{k-\ell+1}\binom k\ell2^{k-\ell}(t-\ell)
\binom{t-\ell}{k-\ell}$, the resulting coefficients are bounded
by \eqref{geo:Btilde}.
\end{proof}

\subsection{Global Locally Cominimal Distance}
Set $L_j=\sum_{\ell=0}^j b_{j\ell}$ and
$Q_j=\sum_{\ell=0}^j b_{j\ell}\beta_{j\ell}$.
\begin{theorem}\label{geo:cominimal-distance}
If $j<t$ and $\lambda L_j<1$, then
\begin{equation}\label{geo:distance-general}
 m^j(\mathcal F)\geq\frac{1-\lambda L_j}{Q_j}Z_j.
\end{equation}
\end{theorem}
\begin{proof}
For nonempty support $A$ of a locally cominimal cocycle, combine
Proposition~\ref{geo:support-recursion} and Lemma~\ref{geo:mixing}:
\[
 |A|\leq\lambda L_j|A|+Q_j|A|^2/Z_j.
\]
Here $Q_j>0$ because the term $\ell=j$ is positive. Subtraction
and division prove the bound for every such $A$. If no such
cocycle exists, the minimum is $+\infty$.
\end{proof}

For $t\geq4$, $2\leq k\leq t-2$, $0<\rho\leq1$ and $K\geq1$,
use the canonical $\kappa$ of \eqref{ten:filling-constants} and
write $\widehat b_{j\ell}$ for \eqref{geo:Btilde} with those values.
Set
\begin{equation}\label{geo:abstract-constants}
 \begin{gathered}
 \widehat\beta_{j\ell}=\binom t\ell2^{t-\ell-1}K^j,\qquad
 \widehat L_j=\sum_{\ell=0}^j\widehat b_{j\ell},\qquad
 \widehat Q_j=\sum_{\ell=0}^j\widehat b_{j\ell}\widehat\beta_{j\ell},\\
 \mathcal J_{t,k}=\{k,k+1,t-k,t-k+1\},\qquad
 \mathcal L(t,k,\rho,K)=\max_{j\in\mathcal J_{t,k}}\widehat L_j,\\
 \mathcal Q(t,k,\rho,K)=\max_{j\in\mathcal J_{t,k}}\widehat Q_j,
 \qquad c(t,k,\rho,K,\lambda)=\frac{1-\lambda\mathcal L}{\mathcal Q}.
 \end{gathered}
\end{equation}
These are finite positive constants except that $c>0$ requires
$\lambda\mathcal L<1$. Every degree in $\mathcal J_{t,k}$ is below $t$.
\begin{corollary}\label{geo:threshold}
If a cosheaf $\mathcal G$ satisfies the hypotheses of
Theorem~\ref{geo:general-filling-table} and
$\lambda\mathcal L(t,k,\rho,K)<1$, then
\begin{equation}\label{geo:uniform-distance}
 m^j(\mathcal G)\geq c(t,k,\rho,K,\lambda)Z_j
 \qquad(j\in\mathcal J_{t,k}).
\end{equation}
\end{corollary}
\begin{proof}
Theorem~\ref{geo:general-filling-table} supplies the canonical
local filling inequalities. Equation~\eqref{geo:beta-ratio} gives
$L_j\leq\mathcal L$ and $Q_j\leq\mathcal Q$ in the required
degrees. Apply Theorem~\ref{geo:cominimal-distance}; decreasing
its numerator and increasing its positive denominator gives
\eqref{geo:uniform-distance}.
\end{proof}

\subsection{Cosystolic Distance and Syndrome Repair}
Put $\ell_j=\max_{v\in X(0)}|X_{\geq v}(j)|$.
Regularity implies $1\leq\ell_j\leq\binom tj n_+^j$.
The following local descent argument uses the local-minimality method
of \cite{DLV2024}; we give the syndrome-repair
argument in full in the present upward-map convention.
\begin{lemma}\label{hom:local-descent}\label{hom:local-descent-expansion}
For $0\leq j\leq t$, $\mu^j(\mathcal F)\geq m^j(\mathcal F)$.
For $0\leq j<t$ and a real $M>0$ with $m^{j+1}(\mathcal F)\geq M$,
\[
 \varepsilon^j(\mathcal F)\geq\min\{\ell_j^{-1},M/Z_j\}.
\]
\end{lemma}
\begin{proof}
A minimum-weight representative of a nonzero cohomology class is
locally cominimal, proving the first inequality (and the empty
case is automatic).
For the second, let $b=\delta^jx$. If $b=0$ there is nothing to
prove. If $|b|_b\geq M$, then
$\operatorname{dist}_b(x,\ker\delta^j)\leq Z_j\leq Z_j|b|_b/M$.
Otherwise $b$ is a cocycle of weight below $M$. While nonzero it
cannot be locally cominimal, so there is a correction $y$ supported
above a vertex with $|b+\delta^jy|_b<|b|_b$ and $|y|_b\leq\ell_j$.
Repeat. The integer syndrome weight decreases at each step, so after
at most the original $|b|_b$ steps the syndrome vanishes. The total
correction has weight at most $\ell_j|b|_b$ and moves $x$ into
$\ker\delta^j$. This proves the claimed minimum bound.
\end{proof}
In particular, \eqref{geo:uniform-distance} gives
$\mu^j\geq cZ_j$ and
$\varepsilon^j\geq\min\{\binom{t}{j}^{-1}n_+^{-j},cZ_{j+1}/Z_j\}$
whenever the two consecutive degrees are among those covered.

\section{Systolic distance and cycle expansion}
\label{hom:section}

Let $F$ be a field of characteristic two, let $X=\mcX$ be a
$t$-dimensional $(n_1,\ldots,n_t)$-regular cubical complex, and let
$\mathcal F$ be a cosheaf with compatible base-code encoders as in
Section~\ref{sec:cosheaves}. All stalks have specified ordered bases.
Write
\[
 C_j=C^j(X;\mathcal F),\qquad D_j=\delta^j,\qquad
 B_j=D_{j-1}^{\mathsf T},\qquad Z_j=|X(j)|,\qquad n=\max_i n_i.
\]
The transpose acts on the algebraic duals in their dual bases;
we identify these coordinate spaces with the corresponding $C_j$.
Spaces outside degrees $0,\ldots,t$ and maps into them are zero.
The cochain square-zero identity gives
$B_jB_{j+1}=(D_jD_{j-1})^{\mathsf T}=0$ for every integer $j$.
For $0\leq j\leq t$, define
\begin{align*}
 \mu_j(\mathcal F)
 &=\min\{|x|_b:x\in\ker B_j\setminus\operatorname{im}B_{j+1}\},\\
 \varepsilon_j(\mathcal F)
 &=\inf_{x\notin\ker B_j}
   \frac{|B_jx|_b}{\operatorname{dist}_b(x,\ker B_j)}.
\end{align*}
Empty minima and infima are $+\infty$. In particular, a lower
bound on $\varepsilon_j$ is an inequality for distance to the
entire kernel, including all homology classes.

The anchored-array method of \cite{DLV2024} transfers estimates from the
cohomology of an auxiliary cosheaf $\mathcal K$ to the transpose
complex of $\mathcal F$. We establish the required simultaneous
coordinates for the actual cosheaf maps below.

\begin{theorem}\label{hom:headline}
\label{hom:transfer}\label{hom:expansion-transfer}
The cosheaf $\mathcal F$ determines a cosheaf $\mathcal K$ on $X$
whose compatible base codes are the dual codes
of those of $\mathcal F$. Put
\[
 B=(2tn)^t,\qquad T=2(t^22^{2t}n^{t+1})^t.
\]
If $t\geq2$, then
\[
 \mu_k(\mathcal F)\geq B^{-1}\mu^{t-k}(\mathcal K)
 \qquad(1\leq k<t).
\]
If $2\leq k\leq t$, $M>0$, and $0<\varepsilon\leq1$ satisfy
$\mu^{t-k+1}(\mathcal K)\geq M$ and
$\varepsilon^{t-k}(\mathcal K)\geq\varepsilon$, then
\[
 \varepsilon_k(\mathcal F)
 \geq\min\left\{\frac{\varepsilon}{T},\frac{M}{BZ_k}\right\}.
\]
The assertions hold in arbitrary stalk bases, including
nonorthogonal changes of basis.
\end{theorem}

\subsection{Anchored arrays and the kernel cosheaf}

For integers $0\leq r\leq p\leq t$, set
\begin{equation}\label{hom:double-complex}
 A^{r,p}(f)=\bigoplus_{\substack{u\in X(p)\\f\leq u}}
                 \mathcal F(u)^*\quad(f\in X(r)),\qquad
 A^{r,p}=\bigoplus_{f\in X(r)}A^{r,p}(f).
\end{equation}
The coordinate at an incident pair $(f,u)$ is denoted $a(f)[u]$.
Its anchor weight is
\[
 |a|_a=|\{f\in X(r):a(f)\ne0\}|.
\]
Define the horizontal and vertical maps by
\begin{align}
 (\Delta a)(g)[u]
 &=\sum_{\substack{f\in X(r)\\f<g}}a(f)[u],
 &&g\in X(r+1),\ u\in X(p),\ g\leq u,
 \label{hom:horizontal-map}\\
 (Ta)(f)[v]
 &=\sum_{\substack{u\in X(p)\\v<u}}R_{uv}^{\mathsf T}a(f)[u],
 &&f\in X(r),\ v\in X(p-1),\ f\leq v.
 \label{hom:vertical-map}
\end{align}
Thus $\Delta:A^{r,p}\to A^{r+1,p}$ and
$T:A^{r,p}\to A^{r,p-1}$. Arrays outside the indicated range
are zero. Write $T_f^p$ for the part of $T$ at the anchor $f$.

\begin{lemma}\label{hom:double-identities}
These maps satisfy
\[
 \Delta^2=T^2=0,\qquad \Delta T=T\Delta,\qquad
 |\Delta a|_a\leq tn|a|_a,\qquad |Ta|_a\leq|a|_a.
\]
\end{lemma}
\begin{proof}
A codimension-two interval in a cube has exactly two intermediate
faces. The two contributions to $\Delta^2$ are identical. Those
to $T^2$ are transposes of two compositions $R$ around a diamond,
and both compositions equal the map between its endpoints.
Characteristic two therefore gives both square-zero identities.
For an output pair $g\leq v$, both mixed compositions equal
\[
 \sum_{\substack{f\in X(r)\\f<g}}
 \sum_{\substack{u\in X(p)\\v<u}}
             R_{uv}^{\mathsf T}a(f)[u].
\]
Every term is defined since $f\leq g\leq v\leq u$.
If the output degree is outside the allowed range, both maps
are zero. This also verifies the identities at the endpoints.
An active $r$-anchor has at most $(t-r)n\leq tn$ immediate
cofaces, by regularity and the higher-face star description.
These are the only possible active output anchors of $\Delta$.
The map $T$ does not change anchors, proving both weight bounds.
\end{proof}

For every face $f$, form the actual local top kernel
\begin{equation}\label{hom:canonical-kernel}
 K(f)=\ker\bigl(T_f^t:A^{\dim f,t}(f)\longrightarrow
                           A^{\dim f,t-1}(f)\bigr).
\end{equation}
For $f\leq g$, let $P_{gf}$ retain the entries at top faces
containing $g$:
\[
 (P_{gf}a)[u]=a[u]\qquad(u\in X(t),\ g\leq u).
\]
Choose any basis of $K(f)$, set $k_f=\dim_FK(f)$, and denote
its coordinate isomorphism by $\iota_f:F^{k_f}\to K(f)$.

\begin{lemma}\label{hom:kernel-restrictions}
The assignments
\[
 \mathcal K(f)=F^{k_f},\qquad
 S_{gf}=\iota_g^{-1}P_{gf}\iota_f\quad(f\leq g)
\]
form a cosheaf. Under $\iota^r=\bigoplus_{f\in X(r)}\iota_f$,
its cochain differential is the restriction of $\Delta$ to
$\bigoplus_fK(f)\subseteq A^{r,t}$, and its block weight is
the anchor weight.
\end{lemma}
\begin{proof}
If $a\in K(f)$ and $g\leq v\in X(t-1)$, then every top face
$u>v$ contains $g$. Hence the equation at $v$ after restriction is
\[
 (T_g^tP_{gf}a)[v]
 =\sum_{u>v}R_{uv}^{\mathsf T}a[u]=(T_f^ta)[v]=0.
\]
This proves $P_{gf}K(f)\subseteq K(g)$. At $g\in X(t)$ the
target equations are empty, which gives the same conclusion.
Retention gives $P_{ff}=I$ and $P_{hg}P_{gf}=P_{hf}$ for every
comparable triple, without a restriction on dimension gaps.
Conjugating by the $\iota_f$ proves the cosheaf identities.
For a cochain $x$, the entry of $\iota^{r+1}\delta^rx$ at
$(g,u)$ is
$\sum_{f<g}(\iota_fx(f))[u]$, exactly the formula for
$\Delta\iota^rx$. Each $\iota_f$ is invertible, so a block is
nonzero precisely when its anchored list is nonzero.
\end{proof}

We henceforth use these isomorphisms to regard $C^r(X;\mathcal K)$
as the corresponding subspace of $A^{r,t}$. Changing a basis of
$K(f)$ only conjugates maps within that block and preserves this
identification and all block weights. The vertical maps $T$
continue to use the specified bases of the original cosheaf.

\subsection{Local transpose exactness and dual base codes}

\begin{theorem}
\label{hom:canonical-dual}
For every face $f$ and every $\dim f\leq p<t$,
\begin{equation}\label{hom:local-exactness}
 \ker T_f^p=\operatorname{im}T_f^{p+1}.
\end{equation}
Moreover, suppose a compatible chart at $f$, based at a vertex
$x\leq f$, uses direction sets $\Omega_i=E_i(x)$ and injective
encoders $E_i:F^{m_i}\to F^{\Omega_i}$ for
$i\in J=\codir(f)$. Choose basis matrices
\[
 Q_i:F^{n_i-m_i}\hookrightarrow F^{\Omega_i},\qquad
 \operatorname{im}Q_i=\ker E_i^{\mathsf T}.
\]
Then $\mathcal K$ has compatible charts throughout the same
star, with the same geometric indexing $\theta_{f,x}$ and
encoders $Q_i$. In particular these encoder images are
$(\operatorname{im}E_i)^\perp$ in the actual edge coordinates.
\end{theorem}
\begin{proof}
Fix the entire family of chart maps $\Psi_{f,g}$, $g\geq f$,
supplied by compatibility at $f$. On dual vectors use
$\Theta_g=(\Psi_{f,g}^{-1})^{\mathsf T}$. For every incidence
$g<h$ in this star, one has the exact matrix identity
\begin{equation}\label{hom:inverse-transpose-chart}
 \Theta_gR_{hg}^{\mathsf T}\Theta_h^{-1}
 =\bigl(\Psi_{f,h}R_{hg}\Psi_{f,g}^{-1}\bigr)^{\mathsf T}.
\end{equation}
Thus the local transpose matrix is the transpose of the tensor
evaluation matrix; using $\Psi_{f,g}$ itself on dual vectors
would not give this identity in general.

Write $r=\dim f$ and $d=t-r$. The geometric indexing of cofaces
by added subsets $I\subseteq J$ and labels in $\prod_{i\in I}\Omega_i$
identifies $A^{r,r+q}(f)$ with
\[
 V_q=\bigoplus_{\substack{I\subseteq J\\|I|=q}}
       \bigotimes_{i\in J}
       \begin{cases}F^{\Omega_i},&i\in I,\\F^{m_i},&i\notin I.
       \end{cases}
\]
Under \eqref{hom:inverse-transpose-chart}, $T_f^{r+q}$ applies
$E_i^{\mathsf T}$ in one active factor and sums over removed
directions. This is exactly the transposed tensor complex of
Lemma~\ref{ten:dual-exactness}, instantiated with encoder $E_i$
and parity-check matrix $Q_i^{\mathsf T}$ in direction $i$.
The latter has independent rows and kernel
$\ker Q_i^{\mathsf T}=(\ker E_i^{\mathsf T})^\perp
=\operatorname{im}E_i$.
The lemma gives exactness for $q<d$ and top kernel
$\operatorname{im}\bigotimes_{i\in J}Q_i$. This proves
\eqref{hom:local-exactness}, including $p=r$. When $r=t$ there
is no such $p$.

For the compatibility assertion, fix
$g=\theta_{f,x}(I,\omega_I)\geq f$ and write $J_g=J\setminus I$.
The restrictions of the geometric indexing and all chart maps at
$f$ give the entire star above $g$, still based at $x$.
Apply the same top-kernel description to the factors in $J_g$.
Each $a\in K(g)$ has a unique parameter
$w_g\in\bigotimes_{i\in J_g}F^{n_i-m_i}$ satisfying, at a top
coface $u$ with remaining labels $\omega_{J_g}$,
\begin{equation}\label{hom:kernel-parameterization}
 \Theta_u a[u]
 =\sum_{b\in\prod_{i\in J_g}[n_i-m_i]}
       w_g(b)\prod_{i\in J_g}Q_i(\omega_i,b_i).
\end{equation}
Existence follows from the top-kernel formula. Uniqueness follows
because each $Q_i$ has a left inverse, so their tensor product
does too. For $J_g=\varnothing$ the formula is the identity on
$F$. A zero-dimensional factor gives a zero parameter space and
the same injectivity statement.

Suppose $h$ is obtained from $g$ by adding direction $j\in J_g$
and edge label $\omega_j$. Retention by $P_{hg}$ fixes that label
in \eqref{hom:kernel-parameterization}. Summing first over $b_j$
gives the unique parameters at $h$:
\[
 w_h(b_{J_g\setminus\{j\}})
 =\sum_{b_j\in[n_j-m_j]}Q_j(\omega_j,b_j)w_g(b).
\]
This is precisely row evaluation by $Q_j$. All maps $\Theta_u$
and matrices $Q_j$ were fixed once for the star at $f$, so these
identities hold simultaneously for every incidence in that star.
Composing $a\mapsto w_g$ with $\iota_g$ gives the required charts
on the actual coordinate stalks $\mathcal K(g)$.
\end{proof}

\subsection{Filling small transpose-kernel vectors}

We first record a contraction that acts only on the anchors.

\begin{lemma}\label{hom:cube-contraction}
Let $p\geq0$ and let $W$ be a finite-dimensional $F$-space.
On the cochains of the ordinary $p$-cube with constant stalk $W$
and identity incidence maps, there are maps $H$ lowering degree
and $P$ preserving degree such that
\[
 dH+Hd=I-P,\qquad |Ha|_b\leq2^p|a|_b.
\]
The map $P$ is zero in positive degrees and, in degree zero,
replaces every vertex value by the value at the all-zero vertex.
Consequently a closed cochain in positive degree is $dHa$,
and a closed degree-zero cochain is constant.
\end{lemma}
\begin{proof}
For one coordinate, write $e_0,e_1,e_*$ for its three face
symbols. Define
\[
 \begin{array}{c|ccc}
       &e_0&e_1&e_*\\\hline
 d_i   &e_*&e_*&0\\
 h_i   &0&0&e_1\\
 P_i   &e_0+e_1&0&0
 \end{array}
\]
and let these maps act identically in the other coordinates and
on $W$. Evaluation on these three symbols gives
$d_ih_i+h_id_i=I-P_i$ and $d_iP_i=P_id_i=0$.
Operators in distinct coordinates commute. Hence, for
\[
 H=\sum_{i=1}^pP_1\cdots P_{i-1}h_i,\qquad
 P=P_1\cdots P_p,
\]
the terms with different indices in $dH+Hd$ either vanish or
cancel in pairs, and the remaining sum telescopes:
\[
 dH+Hd=\sum_{i=1}^pP_1\cdots P_{i-1}(I-P_i)=I-P.
\]
Every positive-degree word contains a star and is killed by $P$.
At degree zero only the all-zero word survives, and its image is
the sum of all vertex words. For a single input word, the $i$th
summand of $H$ can be nonzero only if its first nonzero symbol
is a star in position $i$. There is at most one such $i$, and
the preceding $P$'s produce at most $2^{i-1}\leq2^p$ words.
Subadditivity of support proves the weight bound. For $p=0$,
$H=0$ and $P=I$. Applying the identity to a closed cochain
proves the final statements, including this case.
\end{proof}

\begin{lemma}
\label{hom:horizontal-filling}
Put $H_0=2^{2t}n^t$ and $G_0=2^tn^t$.
For $1\leq r\leq p\leq t$ and $a\in A^{r,p}$ with
$\Delta a=0$, there is $b\in A^{r-1,p}$ with
$\Delta b=a$ and $|b|_a\leq H_0|a|_a$.
For $a\in A^{0,p}$ with $\Delta a=0$, there is a unique
$z\in C_p$ such that $a(v)[u]=z(u)$ for all vertices $v\leq u$;
it satisfies $|z|_b\leq G_0|a|_a$.
\end{lemma}
\begin{proof}
Regroup the entries of $A^{r,p}$ by upper face:
\[
 A^{r,p}=\bigoplus_{u\in X(p)}
              \bigoplus_{f\in X_{\leq u}(r)}\mathcal F(u)^*.
\]
At fixed $u$, the horizontal map is the ordinary cube cochain
map with passive space $\mathcal F(u)^*$. Let $|a|_{\rm pair}$
count nonzero pairs $(f,u)$. Regularity gives
\[
 |a|_{\rm pair}\leq
   \binom{t-r}{p-r}n^{p-r}|a|_a\leq2^tn^t|a|_a.
\]
For $r>0$, apply Lemma~\ref{hom:cube-contraction} to the closed
array at each $u$, with cube dimension $p$. Reassemble its
primitives. Then $\Delta b=a$ and
$|b|_a\leq|b|_{\rm pair}\leq2^p|a|_{\rm pair}\leq H_0|a|_a$.
For $r=0$, the same lemma says that each upper face has a common
value at all its vertices. This defines $z(u)$ uniquely. Every
nonzero $z(u)$ contributes a nonzero pair, whence
$|z|_b\leq|a|_{\rm pair}\leq G_0|a|_a$.
\end{proof}

By \eqref{hom:local-exactness}, for every $f\in X(r)$ and
$r\leq p<t$ choose a linear right inverse
\[
 S_f^p:\ker T_f^p\longrightarrow A^{r,p+1}(f),\qquad
 T_f^{p+1}S_f^p=I.
\]
Such a map is obtained by lifting a basis of the kernel. It
sends zero to zero and therefore never creates a new active
anchor when applied separately at each anchor.

\begin{proposition}[Small-cycle filling]\label{hom:small-cycle-filling}
Suppose $t\geq2$, $1\leq k<t$, and $M>0$ satisfy
$\mu^{t-k}(\mathcal K)\geq M$. If $x\in C_k$ obeys
\[
 B_kx=0,\qquad |x|_b<M/(2tn)^t,
\]
then $x\in\operatorname{im}B_{k+1}$. If also
$\varepsilon^{t-k-1}(\mathcal K)\geq\varepsilon$ for some
$0<\varepsilon\leq1$, a primitive can be chosen with
\[
 B_{k+1}z=x,\qquad
 |z|_b\leq\frac{2(t^22^{2t}n^{t+1})^t}{\varepsilon}|x|_b.
\]
\end{proposition}
\begin{proof}
The zero vector has the zero primitive. Suppose $x\ne0$, and
put $L=t-k$, $d_0=tn$, and $B_0=2^t|x|_b$.
Copy $x$ to the vertices of its supporting faces:
\[
 x^{(0)}(v)[u]=x(u)\qquad(v\leq u\in X(k)).
\]
This lies in $A^{0,k}$ and has anchor weight at most
$2^k|x|_b\leq B_0$. Equal values at the two endpoints of each
edge give $\Delta x^{(0)}=0$. At $(v,w)$ with
$v\leq w\in X(k-1)$, the vertical equation is
$(Tx^{(0)})(v)[w]=(B_kx)(w)=0$.

Construct successively, for $0\leq r<L$,
\[
 z^{(r)}(f)=S_f^{k+r}(x^{(r)}(f)),\qquad
 x^{(r+1)}=\Delta z^{(r)},
\]
where $x^{(r)}\in A^{r,k+r}$ and
$z^{(r)}\in A^{r,k+r+1}$. Inductively the argument of $S_f^{k+r}$
is in its kernel domain: the base case was just proved, and
\[
 Tz^{(r)}=x^{(r)},\quad
 Tx^{(r+1)}=\Delta Tz^{(r)}=\Delta x^{(r)}=0,\quad
 \Delta x^{(r+1)}=0.
\]
Here $k+r<t$ ensures that the right inverse exists. Since it
does not increase anchor support, Lemma~\ref{hom:double-identities}
also yields
\begin{equation}\label{hom:ladder-growth}
 |z^{(r)}|_a\leq|x^{(r)}|_a\leq d_0^rB_0,
 \qquad |x^{(r+1)}|_a\leq d_0^{r+1}B_0.
\end{equation}
The last array $x^{(L)}\in A^{L,t}$ is vertically closed and
hence belongs to $C^L(X;\mathcal K)$. It is horizontally
closed and has weight
\[
 |x^{(L)}|_a\leq2^t(tn)^L|x|_b\leq(2tn)^t|x|_b<M.
\]
The assumed cohomological distance therefore supplies
$u\in C^{L-1}(X;\mathcal K)$ with $\Delta u=x^{(L)}$ and
$Tu=0$. For the quantitative conclusion, choose a minimum-weight
primitive. All primitives form one coset of $\ker\Delta$, so
the assumed expansion in degree $L-1=t-k-1$ gives
\begin{equation}\label{hom:ladder-top-primitive}
 |u|_a\leq\varepsilon^{-1}|x^{(L)}|_a
            \leq\varepsilon^{-1}d_0^LB_0.
\end{equation}
If $x^{(L)}=0$, take $u=0$.

Replace $z^{(L-1)}$ by $z^{(L-1)}+u$. Its horizontal
differential is now zero, and its vertical differential remains
$x^{(L-1)}$. Descend through $r=L-1,\ldots,1$. If the corrected
$z^{(r)}$ satisfies $\Delta z^{(r)}=0$ and $Tz^{(r)}=x^{(r)}$,
Lemma~\ref{hom:horizontal-filling}, with upper degree $k+r+1$,
gives $v^{(r-1)}\in A^{r-1,k+r+1}$ such that
\[
 \Delta v^{(r-1)}=z^{(r)},\qquad
 |v^{(r-1)}|_a\leq H_0|z^{(r)}|_a.
\]
Replace $z^{(r-1)}$ by $z^{(r-1)}+Tv^{(r-1)}$. Then
\[
 \Delta(z^{(r-1)}+Tv^{(r-1)})
 =x^{(r)}+Tz^{(r)}=0,
\]
and its vertical differential remains $x^{(r-1)}$ because $T^2=0$.
This completes the descending induction. If $L=1$, this loop
is empty and the top correction already has horizontal degree zero.

The final $z^{(0)}\in A^{0,k+1}$ is horizontally closed.
Agreement in Lemma~\ref{hom:horizontal-filling} supplies
$z\in C_{k+1}$ with $z^{(0)}(v)[u]=z(u)$. At a $k$-face $w$,
choose any vertex $v\leq w$ and use $Tz^{(0)}=x^{(0)}$:
\[
 (B_{k+1}z)(w)=\sum_{u>w}R_{uw}^{\mathsf T}z(u)=x(w).
\]
This proves the qualitative assertion without using
\eqref{hom:ladder-top-primitive}.

For the quantitative assertion let $q_r$ bound the corrected
weight $|z^{(r)}|_a$. The construction gives
\[
 q_{L-1}\leq(d_0^{L-1}+\varepsilon^{-1}d_0^L)B_0,
 \qquad q_{r-1}\leq d_0^{r-1}B_0+H_0q_r,
 \qquad |z|_b\leq G_0q_0.
\]
Iterating these inequalities, and writing
$a_0=H_0d_0=t2^{2t}n^{t+1}$, yields
\begin{align*}
 |z|_b
 &\leq G_0B_0\left(
      \sum_{j=0}^{L-1}a_0^j+
             \varepsilon^{-1}H_0^{L-1}d_0^L\right)\\
 &\leq H_0|x|_b(2+\varepsilon^{-1}d_0)a_0^{t-2}
 \leq2\varepsilon^{-1}a_0^{t-1}|x|_b
 \leq2\varepsilon^{-1}(t^22^{2t}n^{t+1})^t|x|_b.
\end{align*}
For these inequalities, $1\leq L\leq t-1$, $a_0\geq2$,
$G_0B_0=H_0|x|_b$, $d_0\geq2$, and $\varepsilon\leq1$.
In particular, the geometric sum is at most $2a_0^{L-1}$ and
$2+\varepsilon^{-1}d_0\leq2\varepsilon^{-1}d_0$.
\end{proof}

\subsection{Systolic distance and cycle expansion}

\begin{proof}[Proof of Theorem~\ref{hom:headline}]
The construction and compatible-chart assertions are
Lemma~\ref{hom:kernel-restrictions} and
Theorem~\ref{hom:canonical-dual}. If
$\mu^{t-k}(\mathcal K)=M<\infty$, Proposition~\ref{hom:small-cycle-filling}
shows that every cycle of weight less than $M/B$ is a boundary.
Hence $\mu_k(\mathcal F)\geq M/B$. If that cohomological distance
is infinite, apply the proposition separately to each cycle $x$
with any finite $M>B|x|_b$. Every cycle is then a boundary,
and both sides of the asserted distance inequality are infinite.
This proves the distance transfer.

For expansion, let $x\in C_k$ and $b=B_kx$. If $b=0$ there
is nothing to prove. If $|b|_b\geq M/B$, then
\[
 \operatorname{dist}_b(x,\ker B_k)\leq Z_k
                       \leq (BZ_k/M)|b|_b.
\]
Otherwise $B_{k-1}b=0$ and $0<|b|_b<M/B$. Apply
Proposition~\ref{hom:small-cycle-filling} to $b$ in degree
$k-1$. Its hypotheses are precisely
$1\leq k-1<t$, $\mu^{t-k+1}(\mathcal K)\geq M$, and
$\varepsilon^{t-k}(\mathcal K)\geq\varepsilon$.
It supplies $z\in C_k$ with $B_kz=b$ and
$|z|_b\leq(T/\varepsilon)|b|_b$. Since $x-z\in\ker B_k$,
this bounds the distance to the kernel by the same quantity.
Taking the weaker of the two bounds proves the assertion.

The inverse-transpose identity
\eqref{hom:inverse-transpose-chart} applies to every invertible
stalk chart; none of these arguments uses an orthogonal basis.
\end{proof}

\begin{corollary}
\label{hom:middle-constants}
Let $t\geq4$, $2\leq k\leq t-2$, and $c>0$. Suppose
\[
 m^k(\mathcal F)\geq cZ_k,\quad
 m^{k+1}(\mathcal F)\geq cZ_{k+1},\quad
 m^{t-k}(\mathcal K)\geq cZ_{t-k},\quad
 m^{t-k+1}(\mathcal K)\geq cZ_{t-k+1}.
\]
With $B,T$ as in Theorem~\ref{hom:headline}, put
\begin{align*}
 u_c&=\min\left\{\frac1{\binom tk n^k},\frac{cZ_{k+1}}{Z_k}\right\},\\
 u_d&=\min\left\{\frac1{\binom t{t-k}n^{t-k}},
                         \frac{cZ_{t-k+1}}{Z_{t-k}}\right\},\\
 v_h&=\min\left\{u_d/T,\frac{cZ_{t-k+1}}{BZ_k}\right\}.
\end{align*}
Then all three constants are positive and
\[
 \mu^k(\mathcal F)\geq cZ_k,\qquad
 \mu_k(\mathcal F)\geq cZ_{t-k}/B,\qquad
 \varepsilon^k(\mathcal F)\geq u_c,\qquad
 \varepsilon_k(\mathcal F)\geq v_h.
\]
\end{corollary}
\begin{proof}
Regularity gives at most $\binom tj n^j$ degree-$j$ faces above
any vertex. Apply Lemma~\ref{hom:local-descent} to
$(\mathcal F,k)$, $(\mathcal K,t-k)$, and
$(\mathcal K,t-k+1)$. This gives the corresponding three
cohomological distance bounds. Apply
Lemma~\ref{hom:local-descent-expansion} to $\mathcal F$ in
degree $k$ and $\mathcal K$ in degree $t-k$, respectively,
using syndrome thresholds $cZ_{k+1}$ and $cZ_{t-k+1}$.
The vertex-incidence bound gives $\varepsilon^k(\mathcal F)\geq u_c$
and $\varepsilon^{t-k}(\mathcal K)\geq u_d$. These constants
are positive and at most one. Theorem~\ref{hom:headline}, with
$M=cZ_{t-k+1}$ and $\varepsilon=u_d$ for its expansion part,
now gives the two homological conclusions. All relevant degrees
are in range since $2\leq k,t-k\leq t-2$.
\end{proof}

\section{The abstract parameter theorem}
\label{comb:section}

\subsection{Binary conversion of block estimates}

For an ordered binary basis $\beta$ of a finite field $F$, write
$R_\beta(A)$ for the binary matrix of an $F$-linear map $A$.
The conversion uses $\beta$ on cochains and its trace-dual basis
on dual vectors. In particular, it does not require a self-dual
basis of $F$ or orthogonal bases of the stalks.

\begin{proposition}
\label{bin:restriction-and-soundness}
Let $F=\mathbb F_{2^s}$, $s\geq1$, and let
\[
 C^1\xrightarrow{\ B\ } C^2\xrightarrow{\ A\ }C^3,
 \qquad AB=0,
\]
be a complex of finite-dimensional based $F$-spaces. Each $C^i$
is a finite ordered direct sum of blocks, with its basis obtained
by concatenating their ordered bases. Its dual has the
corresponding dual block bases. Put
\[
 D_i=\dim_F C^i,\qquad
 M_2=\max\{\dim_F W:W\text{ is a block of }C^2\},
\]
and suppose $D_2>0$ and $D_1+D_3>0$.
Let $\mu_c,\mu_h,\varepsilon_c,\varepsilon_h>0$ satisfy
\begin{gather*}
 \min_{x\in\ker A\setminus\operatorname{im}B}|x|_b\geq\mu_c,
 \qquad
 \min_{y\in\ker B^{\mathsf T}\setminus\operatorname{im}A^{\mathsf T}}
                                  |y|_b\geq\mu_h,\\
 |Ax|_b\geq\varepsilon_c\operatorname{dist}_b(x,\ker A)
                       \quad(x\in C^2),\\
 |B^{\mathsf T}y|_b\geq\varepsilon_h
       \operatorname{dist}_b(y,\ker B^{\mathsf T})
                       \quad(y\in(C^2)^*).
\end{gather*}
For any ordered binary basis $\beta$ of $F$, the matrices
\[
 H_X=R_\beta(A),\qquad H_Z=R_\beta(B)^{\mathsf T}
\]
define a binary CSS code with
\begin{gather*}
 N=sD_2,\qquad
 K=s(D_2-\operatorname{rank}_F A-\operatorname{rank}_F B),\qquad
 d\geq\min\{\mu_c,\mu_h\},\\
 \mathsf H\succeq
 \frac{D_2\min\{\varepsilon_c,\varepsilon_h\}}
      {sM_2(D_1+D_3)}\frac{\mathsf D_{\mathscr Q}}N.
\end{gather*}
The tester averages all $s(D_1+D_3)$ listed rows, including
dependent and zero rows. If an integer $w_F\geq1$ bounds the
number of nonzero entries in every row and column of $A$ and $B$,
every binary check has weight
at most $sw_F$, and every qubit is incident to at most $2sw_F$
listed checks.
\end{proposition}

The proof is given in Appendix~\ref{app:css}. The losses visible
here have different origins: a nonzero block contributes at least
one binary coordinate to distance, whereas a correction in a
middle block may use as many as $sM_2$ binary coordinates.
The factor $D_1+D_3$ is the size of the specified check list
before restriction of scalars.

\subsection{Proof of the abstract parameter theorem}

Fix integers $t\geq4$, $2\leq k\leq t-2$, and $s\geq1$, a field
$F$ of cardinality $2^s$, and integers
\[
 n_i\geq2,\qquad 0<m_{i,b}<n_i
           \quad(i\in[t],\ b\in\{0,1\}).
\]
Let $n=\max_i n_i$, and fix $K\geq\max_i n_i/\min_i n_i$,
$0<\rho\leq1$, and $0\leq\lambda<1$. For the rate quantities
of Section~\ref{sec:rate}, use
\[
 \mu_{i,b}=m_{i,b}/n_i,\quad
 S_i=\mu_{i,0}+\mu_{i,1},\quad
 p_i=\min_b\mu_{i,b},\quad q_i=1-\max_b\mu_{i,b},\quad
 \Gamma=\Gamma_k.
\]
The hypotheses on these fixed parameters are
\begin{equation}\label{comb:numerical-data}
 \Gamma>0,\qquad \lambda\mathcal L(t,k,\rho,K)<1,
 \qquad
 c=\frac{1-\lambda\mathcal L(t,k,\rho,K)}
          {\mathcal Q(t,k,\rho,K)}>0,
\end{equation}
where $\mathcal L,\mathcal Q$ are the uniform constants of
Corollary~\ref{geo:threshold}.
Define the face-count ratios and output constants by
\begin{gather}
 r_j=2^{k-j}\frac{e_j(n_1,\ldots,n_t)}{e_k(n_1,\ldots,n_t)}
       \quad(0\leq j\leq t),\qquad
 B=(2tn)^t,\qquad T=2(t^22^{2t}n^{t+1})^t,
 \label{comb:output-constants}\\
 u_c=\min\left\{\frac1{\binom tk n^k},cr_{k+1}\right\},\qquad
 u_d=\min\left\{\frac1{\binom t{t-k}n^{t-k}},
                        c\frac{r_{t-k+1}}{r_{t-k}}\right\},
 \nonumber\\
 v=\min\{u_c,u_d/T,cr_{t-k+1}/B\},\qquad
 \xi=\min\{1,r_{t-k}/B\},\qquad
 w=\max\{2(k+1),t-k+1\}n^{t-k+1}.
 \nonumber
\end{gather}
Here $e_j$ is the elementary symmetric polynomial, with
$e_0=1$ and $e_j=0$ outside $0\leq j\leq t$.

\begin{theorem}\label{comb:main}
Fix the preceding parameters satisfying \eqref{comb:numerical-data}.
For every $\nu\geq1$, let $(X_\nu,\mathcal F_\nu)$ consist of
a $t$-dimensional $(n_1,\ldots,n_t)$-regular cubical complex
and an $F$-cosheaf with compatible base-code encoders of endpoint
dimensions $(m_{i,b})$. Suppose:
\begin{enumerate}
\item The parallel-face graphs satisfy the spectral hypothesis
of Section~\ref{sec:cubical-complexes} with bound $\lambda$.
\item For each face $f$, each nonempty subset $I\subseteq\codir(f)$,
and each field extension $E/F$, both tuples
\[
 (E\otimes_F C_{f,i})_{i\in I},\qquad
 (E\otimes_F C_{f,i}^{\perp})_{i\in I},\qquad
 C_{f,i}=\operatorname{im}E_{f,i},
\]
have weighted product expansion at least $\rho$. The dual code
is taken in the actual edge coordinates of the compatible chart.
\item The vertex-parity masses $V_\nu$ tend to infinity.
\end{enumerate}
Choose a binary basis $\beta$ of $F$, orders of the face sets,
and ordered bases of the individual stalks. Define
\[
 H_{X,\nu}=R_\beta(\delta_\nu^k),\qquad
 H_{Z,\nu}=R_\beta(\delta_\nu^{k-1})^{\mathsf T}.
\]
Then the associated CSS codes have parameters $[[N_\nu,K_\nu,d_\nu]]$
and listed testers $\mathsf H_\nu$ satisfying
\begin{gather*}
 N_\nu=sV_\nu\left(\prod_i n_i\right)e_{t-k}(S)\longrightarrow\infty,
 \qquad
 \frac{K_\nu}{N_\nu}\geq\frac{\Gamma}{e_{t-k}(S)}>0,\\
 \frac{d_\nu}{N_\nu}\geq\frac{c\xi}{sn^{t-k}}>0,\\
 \mathsf H_\nu\succeq
 \frac{v e_{t-k}(S)}
 {sn^{t-k}(e_{t-k+1}(S)+e_{t-k-1}(S))}
               \frac{\mathsf D_{\mathscr Q_\nu}}{N_\nu}.
\end{gather*}
Every listed check has weight at most $sw$, and every qubit is
incident to at most $2sw$ listed checks.
\end{theorem}
\begin{proof}
Fix $\nu$, abbreviate $X=X_\nu$ and $\mathcal F=\mathcal F_\nu$,
and put $Z_j=|X(j)|$, $M_j=\dim_FC^j(X;\mathcal F)$, and
$A_0=V_\nu\prod_i n_i$. All these are finite and their required
denominators are positive: regularity gives positive face counts,
and all endpoint dimensions are strictly between zero and $n_i$.

Construct the actual kernel cosheaf $\mathcal K$ by
\eqref{hom:canonical-kernel}. Theorem~\ref{hom:canonical-dual}
gives local transpose exactness and compatible encoders whose
images are $C_{f,i}^{\perp}$. Thus the second hypothesis supplies
the product-expansion premises of the local filling theorem for
both $\mathcal F$ and $\mathcal K$, with the same $\rho$ and
degree-ratio bound $K$. Their underlying geometry and spectral
bound are identical. Corollary~\ref{geo:threshold}, applied to
these two cosheaves and the degrees
$\{k,k+1,t-k,t-k+1\}$, yields
\[
 m^j(\mathcal F),\ m^j(\mathcal K)\geq cZ_j
 \qquad(j\in\{k,k+1,t-k,t-k+1\}).
\]
All these degrees are below $t$ by the assumed middle-degree
range. No choice of a local length is made in this application:
$\rho,K$ and the fixed lengths already satisfy its hypotheses.

The face counts of Section~\ref{sec:cubical-complexes} give
\[
 Z_j=2^{t-j}V_\nu e_j(n_1,\ldots,n_t),\qquad Z_j/Z_k=r_j.
\]
Apply Corollary~\ref{hom:middle-constants} with these counts and
$n=\max_i n_i$. Substitution of the ratios gives
\begin{gather*}
 \mu^k(\mathcal F)\geq cZ_k,\qquad
 \mu_k(\mathcal F)\geq cZ_{t-k}/B,\\
 \varepsilon^k(\mathcal F)\geq u_c\geq v,\qquad
 \varepsilon_k(\mathcal F)\geq\min\{u_d/T,cr_{t-k+1}/B\}\geq v.
\end{gather*}
Both distances are therefore at least $c\xi Z_k$.

The dimension and rate calculation in Proposition~\ref{bin:euler-rate}
gives
\[
 M_j=A_0e_{t-j}(S),\qquad
 h_k:=M_k-\operatorname{rank}\delta^{k-1}
                -\operatorname{rank}\delta^k\geq A_0\Gamma.
\]
Lemma~\ref{bin:locality} bounds the rows and columns of the two
adjacent differentials by $w$ and every middle block by
$n^{t-k}$. In particular $M_k\leq n^{t-k}Z_k$.
Apply Proposition~\ref{bin:restriction-and-soundness} with
\[
 (C^1,C^2,C^3)=(C^{k-1}(X;\mathcal F),C^k(X;\mathcal F),
                                      C^{k+1}(X;\mathcal F)),
 \quad B=\delta^{k-1},\quad A=\delta^k,
\]
and $\mu_c=\mu_h=c\xi Z_k$, $\varepsilon_c=\varepsilon_h=v$.
The complex identity gives $AB=0$; the preceding paragraphs
verify its four metric hypotheses and its dimension conditions.
It follows that $N_\nu=sM_k$, $K_\nu=sh_k$, and
$d_\nu\geq c\xi Z_k$. These imply the stated rate and distance.
For soundness, replace the actual middle-block bound in that
proposition by the larger value $n^{t-k}$ and cancel $A_0$:
\[
 \sigma_\nu\geq
 \frac{vM_k}{sn^{t-k}(M_{k-1}+M_{k+1})}
 =\frac{v e_{t-k}(S)}
 {sn^{t-k}(e_{t-k+1}(S)+e_{t-k-1}(S))}.
\]
The same proposition gives the two locality bounds.

Every constant in these bounds depends only on the fixed
parameters preceding \eqref{comb:numerical-data}. They are
positive: $c>0$, all elementary symmetric sums in range are
positive, and every displayed minimum has positive entries.
Finally $N_\nu$ is a fixed positive multiple of $V_\nu$, so
the asserted growth follows.
\end{proof}

\begin{proof}[Proof of Theorem~\ref{thm:abstract}]
The hypotheses of that theorem are those of Theorem~\ref{comb:main}.
Indeed, $|X_\nu(0)|=2^tV_\nu$, so their two growth conditions
are equivalent, and universal product expansion means expansion
over every coefficient-field extension as required above.
The explicit distance and soundness bounds in Theorem~\ref{comb:main}
are fixed and positive. Its rate and locality conclusions are
the bounds claimed in Theorem~\ref{thm:abstract}.
\end{proof}

\begin{corollary}\label{comb:weaken}
The distance and soundness conclusions of Theorem~\ref{comb:main}
remain valid if $c$ is replaced throughout the output constants
by any $c'$ with $0<c'\leq c$.
\end{corollary}
\begin{proof}
The quantities $r_j,B,T,w,\xi$ do not depend on $c$. Each other
occurrence of $c$ in the output constants is a positive multiple
inside a minimum. Decreasing $c$ therefore only decreases the
asserted distance and soundness bounds. Rate and locality are
unchanged.
\end{proof}

\section{Arithmetic realization of the cubical complexes}
\label{sec:arithmetic}\label{exar:section}

For a nonarchimedean local field $L$, write $T(L)$ for its
rank-two lattice tree. Its vertices are homothety classes of
$\mathcal O_L$-lattices in $L^2$, and adjacency is represented by
$\pi\Lambda\subsetneq\Lambda'\subsetneq\Lambda$. A product of
these trees has the cubical cells obtained by choosing a vertex
or an edge in each factor. Quotients in the following theorem
initially mean cubical incidence complexes; the face-intersection
axiom of Definition~\ref{def:cubical-complex} is asserted only
at the indicated sufficiently deep levels.

\begin{theorem}\label{exar:tower}
Fix $t\ge1$, pairwise distinct integers $a_i\ge0$ for $i\in[t]$,
and an integer $r\ge1$. Put
\[
 q_i=2^{r+a_i},\qquad n_i=q_i+1,\qquad q_-=\min_i q_i,\qquad
 \lambda_0=\max_i\frac{2\sqrt{q_i}}{q_i+1}<1.
\]
There exist a totally real number field $F$, a totally definite
quaternion algebra $B/F$, a maximal order $\mathcal M\subset B$,
and distinct split finite places $v_i$ of residue cardinality $q_i$,
with a descending sequence of principal norm-one congruence groups
$\Gamma_\nu$ such that the following hold.
\begin{enumerate}
\item The group $\Gamma_\nu$ acts freely and preserves the
directional vertex parities on
$\mathcal T=\prod_{i=1}^tT(F_{v_i})$. The quotient
$X_\nu=\Gamma_\nu\backslash\mathcal T$ is finite.
\item For every integer $R\ge0$, all sufficiently deep quotient
maps identify every full radius-$R$ cubical ball with the
corresponding product-tree ball.
\item At all sufficiently deep levels, $X_\nu$ is an
$(n_1,\ldots,n_t)$-regular, $\lambda_0$-expanding cubical complex
in the sense of Definitions~\ref{def:cubical-complex}
and~\ref{def:expanding-cubical-complex}.
For every upstairs face, the quotient map identifies its
entire upward star, including all incidences, with the entire
upward star of its image.
\item The parity mass $V_\nu=2^{-t}|X_\nu(0)|$ tends to infinity.
The local data are independent of $\nu$, and
\[
 \lambda_0\le\frac2{\sqrt{q_-}},\qquad
 \frac{\max_i n_i}{\min_i n_i}
       \le 2^{\max_i a_i-\min_i a_i}.
\]
\end{enumerate}
At each walking place, an integral splitting identifies the local
action with that of
$\mathrm{SL}_2(F_{v_i})/\{\pm I_2\}$ and the endpoint stabilizers
of Lemma~\ref{exdiag:lattice-tree}.
\end{theorem}

\subsection{Local trees and endpoint stabilizers}

Let $L$ be a nonarchimedean local field with normalized valuation
$v_L:L^\times\to\mathbb Z$, valuation ring $\mathcal O$, uniformizer
$\pi$, and residue field $k=\mathcal O/\pi\mathcal O$ of size $q$.
A lattice is a finitely generated spanning $\mathcal O$-submodule
of $L^2$. Since $\mathcal O$ is a discrete valuation ring, such a
lattice is free of rank two. The graph described above is a
$(q+1)$-regular tree
\cite[Definition~23.5.5 \& Proposition~23.5.8]{Voight2021}.
The neighbors of $[\Lambda]$ are indexed by the lines in
$\Lambda/\pi\Lambda$: the line $\ell$ gives its full inverse image
$\Lambda_\ell$ in $\Lambda$. These lattices are pairwise
nonhomothetic, since they all have the same determinant valuation
and a homothety between two of them must therefore be a unit.

The vertex type of $[g\mathcal O^2]$ is
$v_L(\det g)\bmod2$. A change of lattice basis changes the
determinant by a unit; a homothety changes its valuation by an
even integer. This type is consequently well-defined. Adjacent
vertices have opposite types, and left multiplication by
$g\in\mathrm{GL}_2(L)$ changes the type by
$v_L(\det g)\bmod2$.

\begin{lemma}\label{exdiag:lattice-tree}
Put
\[
 G=\mathrm{SL}_2(L)/\{\pm I_2\},\quad
 s=\begin{pmatrix}0&1\\\pi&0\end{pmatrix},\quad
 v_0=[\mathcal O^2],\quad v_1=s v_0,\quad e_*=\{v_0,v_1\}.
\]
The group $G$ preserves the types and acts transitively on
vertices of either fixed type and on edges. The stabilizers are
\[
 K_0=\mathrm{SL}_2(\mathcal O)/\{\pm I_2\},\qquad
 K_1=sK_0s^{-1},\qquad I=K_0\cap K_1.
\]
Thus the incident edges at $v_b$ are indexed by $K_b/I$,
and $[K_b:I]=q+1$ for $b\in\{0,1\}$. In $K_0$ coordinates,
$I$ consists of the determinant-one integral matrices whose
lower-left entry lies in $\pi\mathcal O$. These integral lifts
reduce onto the upper-triangular subgroup of $\mathrm{SL}_2(k)$. Conjugation by $s$ normalizes $I$
and exchanges the two diagonal residue entries.

If $q$ is even, reduction gives homomorphisms
\[
 K_0\longrightarrow\mathrm{SL}_2(k),\qquad
 K_1\xrightarrow{s^{-1}(\cdot)s}K_0
       \longrightarrow\mathrm{SL}_2(k),
\]
and both are surjective. Under either of these maps, $K_b/I$
identifies with $\mathbb P^1(k)$, with the distinguished coset
$I$ corresponding to $[1:0]$.
\end{lemma}
\begin{proof}
For a vertex $[g\mathcal O^2]$ of type $b$, multiply $g$ by a
scalar power of $\pi$ so that $v_L(\det g)=b$; this does not
change its lattice class. Put
$u=\det(g)/\det(s^b)\in\mathcal O^\times$.
Then $g\operatorname{diag}(u^{-1},1)s^{-b}$ has determinant one
and maps $v_b$ to $[g\mathcal O^2]$. This proves the two
transitivity assertions for vertices.

If a determinant-one matrix stabilizes $[\mathcal O^2]$, it
maps $\mathcal O^2$ onto $c\mathcal O^2$ for a scalar $c$.
Determinants give $2v_L(c)=0$, so $c$ is a unit. Its stabilizer
is therefore $\mathrm{SL}_2(\mathcal O)$, and conjugation gives
the stabilizer of $v_1$. Reduction to $\mathrm{SL}_2(k)$ is
surjective because elementary matrices over $k$ lift to
determinant-one elementary matrices over $\mathcal O$.
The finite group is transitive on the lines of $k^2$. The
neighbor-line description proves edge transitivity.

For every matrix with entries in $L$,
\[
 s^{-1}\begin{pmatrix}a&b\\c&d\end{pmatrix}s
   =\begin{pmatrix}d&c/\pi\\\pi b&a\end{pmatrix}.
\]
The simultaneous integrality conditions are precisely
$a,b,d\in\mathcal O$ and $c\in\pi\mathcal O$, together with
determinant one. Their reductions have reciprocal diagonal
entries, which the displayed conjugation exchanges.
The equality $s^2=\pi I_2$ proves that conjugation by $s$
normalizes $I$. In characteristic two the scalar signs have
the same residue, so reduction factors through the stated
projective stabilizers. The line fixed by the upper-triangular
subgroup is $[1:0]$, giving the final coset description.
\end{proof}

\subsection{Norm-one arithmetic quotients}

Write $\mathcal O_F$ for the ring of integers of a number field
$F$, and $F_v,\mathcal O_v,k_v$ for its completion, valuation
ring, and residue field at a finite place $v$. A totally real
field is a number field all of whose complex embeddings have
real image. A quaternion algebra $B/F$ is totally definite
when $B\otimes_{F,\sigma}\mathbb R$ is Hamilton's quaternion
algebra for every real embedding $\sigma$. The reduced norm
is denoted by $\operatorname{nrd}$, and
\[
 H(F)=B^1(F)=\{x\in B:\operatorname{nrd}(x)=1\}.
\]
A finite place is split when $B_v\simeq M_2(F_v)$; otherwise
it is ramified. We use the standard local and global
quaternion-algebra conventions of~\cite[Chapters~10,14,23]{Voight2021}.

\begin{proposition}\label{exar:prescribed-arithmetic}
For positive integers $\ell_1,\ldots,\ell_t$, there exist a totally
real field $F$, distinct finite places $v_i$ with
$|k_{v_i}|=2^{\ell_i}$, and a totally definite quaternion algebra
$B/F$ split at every $v_i$ and ramified at a finite place.
One may choose $[F:\mathbb Q]$ to be any integer
$d\ge\sum_i\ell_i$. A maximal order $\mathcal M\subset B$ exists,
and the split identifications can be chosen to satisfy
$\mathcal M_{v_i}=M_2(\mathcal O_{v_i})$.
\end{proposition}
\begin{proof}
In the prescribed-local-data construction of
Jordan and Livn\'e~\cite{JordanLivne2000}, take the
number of directions to be $t$, every rational prime to be $2$,
and the prescribed residue degrees to be $\ell_i$. Its degree
condition is $d\ge\sum_i\ell_i$, and the finite ramification
of its quaternion algebra avoids the prescribed places.
This gives the asserted field, places, and algebra.
A quaternion algebra over a characteristic-zero field is
separable; a maximal order therefore exists
\cite[Secition~10.4.2]{Voight2021}. Its completions are maximal orders
by~\cite[Section~23.2.2]{Voight2021}. At a split place the valuation
ring is a principal ideal domain, so
\cite[Corollary~10.5.5]{Voight2021}, with matrix size two,
conjugates that completed order to $M_2(\mathcal O_v)$.
Compose the splitting with this conjugation.
\end{proof}

Fix such $F,B,\mathcal M$ and a nonempty set
$S=\{v_1,\ldots,v_t\}$ of split finite places. For a nonzero
ideal $\mathfrak n\subseteq\mathcal O_F$ supported outside $S$, set
\[
 \mathcal M_S=\{x\in B:x_v\in\mathcal M_v
                     \text{ for all finite }v\notin S\},
\]
\begin{equation}\label{exar:principal-subgroups}
 \Gamma(\mathfrak n)=
 \{x\in H(F)\cap\mathcal M_S:
          x_v-1\in\mathfrak n\mathcal M_v
          \text{ for every }v\mid\mathfrak n\}.
\end{equation}
These are principal congruences on norm-one elements.
Their local inverses are integral: the standard involution
preserves orders and $x^{-1}=\overline x$ when
$\operatorname{nrd}(x)=1$. Consequently the displayed sets
are groups and are finite-index congruence subgroups of
$H(F)\cap\mathcal M_S$. The finite-index assertion follows
also directly from the reduction homomorphism into the finite
rings $\mathcal M_v/\mathfrak n\mathcal M_v$; details of the
local algebra and finiteness are given in
Appendix~\ref{app:arithmetic}.

\begin{lemma}\label{exar:strong-approximation}
Let $v_*$ be a split finite place. If $A$ is a finite set of
finite places disjoint from $\{v_*\}$ and
$U_v\subset H(F_v)$ is nonempty and open for every $v\in A$,
there is $x\in H(F)$ such that
\[
 x_v\in U_v\quad(v\in A),\qquad
 x_v\in\mathcal M_v\quad(v\notin A\cup\{v_*\},\ v\text{ finite}).
\]
\end{lemma}
\begin{proof}
Apply norm-one strong approximation
\cite[Definition~28.5.1 \& Theorem~28.5.3]{Voight2021}
with omitted set $\Sigma$ equal to all infinite places together
with $v_*$. The split place $v_*$ makes $B$ $\Sigma$-indefinite,
which is the source hypothesis. The diagonal $H(F)$ is dense
in the restricted product away from $\Sigma$. The set
\[
 \prod_{v\in A}U_v\times
 \prod_{\substack{v\notin A\cup\Sigma\\v\text{ finite}}}
             (H(F_v)\cap\mathcal M_v)
\]
is a nonempty basic open set in that restricted product.
An element of $H(F)$ in this set satisfies exactly the required
conditions.
\end{proof}

\begin{proposition}\label{exar:finite-quotient}
For every level $\mathfrak n$ above, the action of
$\Gamma(\mathfrak n)$ on $\mathcal T=\prod_iT(F_{v_i})$
has finitely many vertex and cell orbits and finite cell
stabilizers. If $\Gamma(\mathfrak n)$ is torsion free, the
action is free on cells and the quotient map is a cubical
covering. All directional parities descend.
\end{proposition}
The adelic compactness and double-coset proof is in
Appendix~\ref{app:arithmetic}, Lemma~\ref{exar:norm-one-compact}
and the proof following it. Parity preservation follows already
from determinant one in every local factor; hence no element
can invert an edge or permute distinct vertices of a fixed cube.

Choose an odd rational prime $p_0$ unramified in $F$ and below
no place in $S$ or in the finite ramification set of $B$.
Such primes exist because those exceptional sets are finite
\cite[Section~14.4.4 and Lemma~14.5.3]{Voight2021}.
Choose $\ell_0\mid p_0$ and put
\[
 \Gamma_0=\Gamma(\ell_0).
\]
Then $\ell_0$ is split and
$\operatorname{ord}_{\ell_0}(p_0)=1$.
Lemma~\ref{exar:free} in Appendix~\ref{app:arithmetic} proves
that $\Gamma_0$ is torsion free. In particular it acts freely,
has finite quotient, and excludes $-1$. Its homomorphism to
the product of the local projective groups is injective:
a quaternion acting trivially on a local tree is scalar,
and a scalar of reduced norm one is $1$ or $-1$.

\subsection{Parallel components and Ramanujan bounds}

For any torsion-free $\Gamma(\mathfrak n)$, use the quotient
cells to form the parallel-face multigraph $G_{I,j}$, for
$I\subseteq[t]$ and $j\notin I$. Its vertices are the type-$I$
cell orbits and its edges the type-$(I\cup\{j\})$ cell orbits.
An edge joins the two opposite type-$I$ facets. Multiplicities
are retained if distinct quotient cells give the same pair.

\begin{proposition}\label{exar:JL}
For every torsion-free level $\Gamma(\mathfrak n)$ and every
$I\subseteq[t]$, $j\notin I$, the unnormalized adjacency
matrix $A_{I,j}$ satisfies
\[
 \operatorname{Spec}(A_{I,j})
 \subseteq\{-(q_j+1),q_j+1\}\cup[-2\sqrt{q_j},2\sqrt{q_j}],
 \qquad q_j=|k_{v_j}|.
\]
The assertion holds at every deeper principal level supported
outside $S$.
\end{proposition}
\begin{proof}
Livn\'e's Ramanujan criterion~\cite{Livne2001} applies to
the norm-one $S$-arithmetic torsion-free congruence group
$\Gamma(\mathfrak n)$, with its number of directions equal
to $t$, its field and algebra equal to $F,B$, and its local
degrees equal to $q_i+1$. Its Ramanujan conclusion is conditional
on the Ramanujan--Petersson property for the indicated Hilbert
cusp spaces of parallel weight two.

For every cuspidal holomorphic representation in these spaces,
use the unitary normalization of
\cite{Blasius2006}.
Its weight at each of the $d=[F:\mathbb Q]$ real places is two.
All weights are at least two and have the same parity.
Thus~\cite[Theorem~1]{Blasius2006} gives
Ramanujan at every finite place, including each $v_i$.
It therefore supplies the conditional premise of Livn\'e's
theorem. This application imposes no condition that the rational
prime under one walking place be coprime to the other levels.

It remains to identify the operators. Use the graph and star
operator of~\cite{JordanLivne2000}
with its $g=t$, $r_j=q_j+1$, and trivial complex local system.
Orient each active cube direction from parity zero to parity
one. Its graph vertices are exactly the type-$I$ quotient cells;
its directed edges are the type-$(I\cup\{j\})$ cells with both
choices of orientation in direction $j$. The star operator is
therefore
\[
 (S_{j,I}z)(f)=
 \sum_{\substack{e:\ \operatorname{terminal}(e)=f}}
             z(\operatorname{origin}(e))=(A_{I,j}z)(f),
\]
including all edge multiplicities. The source's Ramanujan condition
bounds every eigenvalue except $\pm r_j$ by $2\sqrt{r_j-1}$.
This is precisely the asserted spectral inclusion.
A deeper principal level remains a torsion-free congruence
subgroup of the same norm-one group, so the same application
applies separately at that level.
\end{proof}

\begin{lemma}\label{exar:components}
For every principal subgroup $\Gamma\le\Gamma_0$, the connected
components of $G_{I,j}$ are exactly the fibers of the inactive
parities outside $I\cup\{j\}$. There are
$2^{t-1-|I|}$ such components.
\end{lemma}
\begin{proof}
An upstairs parallel component is the $j$th tree together with
a fixed transverse tuple: one edge in every factor indexed by
$I$, and one vertex in each factor outside $I\cup\{j\}$.
Consequently the downstairs components are the $\Gamma$-orbits
of these transverse tuples. Indeed, the image of one upstairs
component is connected; an identification between two such
images comes from a group element taking one transverse tuple
to the other.

The inactive parities are preserved by $\Gamma$. Conversely,
take two tuples having the same inactive parities.
For each $i\ne j$, Lemma~\ref{exdiag:lattice-tree} gives an
element of $\mathrm{SL}_2(F_{v_i})$ taking the first specified
edge or vertex to the second. The set of all such elements is
a nonempty open stabilizer coset. Apply
Lemma~\ref{exar:strong-approximation} with omitted place $v_j$.
At the other walking places impose these cosets, and at the
finitely many level places impose the principal identity
congruences defining $\Gamma$. The resulting rational norm-one
element is integral at every remaining finite place outside
$S$, hence belongs to $\Gamma$, and maps the first tuple
to the second. Thus each inactive-parity fiber is one orbit.
Every pattern occurs in the product of trees, proving the count.
\end{proof}

\begin{corollary}\label{exar:face-counts}
On each parallel component, the normalized adjacency has simple
eigenvalues $1,-1$ and all remaining eigenvalues have absolute
value at most $2\sqrt{q_j}/(q_j+1)$.
Once the quotient is a cubical complex, it is
$\lambda_0$-expanding, with $\lambda_0$ as in
Theorem~\ref{exar:tower}, and each bipartite half of a parallel component
has $V\prod_{i\in I}n_i$ vertices.
\end{corollary}
\begin{proof}
The components of Lemma~\ref{exar:components} are connected,
regular, bipartite graphs. For a connected regular graph the
constant eigenfunction spans the adjacency eigenspace of its
degree; multiplying by the bipartite sign identifies this
space with the eigenspace of the negative degree. For example,
the equality
$\sum_{\{x,y\}\in E}(z(x)-z(y))^2=0$ for a top eigenvector
forces constancy along every edge and hence everywhere.
The remaining eigenvalues satisfy Proposition~\ref{exar:JL}.
Division by $q_j+1$ proves the norm bound on the orthogonal
complement of these two vectors. The half-size is the
parity-face count of Lemma~\ref{lem:cubical-face-counts},
with the $j$th parity fixed.
\end{proof}

\subsection{Congruence towers and complete upward stars}

Choose a finite place $w\notin S\cup\{\ell_0\}$. For $\nu\ge1$ set
\[
 \Gamma_\nu=\Gamma(\ell_0w^\nu),\qquad
 X_\nu=\Gamma_\nu\backslash\mathcal T .
\]
These are normal finite-index subgroups of $\Gamma_0$.
Their intersection is trivial: if $x$ belongs to every
$\Gamma_\nu$, then $x-1$ belongs to every $w^\nu\mathcal M_w$,
whose intersection is zero, and the map $B\to B_w$ is injective.
The indices $[\Gamma_0:\Gamma_\nu]$ tend to infinity. Otherwise
the nondecreasing indices stabilize; nested subgroups of the
same finite index are equal, giving a finite-index trivial
subgroup of $\Gamma_0$. This would make $\Gamma_0$ finite.
But a finite group cannot have finitely many orbits on the
infinite vertex set of $\mathcal T$. Proposition
\ref{exar:finite-quotient} gives those finite orbits, a contradiction.
Each $\Gamma_0$-cell orbit splits into exactly
$[\Gamma_0:\Gamma_\nu]$ cell orbits because the stabilizers are
trivial. Thus all parity masses grow by this index.

A full radius-$R$ cubical ball means the vertex graph ball
together with every cube all of whose vertices belong to that
ball. The following elementary covering argument will also
supply the whole-star condition used in Section~\ref{sec:equivariance}.

\begin{lemma}\label{exar:injectivity}
For every integer $R\ge0$, all sufficiently deep maps
$\mathcal T\to X_\nu$ identify all full radius-$R$ cubical balls.
For all sufficiently deep levels the quotient face poset
satisfies the intersection axiom, is $(n_1,\ldots,n_t)$-regular,
and the restriction to the entire upward star of every face
is a poset isomorphism onto the entire upward star of its image.
\end{lemma}
\begin{proof}
Take representatives $x_1,\ldots,x_s$ of the finitely many
$\Gamma_0$-vertex orbits. For a fixed $D$, each set
\[
 \{\gamma\in\Gamma_0:
                 d(x_a,\gamma x_a)\le2D\}
\]
is finite: the ball is finite, and freeness allows at most
one group element carrying $x_a$ to each vertex.
Residuality excludes all its nonidentity elements from
$\Gamma_\nu$ for sufficiently large $\nu$. Normality in
$\Gamma_0$ makes this exclusion valid at every center.

If two vertices $y,z$ in a radius-$D$ upstairs ball centered
at $x$ have the same quotient image, then $z=\gamma y$ with
$\gamma\in\Gamma_\nu$, and
$d(x,\gamma x)\le d(x,z)+d(z,\gamma x)\le2D$.
Thus $\gamma=1$, proving vertex injectivity. Every quotient
path of length at most $R$ lifts uniquely from a chosen center,
since the action is free and has no inversions. Apply vertex
injectivity with $D=R+t$. A cube whose quotient vertices lie
in the radius-$R$ ball can be lifted through a chosen lifted
vertex. Every vertex of this lift is within distance $R+t$
of the center and has the same image as the corresponding
path lift. Injectivity identifies them. The same argument
shows uniqueness of the lifted cube, by comparing two lifts
through one vertex. Thus full radius-$R$ balls agree.

Now choose a level where these balls agree for radius $2t$.
Two quotient faces having a common vertex lie, with all their
subfaces, in one such product-tree ball. In a product of trees
their common subfaces are the subfaces of their coordinatewise
intersection, or there are none. This proves the intersection
axiom downstairs. The same ball identifies any list of
distinct-color edges at a vertex with such a list upstairs;
their unique product cube descends, and every competing cube
would also lift in this ball. Hence the quotient is $(n_1,\ldots,n_t)$-regular.
Finally, every coface of a given face lies in the radius-$t$
ball about any vertex of that face. All these cofaces and all
their inclusions are therefore identified simultaneously with
the upstairs upward star.
\end{proof}

\begin{proof}[Proof of Theorem~\ref{exar:tower}]
Apply Proposition~\ref{exar:prescribed-arithmetic} with
$\ell_i=r+a_i$, and choose the integral splittings there.
The preceding congruence construction gives the finite free
quotients, fixed local actions, and growing parity masses.
Lemma~\ref{exar:injectivity} gives all the asserted local
and face-poset properties after one fixed initial segment.
Lemma~\ref{exar:components} and
Corollary~\ref{exar:face-counts} give the exact component
condition and the required spectral norm.
Since $q_i\ge2$, the inequality $2\sqrt{q_i}<q_i+1$ gives
$\lambda_0<1$, and
$2\sqrt{q_i}/(q_i+1)\le2/\sqrt{q_-}$.
For integers $b\ge a$,
\[
 \frac{2^{r+b}+1}{2^{r+a}+1}\le2^{b-a}.
\]
Taking the largest and smallest offsets proves the final ratio
bound. All choices except $\nu$ have been fixed throughout.
\end{proof}

\section{Weighted incidence on unequal norm-one grids}\label{exinc:chapter}
Throughout this section, $k=\overline{\mathbb F}_2$. Varieties are reduced
closed subschemes of projective space; an integral variety is irreducible.
Degrees are projective degrees. For an integer $a\geq1$, an integral
projective curve $C\subseteq\mathbb P^a_k$, and a finite set
$Z\subseteq\mathbb P^a(k)$, put
\[
 \mu_Z(C)=\sum_{z\in Z\cap C(k)}\mult_zC,
\]
where $\mult_zC$ is the Hilbert--Samuel multiplicity of
$\mathcal O_{C,z}$.

\subsection{Norm-one grids and polynomial closure}
\begin{definition}\label{exinc:grid-definition}
\label{exinc:closure-definition}\label{exinc:point-stratum}
Let $e\geq1$ and let $m_1,\ldots,m_e$ be distinct positive integers. Set
\begin{equation}\label{exinc:grid-parameters}
 q_i=2^{m_i},\quad n_i=q_i+1,\quad S_i=\{z\in k^\times:z^{n_i}=1\},
 \quad\Omega=\prod_{i=1}^eS_i,\quad m=\min_i n_i.
\end{equation}
For an integer $\ell\geq0$, let $\mathcal P_\ell$ be the space of
polynomials in $k[X_1,\ldots,X_e]$ having degree at most $\ell$ in each
variable. For $A\subseteq\Omega$, define
\begin{equation}\label{exinc:box-closure}
 I_\ell(A)=\{F\in\mathcal P_\ell:F|_A=0\},\qquad
 \operatorname{cl}_\ell(A)=\{z\in\Omega:F(z)=0\text{ for all }F\in I_\ell(A)\}.
\end{equation}
The set $A$ is $\ell$-closed if $\operatorname{cl}_\ell(A)=A$.
An $i$-line is a set obtained by fixing all coordinates except $i$.
The point stratum is
\[
 A^\circ=\{z\in A:\text{no coordinate line through }z\text{ is contained in }A\}.
\]
For $I\subseteq[e]$ and $a\in\prod_{i\notin I}S_i$, the coordinate slice
is $A_{I,a}=\{z\in\prod_{i\in I}S_i:(z,a)\in A\}$, with factors in their
original order.
\end{definition}
The polynomial $X^{q_i+1}-1$ has derivative $X^{q_i}$, hence
$|S_i|=n_i$. If $\max_iq_i/\min_iq_i\leq K$, then
\begin{equation}\label{exinc:grid-cardinality}
 m\leq n_i\leq Km,\qquad |\Omega|=\prod_i n_i.
\end{equation}
These inequalities hold on every nonempty coordinate subtuple.

\begin{lemma}\label{exinc:closure-slices}
\label{exinc:closed-coordinate-slice}\label{exinc:sliced-point-stratum}
\label{exinc:zero-closure}\label{exinc:one-dimensional-closure}
If $A$ is $\ell$-closed, then every nonempty-coordinate slice $A_{I,a}$
is $\ell$-closed and $(A^\circ)_{I,a}\subseteq(A_{I,a})^\circ$.
A $0$-closed set is empty or the entire grid. In one coordinate a proper
$\ell$-closed set has at most $\ell$ points.
\end{lemma}
\begin{proof}
For $z\notin A_{I,a}$, closedness supplies $F\in I_\ell(A)$ with
$F(z,a)\ne0$. Substituting the fixed coordinates produces a separating
polynomial on the slice with the same individual-degree bounds.
A full line in the slice is a full line in the original grid, proving
the point-stratum inclusion. Constants separate the empty set and cannot
separate a nonempty set. Finally, a proper one-coordinate closed set is
contained in the roots of a nonzero polynomial of degree at most $\ell$.
\end{proof}

Fix $t\geq1$ and $K\geq1$. For $1\leq e\leq t$, set
\begin{equation}\label{exinc:incidence-constants}
 \gamma_e=2^{-\max\{e-2,0\}},\qquad
 B_t(K)=(t-1)(tK+1)+2,\qquad P_1(t,K)=1.
\end{equation}
For $t\geq2$, set $P_2(t,K)=K+1$. For $3\leq e\leq t$, define
\begin{equation}\label{exinc:constant-recurrence}
 \begin{aligned}
 S_e(t,K)&=K^eP_{e-1}(t,K),& L_e(t,K)&=4e!S_e(t,K),\\
 H_e(t,K)&=e^e(1+L_e(t,K)),& P_e(t,K)&=B_t(K)H_e(t,K).
 \end{aligned}
\end{equation}
The symbol $S_e(t,K)$ is a number, distinct from the set $S_i$.

\begin{theorem}\label{exinc:main}
\label{exinc:weighted-grid-incidence}
Fix $t\geq1$, $K\geq1$, and a norm-one grid of dimension $1\leq e\leq t$
with $\max_iq_i/\min_iq_i\leq K$. If $0<\theta\leq1/2$,
$0\leq\ell\leq\theta m$ is an integer, $A\subseteq\Omega$ is $\ell$-closed,
and $C\subseteq\mathbb P^e_k$ is an integral projective curve, then
\begin{equation}\label{exinc:main-bound}
 \mu_{A^\circ}(C)\leq P_e(t,K)\theta^{\gamma_e}m\deg C.
\end{equation}
The grid is embedded by $z\mapsto[1:z_1:\cdots:z_e]$. The same statement
holds on every nonempty subtuple with its actual dimension and minimum
length, retaining the fixed parameters $t,K$.
\end{theorem}

\subsection{Hilbert growth and weighted intersections}
For $a\geq1$ and a nonempty projective scheme $Y\subseteq\mathbb P^a_k$, set
\[
 H_Y(b)=\dim_k(k[X_0,\ldots,X_a]/I_Y)_b.
\]
If $Y$ has dimension $h$,
its Hilbert polynomial has leading term $(\deg Y)b^h/h!$; for a curve it
is $(\deg Y)b+1-p_a(Y)$, defining its arithmetic genus
\cite[I, Theorem 7.5 and Proposition 7.6; IV, \S1]{Hartshorne1977}.
For an integral curve $C$, its finite normalization
$\nu:\widetilde C\to C$ is a nonsingular projective curve
\cite[I, Theorem 6.2A; II, Exercise 3.8; III, Exercise 5.8]{Hartshorne1977}.
At a branch $w\in\nu^{-1}(z)$, write $v_w$ for its normalized discrete
valuation, with $v_w(0)=+\infty$, and for an ideal $J\subseteq\mathcal O_{C,z}$ put
$v_w(J)=\min_{f\in J}v_w(f)$. The formulas used below are
\begin{equation}\label{exag:branch-formulas}
 \mult_zC=\sum_{w\mid z}v_w(\mathfrak m_z),\qquad
 \delta_z(C)=\length_k((\nu_*\mathcal O_{\widetilde C})_z/\mathcal O_{C,z}),
 \qquad p_a(C)=g(\widetilde C)+\sum_z\delta_z(C),
\end{equation}
where $g(\widetilde C)=\dim_kH^1(\widetilde C,\mathcal O_{\widetilde C})$.
The multiplicity formula is the normalization formula for Samuel
multiplicity \cite[Examples 1.2.3, 4.3.4, and 4.3.6]{Fulton1998}; the genus
formula follows from the normalization exact sequence
\cite[IV, Exercise 1.8]{Hartshorne1977}.

\begin{lemma}\label{exag:external-intersection}
\label{exag:weighted-bezout}
Let $V\subseteq\mathbb P^a_k$ be integral of dimension $h\geq1$ and let
$F$ be a homogeneous polynomial of degree $b\geq1$ not vanishing on $V$.
Every irreducible component of $V\cap V_+(F)$ has dimension $h-1$, and
\[
 \sum_{W\in\operatorname{Irr}(V\cap V_+(F))}\deg W\leq b\deg V.
\]
If $V=C$ is a curve and $Z\subseteq C(k)\cap V_+(F)$, then
$\mu_Z(C)\leq b\deg C$. The same bound holds for an affine polynomial
of total degree at most $b$, using its homogenization.
\end{lemma}
\begin{proof}
The proper-section theorem, with scheme multiplicities, gives the first
assertion and total degree $b\deg V$
\cite[I, Theorem 7.7]{Hartshorne1977}. All component multiplicities are
positive, so dropping them gives the displayed inequality. For a curve,
the pullback of the section $F$ to its normalization has zero divisor of
degree $b\deg C$. At every $z\in Z$, its local equation belongs to
$\mathfrak m_z$, so its total branch order is at least $\mult_zC$ by
\eqref{exag:branch-formulas}. Summing proves the weighted assertion.
\end{proof}

For $1\leq c\leq a-1$, an \emph{expected component} of $c$ homogeneous
equations in $\mathbb P^a$ means an actual irreducible component of
dimension $a-c$; it does not mean an arbitrary subvariety of that dimension.

\begin{lemma}\label{exag:expected-component-hilbert}
\label{exag:expected-hilbert}\label{exag:expected-genus}
Let $a\geq2$, $1\leq c\leq a-1$, and $B\geq1$ be integers. Let
$E_1,\ldots,E_c\in k[X_0,\ldots,X_a]$ be nonzero homogeneous forms of
degrees between $1$ and $B$. If $Y$ is a nonempty reduced union of
expected components of their common zero locus, then, for every $b\geq cB$,
\begin{equation}\label{exag:expected-hilbert-bound}
 H_Y(b)\geq\deg Y\binom{b-cB+a-c}{a-c}.
\end{equation}
For $c=a-1$ this implies
$p_a(Y)\leq((a-1)B-1)\deg Y+1$.
\end{lemma}
\begin{proof}
Apply \cite[Corollary 3]{ChardinPhilippon1999}, with the erratum
\cite{ChardinPhilippon2002}, to $S=k[X_0,\ldots,X_a]$,
$I=(E_1,\ldots,E_c)$, and $J=I_Y$. In the source notation the projective
dimension is $a$, the codimension is $c$, and the listed degrees are
$d_1\geq\cdots\geq d_c$, where $d_i\leq B$. Its required inclusion holds:
if $\mathscr P_c$ is the set of codimension-$c$ minimal primes of $I$,
and $\mathscr Q\subseteq\mathscr P_c$ selects the components of $Y$, then
\[
 I^{\langle c\rangle}=\bigcap_{\mathfrak p\in\mathscr P_c}Q_{\mathfrak p}
 \subseteq\bigcap_{\mathfrak p\in\mathscr Q}\mathfrak p=J,
 \qquad\operatorname{codim}J=c.
\]
Here $Q_{\mathfrak p}$ are the isolated primary factors, and
$Q_{\mathfrak p}\subseteq\mathfrak p$ suffices; no reducedness along the
components is required. The cited result gives
\[
 b>\sum_i d_i-c\quad\Longrightarrow\quad
 H_Y(b)\geq\deg Y\binom{b+a-\sum_i d_i}{a-c}.
\]
For $b\geq cB$ this implies \eqref{exag:expected-hilbert-bound} by
monotonicity of the binomial coefficient. Corollaire 3 is unaffected by
the erratum's correction to the regularity definition and proof of
Proposition 2. For $c=a-1$, compare the resulting lower bound
$(\deg Y)(b-cB+1)$ with the eventual Hilbert polynomial
$(\deg Y)b+1-p_a(Y)$.
\end{proof}

\begin{lemma}\label{exag:normalization-intersection}
\label{exag:weighted-union-intersection}
Let $C,C'\subseteq\mathbb P^a_k$ be distinct integral curves, $a\geq2$,
and let $Z\subseteq C(k)\cap C'(k)$. Then
\begin{equation}\label{exag:normalization-union-bound}
 \mu_Z(C)\leq p_a(C\cup C')-p_a(C')-g(\widetilde C)+1.
\end{equation}
If both curves are expected components of $a-1$ equations of degrees
between $1$ and $B$, then
\begin{equation}\label{exag:expected-weighted-intersection}
 \mu_Z(C)\leq(a-1)B(\deg C+\deg C')+2.
\end{equation}
The union is reduced; the intersection is scheme-theoretic.
\end{lemma}
\begin{proof}
The intersection is finite because the curves are distinct. At $z$ put
$A=\mathcal O_{C,z}$, let $J\subset A$ be the ideal induced by $C'$,
and put $\widetilde A=(\nu_*\mathcal O_{\widetilde C})_z$. The cokernel
of $A/J\to\widetilde A/J\widetilde A$ is a quotient of
$\widetilde A/A$. Thus
\[
 \sum_{w\mid z}v_w(J)=\length_k(\widetilde A/J\widetilde A)
 \leq\length_k(A/J)+\delta_z(C).
\]
Since $J\subseteq\mathfrak m_z$, \eqref{exag:branch-formulas} gives
$\mult_zC\leq\length_k(A/J)+\delta_z(C)$. Sum over $Z$, then include
the remaining nonnegative terms, obtaining
\[
 \mu_Z(C)\leq\length(C\cap C')+\sum_z\delta_z(C).
\]
The exact sequence
$0\to\mathcal O_{C\cup C'}\to\mathcal O_C\oplus\mathcal O_{C'}
\to\mathcal O_{C\cap C'}\to0$, or its Hilbert polynomials, gives
\[
 \length(C\cap C')=p_a(C\cup C')-p_a(C)-p_a(C')+1.
\]
Substitution and \eqref{exag:branch-formulas} prove the first bound.
For the second, Lemma~\ref{exag:expected-component-hilbert} applied to
$Y=C\cup C'$ gives
$p_a(Y)\leq((a-1)B-1)(\deg C+\deg C')+1$.
Both $p_a(C')$ and $g(\widetilde C)$ are nonnegative by
\eqref{exag:branch-formulas}, proving the stated weaker bound.
\end{proof}

\subsection{The two-variable Frobenius bound}
\begin{lemma}\label{exinc:pair-norm-incidence}
Let $Q=2^u$ and $c=2^v$ for integers $u,v\geq1$. If
$F\in k[X,Y]$ is irreducible of bidegree $(a,b)$ with $a,b\geq1$,
and $C\subseteq\mathbb P^2_k$ is its reduced projective closure, then
\begin{equation}\label{exinc:pair-multiplicity-bound}
 \mu_{\mu_{Q+1}(k)\times\mu_{cQ+1}(k)}(C)\leq(c+1)ab.
\end{equation}
\end{lemma}
\begin{proof}
Coefficient Frobenius $u\mapsto u^{cQ}$ is an automorphism of $k$. Put
\[
 H=X^aY^bF^{[cQ]}(X^{-1},Y^{-1}),\qquad G=H(X^c,Y).
\]
As $F$ has positive degree in both variables, neither $X$ nor $Y$ divides
it. Laurent inversion and coefficient conjugation preserve
irreducibility, and clearing denominators introduces no coordinate
factor. Consequently $H$ is irreducible of bidegree $(a,b)$.
The map $(x,y)\mapsto(x^c,y)$ is finite and a universal homeomorphism:
on coordinate rings each variable is integral over its $c$th power,
and on every extension field there is at most one $c$th root. Hence the
reduced inverse image of $V(H)$ is irreducible, and unique factorization
gives $G=\lambda J^r$, with $J$ irreducible and $r\geq1$. If $F\mid G$,
then $J$ is proportional to $F$; comparing the two degrees gives
$ra=ca$ and $rb=b$, contradicting $c>1$. Thus $F,G$ have no common factor.

For $(\alpha,\beta)\in V(F)\cap(\mu_{Q+1}\times\mu_{cQ+1})$,
\[
 \alpha^{cQ}=\alpha^{-c},\qquad\beta^{cQ}=\beta^{-1},\qquad
 G(\alpha,\beta)=\alpha^{ca}\beta^bF(\alpha,\beta)^{cQ}=0.
\]
The exact bihomogenizations define divisors of bidegrees $(a,b)$ and
$(ca,b)$ on $\mathbb P^1\times\mathbb P^1$, without boundary components
or common components. Their intersection number is $(c+1)ab$
\cite[V, Theorem 1.1, Proposition 1.4, Example 1.4.3]{Hartshorne1977}.
At each branch $w$ over such a point,
\[
 v_w(G)\geq\min\{v_w(x-\alpha),v_w(y-\beta)\}.
\]
The right side sums to $\mult_{(\alpha,\beta)}C$ by
\eqref{exag:branch-formulas}; the left side sums to the local
intersection multiplicity \cite[Example 1.2.3]{Fulton1998}. Summing and
including the other nonnegative intersection terms proves the bound,
even when the divisor defined by $G$ is nonreduced.
\end{proof}

\subsection{The anisotropic grid-preserving transform}
Use the grid of Definition~\ref{exinc:grid-definition} and suppose its
power ratio is at most $K$. Put
\begin{equation}\label{exinc:transform-definition}
 Q=2^{\max_i m_i},\qquad c_i=Q/q_i,\qquad c_* =\max_i c_i\leq K,
 \qquad X=(\mathbb P^1_k)^e.
\end{equation}
On $X$, with $x_i=V_i/U_i$, set
\[
 \Phi(([U_i:V_i])_i)=([V_i^{c_i}:U_i^{c_i}])_i,
 \qquad T(V)=(\Phi^{-1}(V^{[Q]}))_{\mathrm{red}}.
\]
Here $V^{[Q]}$ applies $u\mapsto u^Q$ to coefficients of all defining
equations. For a curve meeting the torus, use its reduced torus closure
in $X$ when applying $T$, and return to its reduced torus closure in
$\mathbb P^e$ afterward. Both compactifications have the same function
field and coordinate functions, since they contain the same dense open
subscheme.
For $0\ne F\in k[X_1,\ldots,X_e]$, set
\begin{equation}\label{exinc:polynomial-transform}
 \mathcal T(F)=\prod_iX_i^{c_i\deg_{X_i}F}
               F^{[Q]}(X_1^{-c_1},\ldots,X_e^{-c_e}).
\end{equation}

\begin{lemma}\label{exinc:transform}
\label{exinc:polynomial-transport}\label{exinc:proper-transport}
The map $T$ bijects reduced closed subschemes of $X$, preserving
irreducible components and dimensions, and
\begin{equation}\label{exinc:grid-trace-equality}
 T(V)(k)\cap\Omega=V(k)\cap\Omega.
\end{equation}
If $\deg F\leq B$, then $\mathcal T(F)\ne0$,
$\deg\mathcal T(F)\leq eKB$, and it has the same grid zeros as $F$.
If $C'=T^{-1}(C)$ and $F|_{C'}\ne0$, then $\mathcal T(F)|_C\ne0$.
\end{lemma}
\begin{proof}
Coordinate Frobenius powers are finite universal homeomorphisms, and
inversion and coefficient conjugation are invertible. Their effects on
reduced closed loci give the asserted bijection. For $z\in\Omega$,
$z_i^{q_i}=z_i^{-1}$, hence $\Phi(z)=z^{[Q]}$. This proves
\eqref{exinc:grid-trace-equality}. Write $F=\sum_\alpha u_\alpha X^\alpha$
and $d_i=\deg_{X_i}F$. Then
\[
 \mathcal T(F)=\sum_\alpha u_\alpha^Q\prod_iX_i^{c_i(d_i-\alpha_i)},
 \quad \deg\mathcal T(F)\leq\sum_i c_id_i\leq eKB,
 \quad \mathcal T(F)(z)=\prod_i z_i^{c_id_i}F(z)^Q.
\]
Distinct exponents stay distinct, so the polynomial is nonzero. The last
formula proves the grid assertion. Finally
$\Phi(C_X)=(C'_X)^{[Q]}$; pullback of a nonzero function under this
dominant map is nonzero, and the monomial prefactor is invertible on
the torus.
\end{proof}

For an integral curve $C\subseteq\mathbb P^e$ with dense affine part,
let $d_i(C)$ be the degree of the pole divisor of $x_i$ on its
normalization, with $d_i(C)=0$ for a constant coordinate. A nonconstant
rational function defines a finite morphism to $\mathbb P^1$, and this
pole degree is its full function-field degree, including its
inseparable part \cite[II, Proposition 6.8, Proposition 6.9, Corollary 6.10]{Hartshorne1977}.

\begin{lemma}\label{exinc:coordinate-degrees}
\label{exinc:predecessor-degree}\label{exinc:predecessor-distinct}
For a curve $C$ with dense affine part,
\begin{equation}\label{exinc:coordinate-degree-comparison}
 \deg C\leq\sum_i d_i(C)\leq e\deg C.
\end{equation}
If $C$ meets the torus and $C'=T^{-1}(C)$, there is one positive integer
$r$ such that
\begin{equation}\label{exinc:coordinate-degree-pullback}
 c_i d_i(C)=r d_i(C')\quad(i\in[e]),\qquad \deg C'\leq eK\deg C.
\end{equation}
If at least two coordinates are nonconstant, then $C'\ne C$.
\end{lemma}
\begin{proof}
On the normalization, let $E_0=\operatorname{div}(\nu^*X_0)$, of degree
$\deg C$. Each coordinate pole divisor is bounded by $E_0$, proving the
upper bound in \eqref{exinc:coordinate-degree-comparison}. Choose a
linear form $L$ nonzero at the finite set $C\cap V_+(X_0)$. Such an $L$
exists because $k$ is infinite. The pole divisor of $L/X_0$ is exactly
$E_0$. Since $L/X_0$ is a linear combination of $1,x_1,\ldots,x_e$, the
valuation inequality bounds its pole order at each point by the sum
of the coordinate pole orders. This proves the lower bound.

The finite dominant map $\Phi|_{C_X}:C_X\to(C'_X)^{[Q]}$ induces a
finite map on normalizations of degree
$r=[k(C_X):k((C'_X)^{[Q]})]$. Under it the $i$th coordinate pulls back
to $x_i^{-c_i}$. Pullback multiplies divisor degree by $r$, and the
zero and pole divisors of a rational function have equal degree
\cite[II, Proposition 6.9 and Corollary 6.10]{Hartshorne1977}.
Coefficient conjugation preserves degrees. Consequently
$c_i d_i(C)=r d_i(C')$, including the constant-coordinate case, and
\[
 \deg C'\leq\sum_i d_i(C')=r^{-1}\sum_i c_id_i(C)
 \leq c_*\sum_i d_i(C)\leq eK\deg C.
\]
If $C'=C$ and coordinates $i,j$ vary, their positive degrees imply
$c_i=r=c_j$. The distinct exponents $m_i,m_j$ make this impossible.
\end{proof}

\begin{corollary}\label{exinc:expected-predecessor}
Suppose $e\geq2$, $C$ meets the torus and has two nonconstant
coordinates, and $C'=T^{-1}(C)$. If $C,C'$ are expected components of
$e-1$ homogeneous equations of degrees between $1$ and $B$, then
\begin{equation}\label{exinc:expected-predecessor-bound}
 \mu_\Omega(C)\leq((e-1)(eK+1)+2)B\deg C.
\end{equation}
\end{corollary}
\begin{proof}
The two curves are distinct, have the same grid trace, and satisfy
$\deg C'\leq eK\deg C$. Apply
\eqref{exag:expected-weighted-intersection} to this common trace and use
$B\deg C\geq1$.
\end{proof}

\subsection{Induction on the number of coordinates}
The recurrences \eqref{exinc:incidence-constants}--\eqref{exinc:constant-recurrence}
give, by induction on $e$,
\begin{equation}\label{exinc:constant-monotonicity}
 \begin{gathered}
 1\leq P_e(t,K)\leq P_t(t,K),\quad\gamma_e\geq\gamma_t,\\
 P_e(t,K)\geq KP_{e-1}(t,K)\quad(3\leq e\leq t),\\
 \max\{1,eK,(e-1)(eK+1)+2\}\leq B_t(K)\quad(2\leq e\leq t).
 \end{gathered}
\end{equation}
Indeed $P_2=K+1\geq P_1$, and
$B_t e^e(1+4e!K^eP_{e-1})\geq KP_{e-1}$ for $e\geq3$.
Every factor in the recurrence is positive and nondecreasing in $K$.

\begin{proof}[Proof of Theorem~\ref{exinc:main}]
Fix $t\geq1$ and $K\geq1$. In particular every constant $P_j(t,K)$
in the induction has the same first argument $t$. We prove the statement
by induction on the actual number $e$ of coordinates. At the induction
step, the assertion for every smaller actual dimension is available
for every grid and every curve satisfying its hypotheses.
Induct on $1\leq e\leq t$, simultaneously over all
\[
 m_i\geq1\ (i\in[e]),\quad(m_i)\text{ pairwise distinct},\quad
 \frac{\max_iq_i}{\min_iq_i}\leq K,\quad0<\theta\leq\tfrac12,
 \quad0\leq \ell\leq\theta m,
\]
and all stated closed sets and curves. Every inherited subtuple has
the same ratio bound by \eqref{exinc:grid-cardinality}.
The cases $C\subseteq V_+(X_0)$, $A=\varnothing$, $A=\Omega$, or $\ell=0$
give zero incidence; the last uses Lemma~\ref{exinc:zero-closure}.
The same holds if $C(k)\cap A^\circ=\varnothing$. Hence assume that
$C$ has dense affine part, $\varnothing\ne A\subsetneq\Omega$,
$\ell\geq1$, and $C(k)\cap A^\circ\ne\varnothing$.
In particular,
\begin{equation}\label{exinc:nontrivial-scale}
 m\theta\geq \ell\geq1.
\end{equation}

For $e=1$, $C=\mathbb P^1$ and Lemma~\ref{exinc:one-dimensional-closure} gives
\[
 \mu_{A^\circ}(C)\leq|A|\leq \ell\leq\theta m=P_1(t,K)\theta^{\gamma_1}m\deg C.
\]
For $e=2$ (hence $t\geq2$), a coordinate line either misses the grid, has a full slice
(contributing no point-stratum points), or has a proper $\ell$-closed slice
of size at most $\ell$. Assume both coordinates vary. Choose
$0\ne F\in I_\ell(A)$ by separation from any $z_0\in\Omega\setminus A$.
If $F|_C\ne0$, weighted B\'ezout, Lemma~\ref{exag:weighted-bezout}, gives
$\mu_{A^\circ}(C)\leq2\ell\deg C$.
Otherwise an irreducible factor $G\mid F$ defines the affine curve $C$:
its projective hypersurface closure is integral of dimension one,
and contains $C$, so equals $C$. Put
\[
 a=\deg_{X_1}G,\quad b=\deg_{X_2}G,\qquad1\leq a,b\leq \ell.
\]
Reorder the two coordinates so that $m_1<m_2$. Their lengths are
$Q+1,cQ+1$, with $Q=2^{m_1}$ and
$c=2^{m_2-m_1}=q_2/q_1\in[2,K]$. Lemma~\ref{exinc:pair-norm-incidence} gives
\[
 \mu_{A^\circ}(C)\leq(c+1)ab
 \leq(K+1)\ell\min\{a,b\}\leq(K+1)\ell\deg C.
\]
Since $2\leq K+1$, both cases satisfy the $e=2$ claim, including singular points.

Now let $3\leq e\leq t$. If $x_i|_C$ is constant, its value either
lies outside $S_i$ or defines a coordinate hyperplane slice.
The linear identification with $\mathbb P^{e-1}$ preserves degree and
multiplicity. The slice lemmas give closedness and inclusion of point strata,
and its actual minimum $m'$ satisfies
\[
 m\leq m'\leq Km,\qquad \ell\leq\theta m',\qquad
 \mu_{A^\circ}(C)\leq P_{e-1}(t,K)\theta^{\gamma_{e-1}}m'\deg C
 \leq KP_{e-1}(t,K)\theta^{\gamma_e}m\deg C.
\]
This is included in $P_e(t,K)$. Thus assume all coordinates vary and
$C\cap\Omega\ne\varnothing$. Lemmas~\ref{exinc:predecessor-distinct}
and~\ref{exinc:transform} give
\begin{equation}\label{exinc:common-trace}
 C'=T^{-1}(C)\ne C,\qquad Z=A^\circ,\qquad Z_C=Z\cap C=Z\cap C'.
\end{equation}

First establish the auxiliary bound, for integral $V\subseteq\mathbb P^e$
with $2\leq h=\dim V\leq e-1$, $c=e-h$, and $j=c+1$:
\begin{equation}\label{exinc:slicing-estimate}
 |Z\cap V|\leq S_e(t,K)\theta^{\gamma_j}m^h\deg V.
\end{equation}
At each component meeting $Z$ of dimension at least two, choose a
nonconstant unfixed coordinate and cut by all its grid hyperplanes.
Such a coordinate exists because otherwise its dense affine part is a point.
Each cut is proper, so Lemma~\ref{exag:external-intersection} gives
\[
 \sum_{W\in\operatorname{Irr}(V\cap V_+(X_i-\alpha X_0))}\deg W\leq\deg V,
 \qquad \dim W=\dim V-1,
 \qquad \sum_{\alpha\in S_i}\sum_W\deg W\leq Km\deg V.
\]
For this application of Lemma~\ref{exag:external-intersection},
the integral variety is $V$ and the nonzero restricted form is
$X_i-\alpha X_0$, of degree one. Only the distinct reduced components
of the intersection are counted; its scheme need not be reduced.
Discard components missing $Z$. After $h-1$ cuts, along distinct
coordinates on each branch, terminal curves $D_\lambda$ cover $Z\cap V$ and satisfy
\[
 \sum_\lambda\deg D_\lambda\leq(Km)^{h-1}\deg V,
 \qquad D_\lambda\subseteq\mathbb P^j\text{ in a coordinate slice}.
\]
Each slice has $m_J\in[m,Km]$, $\ell\leq\theta m_J$, and ratio at most $K$.
The lower-dimensional induction and slice lemmas therefore give
\[
 |Z\cap V|\leq\sum_\lambda\mu_{Z\cap D_\lambda}(D_\lambda)
 \leq P_j(t,K)\theta^{\gamma_j}Km\sum_\lambda\deg D_\lambda
 \leq K^hP_j(t,K)\theta^{\gamma_j}m^h\deg V.
\]
Since $j\leq e-1$ and $K^hP_j(t,K)\leq K^eP_{e-1}(t,K)=S_e(t,K)$,
this proves \eqref{exinc:slicing-estimate}. Reduced components suffice;
no reducedness of the intersection schemes is assumed.

Because $A$ is proper, choose $z_0\in\Omega\setminus A$.
Lemma~\ref{exinc:closure-slices}, applied to this $z_0$, gives a
polynomial $0\ne F\in I_\ell(A)$ with $F(z_0)\ne0$.
Since $A$ is nonempty, $F$ cannot be a nonzero constant.
Each individual degree is at most $\ell$, so
$1\leq\deg F\leq e\ell$. The possible first exits are
\[
 \begin{array}{c|c|c}
 \text{condition}&\text{polynomial nonzero on }C\text{ and zero on }Z_C&\text{degree bound}\\ \hline
 F|_C\ne0&F&e\ell\\
 F|_{C'}\ne0&\mathcal T(F)&e^2K\ell.
 \end{array}
\]
These use Lemmas~\ref{exinc:proper-transport} and~\ref{exag:external-intersection}.
Otherwise homogenize $F$ to $E_1$, which vanishes on both curves.

At stage $1\leq c\leq e-1$, maintain
\[
 \begin{gathered}
 X_c=V_+(E_1,\ldots,E_c),\quad B_c=\max_{1\leq r\leq c}\deg E_r,
 \quad h=e-c,\quad C,C'\subseteq X_c,\\
 \mathscr V_c=\{V\in\operatorname{Irr}(X_c):C\subseteq V\text{ or }C'\subseteq V\},
 \qquad \mathscr V_c\ne\varnothing,\qquad
 \forall V\in\mathscr V_c:\dim V=h.
 \end{gathered}
\]
The hypersurface theorem Lemma~\ref{exag:external-intersection}, applied first to
$\mathbb P^e$, verifies this for $c=1$. The component family is finite because $X_c$ is Noetherian. If $h=1$, the two
curves themselves are isolated expected components. If $h\geq2$, put $j=c+1$.
Every tracked $V$ is an isolated expected component of $X_c$.
Lemma~\ref{exag:expected-component-hilbert}, with
$(a,c,B,Y)=(e,c,B_c,V)$, therefore gives
\begin{equation}\label{exinc:tracked-hilbert}
 H_V(b)\geq\deg V\binom{b-cB_c+h}{h}\qquad(b\geq cB_c).
\end{equation}
Here the Hilbert function is the dimension of the indicated graded
coordinate piece, as in the definition of $H_V$ above.
Set $x=(h!S_e(t,K))^{1/h}m\theta^{\gamma_j/h}$ and
$b=cB_c+\lfloor x\rfloor+1$. Then
\[
 \binom{b-cB_c+h}{h}\geq\frac{(b-cB_c)^h}{h!}
 >S_e(t,K)m^h\theta^{\gamma_j},
\]
while $\gamma_j/h\leq1$, $m\theta\geq1$, and $L_e(t,K)=4e!S_e(t,K)$ imply
\[
 \lfloor x\rfloor+1\leq
 \bigl((h!S_e(t,K))^{1/h}+1\bigr)m\theta^{\gamma_j/h}
 \leq L_e(t,K)m\theta^{\gamma_j/h}.
\]
Together with \eqref{exinc:slicing-estimate}, this proves the common choice
\begin{equation}\label{exinc:locator-degree-choice}
 cB_c\leq b\leq cB_c+L_e(t,K)m\theta^{\gamma_j/h},\qquad
 \forall V\in\mathscr V_c:\ H_V(b)>|Z\cap V|.
\end{equation}
For each $V\in\mathscr V_c$, the map
\[
 \operatorname{ev}_{Z\cap V}:(k[X_0,\ldots,X_e]/I(V))_b
 \longrightarrow k^{Z\cap V},\qquad
 [P]\longmapsto(P(z)/X_0(z)^b)_z
\]
has nonzero kernel. Lift a nonzero kernel vector to $P_V$; then
\[
 P_V|_V\ne0,\quad P_V|_{Z\cap V}=0,
 \qquad Z_C\subseteq Z\cap V\quad\Longrightarrow\quad P_V|_{Z_C}=0.
\]
The last containment holds for both types of tracked components by
the common-trace identity. In particular every exit polynomial below
vanishes on the nonempty set $Z_C$ and therefore has positive degree.
The possible exits are now
\[
 \begin{array}{c|c|c}
 P_V|_C\ne0&P_V&b\\
 P_V|_{C'}\ne0&\mathcal T(P_V(1,X_1,\ldots,X_e))&eKb.
 \end{array}
\]
Otherwise all $P_V$ vanish on both curves. Define proper linear subspaces
\[
 L_V=\{\lambda\in k^{\mathscr V_c}:\textstyle\sum_W\lambda_WP_W|_V=0\}.
\]
They are proper since $P_V|_V\ne0$. Since a vector space over the infinite field $k$ is not a finite union of proper linear subspaces, choose
$\lambda\notin\bigcup_VL_V$ and set $E_{c+1}=\sum_W\lambda_WP_W$. Thus
\[
 \deg E_{c+1}=b,\qquad E_{c+1}|_C=E_{c+1}|_{C'}=0,
 \qquad\forall V\in\mathscr V_c:\ E_{c+1}|_V\ne0.
\]
To verify the next invariant, let $U\in\operatorname{Irr}(X_{c+1})$ contain
either curve. Choose $V\in\operatorname{Irr}(X_c)$ containing $U$, then
$W\in\operatorname{Irr}(V\cap V_+(E_{c+1}))$ containing $U$. Necessarily
\[
 V\in\mathscr V_c,\qquad\dim W=h-1,\qquad
 U\subseteq W\subseteq X_{c+1}\quad\Longrightarrow\quad U=W.
\]
The dimension equality is the proper-section assertion in
Lemma~\ref{exag:external-intersection}; the last implication is maximality in the
definition of an irreducible component.
Thus the new tracked components have dimension $h-1$. Unrelated
components impose no extra dimension requirement. With no exit, after
$e-1$ equations Corollary~\ref{exinc:expected-predecessor} gives
\[
 \mu_Z(C)\leq B_t(K)B_{e-1}\deg C.
\]

It remains to bound the terminal and exit degrees. For $2\leq j\leq e-1$, the inequality $r+1\leq2^r$ at $r=e-j\geq1$ gives
\[
 \frac{\gamma_j}{e-j+1}
 =\frac{2^{2-j}}{e-j+1}\geq2^{2-e}=\gamma_e.
\]
Equality holds at $j=e-1$, so
\[
 \min\left\{1,\min_{2\leq j\leq e-1}\frac{\gamma_j}{e-j+1}\right\}
 =\gamma_e.
\]
Put $M_e=m\theta^{\gamma_e}$. Since $0<\theta\leq1$,
\[
 B_1\leq e\ell\leq eM_e,\qquad
 b\leq cB_c+L_e(t,K)M_e\leq eB_c+L_e(t,K)M_e.
\]
The identical recurrence bounds every later equation and exit candidate:
\[
 B_c\leq\left(e^c+L_e(t,K)\sum_{u=0}^{c-2}e^u\right)M_e
 \leq H_e(t,K)M_e\qquad(c\leq e-1),
\]
with empty sum for $c=1$. An exit candidate is the next term of this
same recurrence and occurs by stage $e-1$. Weighted B\'ezout
Lemma~\ref{exag:weighted-bezout} therefore gives
\[
 \mu_Z(C)\leq
 \begin{cases}
 H_e(t,K)M_e\deg C,&\text{untransformed exit},\\
 eKH_e(t,K)M_e\deg C,&\text{transformed exit},\\
 B_t(K)H_e(t,K)M_e\deg C,&\text{terminal expected components}.
 \end{cases}
\]
Since $1,eK\leq B_t(K)$ by \eqref{exinc:constant-monotonicity} and
\[
 P_e(t,K)=B_t(K)H_e(t,K)\geq KP_{e-1}(t,K),
\]
all branches, including constant-coordinate slices, satisfy
$\mu_{A^\circ}(C)\leq P_e(t,K)\theta^{\gamma_e}m\deg C$.
The simultaneous quantifiers give the assertion on every nonempty coordinate subtuple.
\end{proof}

\section{Product expansion of Reed--Solomon codes}\label{exloc:chapter}
For a tuple of codes over a field $E$, write $A(C_1,\ldots,C_d)$ for
the directional sum from Section~\ref{sec:preliminaries}, and retain
$\ell_i(x_i)$ for the number of nonzero $i$-lines. Unless a coefficient
field is specified, geometric statements in this section are over
$k=\overline{\mathbb F}_2$.

\subsection{Affine and projective evaluation codes}
\begin{definition}[Reed--Solomon Codes]\label{exloc:rs-definition}
\label{exloc:generalized-rs-definition}\label{exloc:projective-rs-definition}
For a field $E$, a nonempty finite subset $S\subset E$ of size $n$, and
an integer $0\leq u\leq n$, set
\[
 \operatorname{RS}_E(S,u)=\{(f(a))_{a\in S}:f\in E[X],\ \deg f<u\}.
\]
For $u=0$ the code is zero. A generalized Reed--Solomon code is
$\operatorname{diag}(v_a)\operatorname{RS}_E(S,u)$ with
$v\in(E^\times)^S$. For a prime power $q$ such that $E$ contains $\mathbb F_q$, choose nonzero
representatives $r_x=(X_x,Y_x)\in\mathbb F_q^2$ for
$x\in\mathbb P^1(\mathbb F_q)$, and for integers $0\leq u\leq q$ set
\[
 \operatorname{PRS}_q(u)=\{(F(r_x))_x:F\in E[X,Y]_u\},
\]
where $E[X,Y]_u$ denotes the homogeneous forms of degree $u$.
\end{definition}

\begin{lemma}\label{rs:dimension-distance}
\label{exloc:rs-dual}
Let $E$ be a field and $S\subset E$ a nonempty finite set of size $n$.
For every integer $1\leq u\leq n$, the code $\operatorname{RS}_E(S,u)$ has dimension
$u$ and distance $n-u+1$. For integers $0\leq u\leq n$, put
$v_a=\prod_{b\in S\setminus\{a\}}(a-b)^{-1}$. Then
\[
 \operatorname{RS}_E(S,u)^\perp
 =\operatorname{diag}(v_a)\operatorname{RS}_E(S,n-u).
\]
\end{lemma}
\begin{proof}
A nonzero polynomial of degree at most $u-1$ has at most $u-1$ roots.
Thus evaluation is injective and the distance is at least $n-u+1$.
The polynomial $\prod_{a\in T}(X-a)$, for any $|T|=u-1$, attains that
weight. For duality, the leading coefficient in the Lagrange formula gives
\[
 \sum_{a\in S}v_a a^j=0\qquad(0\leq j\leq n-2).
\]
The products of polynomials of degrees below $u$ and $n-u$ have degree
at most $n-2$, proving orthogonality. Both spaces have the required
complementary dimensions. The cases $u=0,n$ are immediate.
\end{proof}

\begin{proposition}\label{exloc:norm-bridge}
Let $q$ be a prime power, let $E$ contain $\mathbb F_{q^2}$, choose
$\alpha\in\mathbb F_{q^2}\setminus\mathbb F_q$, and put $\beta=\alpha^q$.
The map
\[
 x=[X:Y]\longmapsto z_x=\frac{X-\alpha Y}{X-\beta Y}
\]
is a bijection from $\mathbb P^1(\mathbb F_q)$ onto
$S=\{z\in\mathbb F_{q^2}^\times:z^{q+1}=1\}$. For the chosen
representatives put $d_x=X_x-\beta Y_x\ne0$. For every $0\leq u\leq q$,
\[
 \operatorname{PRS}_q(u)=\operatorname{diag}(d_x^u)
                    \operatorname{RS}_E(S,u+1)
\]
after reindexing by $x\mapsto z_x$. Its dimension is $u+1$, its distance
is $q+1-u$, and its dual code is a generalized Reed--Solomon code
of dimension $q-u$. In the dual formula on $S$, $v_z=z/(q+1)$.
\end{proposition}
\begin{proof}
The two linear expressions in the ratio never vanish on a nonzero
vector over $\mathbb F_q$, and $z_x^q=z_x^{-1}$. The transformation has
inverse $z\mapsto[\beta z-\alpha:z-1]$, so it is injective. Both sets
have size $q+1$: $X^{q+1}-1$ is separable and all its roots lie in
$\mathbb F_{q^2}$ because $q+1$ divides $q^2-1$.
The polynomial-space isomorphism
\[
 E[X,Y]_u\longrightarrow E[z]_{\leq u},\qquad
 F\longmapsto\frac{F(\beta z-\alpha,z-1)}{(\beta-\alpha)^u}
\]
sends $F$ to $f$ with $F(r_x)=d_x^uf(z_x)$. The dimension, distance,
and dual assertions follow from Lemma~\ref{exloc:rs-dual} and
Lemma~\ref{cert:monomial-invariance}. Finally differentiation of
$\prod_{z\in S}(X-z)=X^{q+1}-1$ gives
$\prod_{w\ne z}(z-w)=(q+1)z^{-1}$.
\end{proof}

For integers $e\geq1$ define
\begin{equation}\label{exloc:arrangement-constants}
 F_e=2^e-1,\qquad \sigma_e=\sum_{j=2}^{F_e}j
 =2^{e-1}(2^e-1)-1,\qquad c_e=1+4^e\sigma_e.
\end{equation}
The empty sum is zero, so $c_1=1$. Fix $t\geq1$, $K\geq1$, and
$0<\epsilon\leq1/2$, and set
\begin{equation}\label{exloc:uniform-pe-parameters}
 \begin{split}
 \eta&=\frac{\epsilon}{2tc_tK},\qquad
 \theta=\min\left\{\frac12,
   \left(\frac{\eta}{P_t(t,K)}\right)^{1/\gamma_t}\right\},\\
 \rho&=\rho(t,K,\epsilon)=\frac1{2t}
           \left(\frac{\epsilon\theta}{4c_tK}\right)^t.
 \end{split}
\end{equation}
The functions $P_t(t,K)$ and $\gamma_t$ are defined in
\eqref{exinc:incidence-constants}--\eqref{exinc:constant-recurrence}.

\begin{theorem}\label{exloc:finite-grid-pe}
Fix an integer $t\geq1$ and real numbers $K\geq1$ and
$0<\epsilon\leq1/2$, and fix
$\eta,\theta,\rho$ by \eqref{exloc:uniform-pe-parameters}.
Let $d$ be an integer with $1\leq d\leq t$, choose pairwise distinct positive integers
$m_1,\ldots,m_d$, and put $q_i=2^{m_i}$. Assume
$\max_i q_i/\min_i q_i\leq K$. Define $n_i=q_i+1$ and
$S_i=\{z\in\overline{\F}_2^\times:z^{n_i}=1\}$ for every $i\in[d]$.
Let $E_0\subseteq\overline{\F}_2$ be a finite field containing
$\bigcup_{i=1}^d S_i$, and let $E/E_0$ be any field extension,
with the given sets embedded in $E$ by that inclusion.
For every tuple of integers $(k_i)_{i=1}^d$ satisfying
$\epsilon n_i\leq k_i\leq(1-\epsilon)n_i$ for every $i$, the
$E$-linear tuple $(\operatorname{RS}(S_i,k_i))_{i=1}^d$ has
weighted product expansion at least $\rho$.
For every nonempty $I\subseteq[d]$, the same constant applies to
the subtuple indexed by $I$ in increasing order; in each of its
factors one may independently take the dual code,
permute coordinates, and multiply coordinates by nonzero elements
of $E$. All these choices are made after $\rho$ has been fixed.
\end{theorem}

\subsection{Point interpolation from weighted incidence}
Let $a\geq1$, let $Z\subseteq\mathbb P^a(k)$ be finite and reduced,
and let $L$ be a line bundle on $\mathbb P^a$. For an integer $j\geq0$,
$L$ separates ordinary $j$-jets on $Z$ if
\[
 H^0(\mathbb P^a,L)\longrightarrow
 \bigoplus_{z\in Z}L_z/\mathfrak m_z^{j+1}L_z
\]
is onto. It separates $P$-Frobenius jets, for an integer $e\geq1$ and $P=2^e$, if the same
map with $\mathfrak m_z^{[P]}=(f^P:f\in\mathfrak m_z)$ in place of
$\mathfrak m_z^{j+1}$ is onto.

\begin{lemma}\label{exag:external-multipoint-jets}
\label{exag:external-multipoint-adjoint}
Let $a\geq2$, let $Z\subset\mathbb P^a(k)$ be finite with at least two
points, and let $\Lambda>0$ satisfy $\mu_Z(C)\leq\Lambda\deg C$
for all integral projective curves $C\subseteq\mathbb P^a_k$.
For each $0<\delta<1/\Lambda$, there is $d_0$ such that
$\mathcal O(d)$ separates $\lfloor\delta d\rfloor$-jets for $d\geq d_0$.
The threshold may depend on $Z$.

For any $a,l\geq1$ and nonempty finite reduced $Z\subset\mathbb P^a(k)$,
if $m,e\geq1$, $m<2^e-1$, and $\mathcal O(lm)$ separates
$2^e$-Frobenius jets on $Z$, then $\mathcal O(l-a-1)$ separates points
of $Z$.
\end{lemma}
\begin{proof}
For the first assertion, \cite[Definitions 5.1, 5.3 and Theorem 5.7]{DiPasqualeNguyenSeceleanu2023}
give
\[
 \lim_{d\to\infty}\frac{s_Z(d)}d
 =\inf_{\mu_Z(C)>0}\frac{\deg C}{\mu_Z(C)}\geq\Lambda^{-1},
\]
where $s_Z(d)$ is the largest simultaneously separated jet order.
The source's homogeneous quotients coincide with the stated local jets:
at $z=[1:0:\cdots:0]$, degree-$d$ monomials in
$k[X_0,\ldots,X_a]/(X_1,\ldots,X_a)^{j+1}$ dehomogenize to
$u^\alpha$ with $|\alpha|\leq j$, whenever $j\leq d$.
A line through two points gives $\Lambda\geq2$, so the required orders
$\lfloor\delta d\rfloor$ satisfy this inequality. The homogeneous
multiplicity at a point of a curve agrees with its local
Hilbert--Samuel multiplicity: in the chart $X_0\ne0$, localization at
the homogeneous point prime is the local curve ring after the residue
field extension $k\subset k(X_0)$, which preserves lengths of successive
maximal-ideal quotients. Thus the source's curve infimum has the stated
normalization. Its Section 5 has no characteristic-zero restriction.

For the second assertion we use the simultaneous-jet implication in the
Frobenius trace argument of \cite{MustataSchwede2014},
with $X=\mathbb P^a$, $p=2$, and $L=\mathcal O(l)$.
Here is the precise application of its amplification and trace argument.
Choose the largest $e_0\geq e$ for which $L^m$ separates
$2^{e_0}$-Frobenius jets; it exists because the target dimension
$|Z|2^{ae_0}$ is bounded by the fixed dimension of $H^0(L^m)$.
By the Frobenius-jet amplification lemma in that paper,
$L^{m_s}$ separates $P_s$-Frobenius jets for
\[
 P_s=2^{se_0},\qquad m_s=m\frac{P_s-1}{2^{e_0}-1},\qquad
 P_s-1-m_s\longrightarrow\infty.
\]
For large $s$, the bundle
$\omega_X\otimes L^{P_s-m_s}=\mathcal O(l(P_s-m_s)-a-1)$
is globally generated. Choose a section nonzero at the finite set $Z$.
Multiplication by it shows that $\omega_X\otimes L^{P_s}$ separates
$P_s$-Frobenius jets. The Frobenius trace diagram in the cited proof
then makes restriction of $\omega_X\otimes L$ onto $Z$ surjective.
The identity $\omega_{\mathbb P^a}=\mathcal O(-a-1)$
\cite[II, Example 8.20.1]{Hartshorne1977} yields the conclusion.
This application uses the proved simultaneous-jet implication, rather
than replacing the different one-point hypotheses of the numbered theorem.
\end{proof}

\begin{theorem}\label{exag:point-interpolation}
Let $a\geq1$, let $Z\subset\mathbb P^a(k)$ be finite and reduced, and
let $\Lambda>0$ satisfy $\mu_Z(C)\leq\Lambda\deg C$ for every integral
curve $C\subseteq\mathbb P^a$. If $b\geq0$ is an integer and
$b+a+1>a\Lambda$, then
\[
 H^0(\mathbb P^a,\mathcal O(b))\twoheadrightarrow H^0(Z,\mathcal O_Z(b)).
\]
In particular $b=\lceil a\Lambda\rceil$ suffices.
\end{theorem}
\begin{proof}
The cases $|Z|\leq1$ are immediate. If $a=1$, the hypothesis applied to
$\mathbb P^1$ gives $|Z|\leq\Lambda$, and the integer inequality
$b>\Lambda-2$ gives $b\geq|Z|-1$. The product of the linear forms
vanishing at all but one point, multiplied by a form nonzero throughout
$Z$ to reach degree $b$, interpolates that point.

Suppose $a\geq2$ and put $l=b+a+1>a\Lambda$. Choose
$a/l<\delta<1/\Lambda$. For $P=2^e$ set
\[
 m_e=\left\lceil\frac{a(P-1)}{l\delta}\right\rceil.
\]
For large $e$, $1\leq m_e<P-1$, $lm_e\geq d_0$, and
$\lfloor\delta lm_e\rfloor\geq a(P-1)$.
At a point of the smooth $a$-dimensional space $\mathbb P^a$,
\[
 \mathfrak m_z^{a(P-1)+1}\subseteq\mathfrak m_z^{[P]}:
\]
in local affine coordinates, each monomial of degree above $a(P-1)$
has some exponent at least $P$. Thus ordinary jet separation from
Lemma~\ref{exag:external-multipoint-jets} implies that
$\mathcal O(lm_e)$ separates $P$-Frobenius jets. The second assertion
of that lemma gives separation by $\mathcal O(l-a-1)=\mathcal O(b)$.
The configuration-dependent threshold $d_0$ affects the auxiliary
choice of $e$, not the asserted degree $b$.
\end{proof}

\subsection{Regularity and interpolation on coordinate-flat arrangements}
Write $X_I=\prod_{i\in I}\mathbb P^1_k$, $X_\varnothing=\operatorname{Spec}k$,
and $X_e=X_{[e]}$. For $v\in\mathbb Z^I$, put
$\mathcal O_{X_I}(v)=\bigotimes_{i\in I}\operatorname{pr}_i^*\mathcal O(v_i)$.
The ambient projective space of the Segre embedding of $X_e$ has
homogeneous coordinate ring $S_e=k[z_U:U\subseteq[e]]$. The coordinate
$z_U$ restricts to $\prod_{i\notin U}X_i\prod_{i\in U}Y_i$ on $X_e$;
the hyperplane bundle restricts to
$\mathcal O_{X_e}(1,\ldots,1)$ \cite[II, Exercise 4.9]{Hartshorne1977}.
Use graded regularity and saturated ambient ideals as in
Appendix~\ref{app:regularity}.

\begin{definition}\label{exloc:flat-arrangement}
\label{exloc:individual-interpolation}\label{exloc:segre-definition}
For $J\subseteq[e]$ and $b\in X_{[e]\setminus J}(k)$, a coordinate flat
is $X_J\times\{b\}$. A coordinate-flat arrangement is a reduced union
\[
 Y=\bigcup_{J\in\mathcal F}(X_J\times B_J)\subseteq X_e,
\]
where $\mathcal F\subseteq2^{[e]}$ and each $B_J$ is a finite reduced
subscheme of the indicated complementary product. A finite reduced
$B\subseteq X_I$ \emph{interpolates in individual degree $R$} if
\[
 H^0(X_I,\mathcal O((R)_{i\in I}))\twoheadrightarrow
 H^0(B,\mathcal O((R)_{i\in I})|_B).
\]
The definition is independent of fiber trivializations. For $I=\varnothing$
the only possibilities are the empty scheme and $\operatorname{Spec}k$.
\end{definition}

\begin{lemma}\label{exloc:segre-and-products}
For every $n\geq0$, $(S_e)_n\to H^0(X_e,\mathcal O(n,\ldots,n))$ is
onto. If $R\geq1$ and $B_J$ interpolates in individual degree $R$, then
\[
 \operatorname{reg}I_{X_J\times B_J}\leq\max(R+1,e).
\]
If $r,l\geq1$, $e=r+l$, and $Z\subseteq X_r$ is nonempty with nonzero
proper ambient ideal satisfying $\operatorname{reg}I_Z\leq a$, then
\[
 \operatorname{reg}I_{X_l\times Z}\leq\max(a,e).
\]
Also $\operatorname{reg}I_{X_e}\leq e$ when $I_{X_e}\ne0$.
\end{lemma}
\begin{proof}
For $0\leq a_i\leq n$, choose a chain $U_1\subseteq\cdots\subseteq U_n$
in which $i$ occurs $a_i$ times. The product of its Segre coordinates
restricts to $\prod_iX_i^{n-a_i}Y_i^{a_i}$, proving the first assertion.
These are all the multihomogeneous monomials of degree $(n,\ldots,n)$.
For a finite $B_J$, multiplication by a section of $\mathcal O(1,\ldots,1)$
nonzero at every point extends interpolation from degree $R$ to all
$n\geq R$. Tensoring with the free-coordinate sections proves ambient
restriction surjectivity onto $W=X_J\times B_J$ in those degrees.
By \eqref{exloc:external-projective-cohomology}--\eqref{exloc:external-kunneth},
\[
 H^u(W,\mathcal O_W(n))=0\quad(u>0,n\geq-1),\qquad
 H^u(W,\mathcal O_W(n))=0\quad(u>\dim W).
\]
Set $m=\max(R+1,e)$. The ideal sequence at twist $m-i$ has zero
$H^i(\mathcal I_W(m-i))$: for $i=1$ use restriction; for $i>1$ use
the displayed vanishings, since either $i-1>\dim W$ or $m-i\geq-1$.
Ambient intermediate cohomology vanishes, and its top group also
vanishes because the twist is at least minus its projective dimension.
Lemma~\ref{exloc:external-regularity} gives the bound. Taking $W=X_e$
and $m=e$ proves its asserted bound by the same calculation.

For $W=X_l\times Z$, set $m=\max(a,e)$. Restriction in degree $m-1$
is surjective by the first assertion and
Lemma~\ref{exloc:external-regularity} for $Z$. For higher ideal
cohomology set $j=i-1>0$ and $n=m-j-1$. If $n\geq0$, K\"unneth reduces
$H^j(W,\mathcal O_W(n))$ to
$H^0(X_l,\mathcal O(n,\ldots,n))\otimes H^j(Z,\mathcal O_Z(n))=0$,
since $n\geq a-j-1$. If $n=-1$, the $X_l$ factor has zero cohomology
in every degree. If $n\leq-2$, then $j\geq m+1>e\geq\dim W$.
The same ideal criterion proves the product bound. Empty and zero
ideals have their stated regularity conventions and are handled directly.
\end{proof}
\begin{lemma}\label{exloc:star-initial-ideal}
Let $e\geq1$ be an integer,
let $c\in X_e(k)$, and let $\varnothing\ne\mathcal F\subseteq2^{[e]}$.
Let $Y_c\subseteq X_e$ be the reduced union of the coordinate flats
through $c$ whose free sets belong to $\mathcal F$, and put
$\mathcal D=\bigcup_{J\in\mathcal F}2^J$.
After independent projective coordinate changes sending $c$ to
$([1:0],\ldots,[1:0])$, there exists a term order on
$S_e=k[z_U:U\subseteq[e]]$ for which
\begin{equation}\label{exloc:star-initial-formula}
 \operatorname{in}(I_{Y_c})
 =(z_U:U\notin\mathcal D)
 +(z_Uz_V:U,V\in\mathcal D,\ U\not\subseteq V,\ V\not\subseteq U).
\end{equation}
Consequently $\reg(S_e/I_{Y_c})\leq e+1$.
\end{lemma}
\begin{proof}
After the stated coordinate changes, give $z_U$ positive weight
$w_U=e^2+1-|U|^2$ and refine by a term order. For incomparable $U,V$,
\[
 w_U+w_V-w_{U\cap V}-w_{U\cup V}
 =2|U\setminus V|\,|V\setminus U|>0.
\]
If $U\cup V\in\mathcal D$, the relation
$z_Uz_V-z_{U\cap V}z_{U\cup V}$ vanishes on $Y_c$; otherwise
$z_Uz_V$ itself vanishes. Also $U\notin\mathcal D\Rightarrow z_U|_{Y_c}=0$.
Thus the right side $J_0$ of \eqref{exloc:star-initial-formula} satisfies
$J_0\subseteq\operatorname{in}(I_{Y_c})$.

Its standard monomials are precisely
\[
 z_{U_1}\cdots z_{U_n},\qquad U_1\subseteq\cdots\subseteq U_n\in\mathcal D
 \quad(n\geq1),
\]
and $1$ in degree zero. Choose a component whose free set contains $U_n$.
The selected chain restricts nontrivially there. Distinct surviving chains
restrict to distinct monomials: their coordinate occurrence counts recover
the nested sets by $i\in U_j\iff a_i\geq n-j+1$.
Thus the coefficient of that chain cannot cancel in any relation.
The standard-monomial basis theorem \cite[Theorem 15.3]{Eisenbud1995}
and this independence prove
$J_0=\operatorname{in}(I_{Y_c})$.

The ideal $J_0$ is squarefree. Its squarefree standard monomials are
strict chains in $\mathcal D\subseteq2^{[e]}$, hence have at most $e+1$
factors. Lemma~\ref{exloc:initial-betti-comparison}, with $I=I_{Y_c}$, $J=J_0$, and $s=e+1$, give
\[
 \reg(S_e/I_{Y_c})\leq\reg(S_e/J_0)\leq e+1.
\]
\end{proof}

\begin{lemma}\label{exloc:arrangement-constant-bounds}
For every integer $e\geq1$, the constants of
\eqref{exloc:arrangement-constants} satisfy
\[
 \sigma_e=\sum_{j=2}^{F_e}j=2^{e-1}(2^e-1)-1,
 \qquad c_e\leq c_{e+1}.
\]
An empty sum is zero. For every integer $e\geq2$,
\[
 4^e\bigl(\sigma_e-(F_e-1)(e-1)\bigr)\geq e+F_e.
\]
\end{lemma}
\begin{proof}
The recursion partitions the integers $2,\ldots,F_e$, so
\[
 \sigma_e=\sum_{j=2}^{F_e}j
 =\frac{F_e(F_e+1)}2-1=2^{e-1}(2^e-1)-1,
 \qquad c_e=1+4^e\sigma_e\leq c_{e+1}.
\]
For $e\geq2$, induction gives $2^e\geq2e$. With $F=2^e-1$,
\[
 F\geq e+1,\quad F+4-2e\geq3,\quad
 \sigma_e-(F-1)(e-1)=\frac{(F-1)(F+4-2e)}2\geq\frac32(F-1).
\]
Consequently
\[
 4^e\bigl(\sigma_e-(F-1)(e-1)\bigr)
 \geq24(F-1)\geq2F\geq e+F.
\]
\end{proof}

\begin{theorem}\label{exloc:flat-regularity}
Let $e,R$ be integers
with $e\geq1$ and $R\geq1$, and let
$\mathcal F\subseteq2^{[e]}$. For every $J\in\mathcal F$, let
$B_J\subseteq X_{[e]\setminus J}$ be a finite reduced closed
subscheme interpolating in individual degree $R$.
Let $Y=\bigcup_{J\in\mathcal F}(X_J\times B_J)\subseteq X_e$
with its reduced structure. In the standard graded ring
$S_e=k[z_U:U\subseteq[e]]$, its saturated homogeneous ideal satisfies
$\reg I_Y\leq c_e(R+1)$, where $c_e$ is the integer of
\eqref{exloc:arrangement-constants}.
\end{theorem}
\begin{proof}
Delete empty orientations and put $h=|\mathcal F|$. Induct lexicographically
on $(e,h)$, proving also, for $1\leq h\leq2^e-1$,
\[
 \reg(S_e/I_Y)\leq A_{e,h}(R),\qquad
 A_{e,h}(R)=\max(R,e+1)+\sum_{j=2}^h
       \bigl(1+4^e(\max(e,jR)-1)\bigr).
\]
The cases $h=0$, $h=1$, and $[e]\in\mathcal F$
follow from Lemma~\ref{exloc:segre-and-products}, including $Y=\varnothing,X_e$.
If $L=\bigcap_{J\in\mathcal F}J\ne\varnothing$, factor
\[
 Y=X_{|L|}\times\overline Y,\qquad e'=e-|L|<e,
 \qquad\overline Y=\bigcup_{J\in\mathcal F}(X_{J\setminus L}\times B_J).
\]
The indexing sets and distinct orientations retain their interpolation
data; the proper nonempty $\overline Y$ has nonzero proper ambient ideal.
Apply dimension induction and the product bound. The quantitative bound
for this branch is checked below. It remains to assume
\[
 2\leq h\leq2^e-1,\qquad\bigcap_JJ=\varnothing.
\]
Put $S=S_e$, $\mathfrak m=(z_U)_U$, $Y_J=X_J\times B_J$, and
\[
 M=S/I_Y,\qquad Y^{(J)}=\bigcup_{J'\ne J}Y_{J'},\qquad
 M^J=S/I_{Y^{(J)}}.
\]
The quotient map and its annihilator satisfy
\[
 M\twoheadrightarrow M^J,\qquad
 I_{Y_J}I_{Y^{(J)}}\subseteq\bigcap_{J'}I_{Y_{J'}}=I_Y,
 \qquad I_{Y_J}\ker(M\to M^J)=0.
\]
The common set $C=(\bigcap_JY_J)(k)$ is finite, since every coordinate
is fixed by at least one orientation. For $c\in C$, interpolation on
$B_J$, multiplication by nonvanishing degree-$R$ free-coordinate sections,
and Segre surjectivity supply $e_{J,c}\in S_R$ such that
\[
 e_{J,c}(c)\ne0,\qquad
 e_{J,c}|_{X_J\times\{b\}}=0\quad(b\in B_J,\ b\ne c_{[e]\setminus J}).
\]
Let $Y_c$ be the union of component flats through $c$, and put
\[
 h_c=\prod_{J\in\mathcal F}e_{J,c}\in S_{hR},\quad h_c(c)\ne0,
 \qquad M_c=S/I_{Y_c},\qquad h_c I_{Y_c}\subseteq I_Y.
\]
The last inclusion holds componentwise: $I_{Y_c}$ vanishes on components
through $c$, and one selector factor vanishes on every other component.
Reducedness of $Y$ turns this into the stated ideal inclusion. Hence
$(h_c)\ker(M\to M_c)=0$.

Define
\[
 \mathfrak b=\sum_JI_{Y_J}+\sum_{c\in C}(h_c),\qquad
 D_0=\max(e,hR),\qquad N=2^e.
\]
A common zero of the first family lies in $C$ and is excluded by its own
$h_c$. Thus $V_+(\mathfrak b)=\varnothing$. The product regularity bound
and generator-degree theorem give generators of the first ideals in degrees
$\leq\max(R+1,e)\leq D_0$; also $\deg h_c=hR\leq D_0$.
Multiplying lower-degree generators by forms nonzero at a specified point shows
that $\mathfrak b_{D_0}$ is basepoint-free.

Choose $f_1,\ldots,f_N\in\mathfrak b_{D_0}$ successively, each
avoiding the finitely many components of the current reduced projective
zero locus. This is possible by basepoint freeness and avoidance of a
finite union of proper linear subspaces over $k$. Every proper section
lowers positive component dimension by one
(Lemma~\ref{exag:external-intersection}); a section nonzero at a
zero-dimensional component removes it. Starting in $\mathbb P^{N-1}$,
after $N$ choices the common zero locus is empty. If it becomes empty
earlier, any nonzero remaining choices suffice. The parameter-ideal
bound, Lemma~\ref{ped:projective-parameter-power}, gives
\[
 \mathfrak m^\tau\subseteq(f_1,\ldots,f_N)\subseteq\mathfrak b,
 \qquad \tau=N(D_0-1)+1.
\]
All annihilator ideals are proper: $B_J\ne\varnothing$ and
$0\ne h_c\in S_{hR}$ with $hR>0$. All quotient maps have degree zero;
their finite target family is indexed by $\mathcal F\sqcup C$.
Empty $C$ omits only the second family.

If $b_0\geq0$ bounds every target regularity, Lemma~\ref{exloc:external-approximation}, with the quotient ideals
$I_{Y^{(J)}}$, $I_{Y_c}$ and annihilator ideals $I_{Y_J}$, $(h_c)$,
gives
\[
 \reg M\leq b_0+1+(\tau-1)N
 =b_0+1+4^e(\max(e,hR)-1).
\]
Since $\reg M_c\leq e+1$ by Lemma~\ref{exloc:star-initial-ideal}, induction gives
\[
 \reg M\leq A_{e,h-1}(R)+1+4^e(\max(e,hR)-1)=A_{e,h}(R).
\]
The common-free-coordinate branch also satisfies this recurrence:
\[
 \reg(S/I_Y)\leq\max\{A_{e',h}(R),e-1\}\leq A_{e,h}(R).
\]
For $F=2^e-1$, the arithmetic-series estimate and
Lemma~\ref{exloc:arrangement-constant-bounds} give
\[
 \begin{aligned}
 A_{e,F}(R)+1
 &\leq c_eR+e+F+1+4^e(F-1)(e-1)\\
 &\leq c_eR+1+4^e\sigma_e=c_e(R+1).
 \end{aligned}
\]
For nonzero proper $I_Y$, use $\reg I_Y=\reg(S/I_Y)+1$.
The exceptional ideals were handled initially; $e=1$ is the one-orientation case.
\end{proof}

\begin{theorem}\label{exloc:anisotropic-extension}
Let $e,d,R$ be integers
with $1\leq e\leq d$ and $R\geq1$, and choose
$\mathcal F\subseteq2^{[e]}$ and finite reduced indexing schemes
$B_J\subseteq X_{[e]\setminus J}$, for $J\in\mathcal F$, each
interpolating in individual degree $R$. Let
$Y=\bigcup_{J\in\mathcal F}(X_J\times B_J)$ have its reduced structure,
let $\mathcal I_{Y/X_e}$ be its ideal sheaf in $X_e$, and put
$b=c_d(R+1)-1$. For every tuple $v=(v_i)_{i=1}^e\in\mathbb Z^e$
with $v_i\geq b$ for every $i\in[e]$,
\[
 H^1(X_e,\mathcal I_{Y/X_e}\otimes\mathcal O_{X_e}(v))=0.
\]
Equivalently the restriction map
$H^0(X_e,\mathcal O_{X_e}(v))\to H^0(Y,\mathcal O_{X_e}(v)|_Y)$
is surjective. We abbreviate the tensor product sheaf by
$\mathcal I_{Y/X_e}(v)$.
\end{theorem}
\begin{proof}
Set $a=c_d(R+1)$, $b=a-1\geq0$. Empty $Y$ and $Y=X_e$ follow
from product cohomology or the zero ideal. Otherwise
\[
 \reg I_Y\leq c_e(R+1)\leq a
 \quad\Longrightarrow\quad
 (S_e)_b\twoheadrightarrow H^0(Y,\mathcal O_Y(b)).
\]
The implication is Lemma~\ref{exloc:external-ideal-regularity}
for the nonzero proper ambient ideal. It factors through sections on $X_e$;
the intrinsic ideal sequence and $H^1(X_e,\mathcal O(b,\ldots,b))=0$ give
\[
 H^1(X_e,\mathcal I_{Y/X_e}(b,\ldots,b))=0.
\]
Fix $d,R,b$ throughout. Induct on $e$ and then on
$\sum_i(v_i-b)$. For positive excess choose $i$ with $v_i>b$ and a fiber
$H=\{x_i=\beta\}$ avoiding every value fixed in coordinate $i$ by a component of $Y$.
On an affine chart write $A$ for the smooth ambient coordinate ring,
$I$ for the reduced ideal of $Y$, and $h$ for the fiber equation.
Write $I=\bigcap_\lambda P_\lambda$, where $P_\lambda$ are the
ideals of the coordinate flats present on this chart. Each $A/P_\lambda$
is a localization of a polynomial ring and is a domain.
If a flat fixes coordinate $i$, the choice of the fiber makes $h$
a nonzero constant on it. If that coordinate is free, $h$ is a
nonzero linear polynomial on it. Hence $h\notin P_\lambda$ for
every surviving component. If $hu\in I$, then
$(h+P_\lambda)(u+P_\lambda)=0$ in this domain for each $\lambda$,
so $u\in P_\lambda$ for every $\lambda$. Thus
\[
 hu\in I\ \Longrightarrow\ u\in I,\qquad I\cap hA=hI.
\]
Multiplication by $h$ is also injective on $A$, which is a domain.
The kernel of reduction $I\to(I+(h))/(h)$ is $I\cap hA=hI$;
reduction is surjective by the definition of $I+(h)$.
Consequently the local exact sequences
$0\to I\xrightarrow{h}I\to(I+(h))/(h)\to0$ glue, after twisting, to
\[
 0\longrightarrow\mathcal I_{Y/X_e}(v-\mathbf e_i)
 \longrightarrow\mathcal I_{Y/X_e}(v)
 \longrightarrow(i_H)_*\mathcal I_{Y\cap H/H}(\widehat v_i)
 \longrightarrow0.
\]
Tensoring with a line bundle is exact.
Here $\widehat v_i$ omits coordinate $i$. On the open set of the $i$th
factor avoiding fixed values, all components fixing that coordinate disappear;
the remaining arrangement is the product of that open set with its reduced
projection. To check this as a statement about schemes, use a product
affine chart with coordinate ring $B[t]$ and write the surviving flat
ideals as $J_\lambda B[t]$, where $t$ is the free $i$th coordinate.
Coefficientwise membership gives
$\bigcap_\lambda J_\lambda B[t]=(\bigcap_\lambda J_\lambda)B[t]$.
This identity remains valid after restricting to the chosen open set
because localization is exact. Setting $t=\beta$ therefore
gives the projected union with ideal $\bigcap_\lambda J_\lambda$.
This ideal is radical: if a power of an element belongs to each
$J_\lambda$, primality of each coordinate-flat ideal places the
element in each $J_\lambda$. Thus $Y\cap H$ is reduced
scheme-theoretically. A surviving orientation contains $i$ as a free
coordinate; deleting $i$ leaves exactly its original fixed-coordinate
indexing set $B_J$. In particular the parameter $R$ is unchanged.

A closed immersion preserves sheaf cohomology
\cite[III, Lemma 2.10]{Hartshorne1977}. Thus the third term has $H^1$
equal to $H^1(H,\mathcal I_{Y\cap H/H}(\widehat v_i))$. For $e>1$ it
vanishes by dimension induction with the same $d,R,b$. For $e=1$, $H$
is a point and its ideal sheaf is zero or $\mathcal O_H$, both with
zero $H^1$. Excess induction annihilates $H^1$ of the first term;
the cohomology exact sequence therefore annihilates $H^1$ of the middle.
Finally $v_i\geq b\geq0$ implies $H^1(X_e,\mathcal O(v))=0$, so this
vanishing is equivalent to restriction surjectivity.
\end{proof}

\begin{lemma}\label{exloc:scheme-gluing}
Let $e\geq1$ be an integer,
let $s\geq0$ be an integer, and let $Y_1,\ldots,Y_s\subseteq X_e$
be coordinate flats. Let $Y=\bigcup_{j=1}^s Y_j$ with reduced
structure and let $\mathcal L$ be an invertible $\mathcal O_Y$-module
sheaf. For every choice of sections
$s_j\in H^0(Y_j,\mathcal L|_{Y_j})$ that satisfy
$s_j|_{Y_j\cap Y_h}=s_h|_{Y_j\cap Y_h}$ for every $j,h\in[s]$,
there exists a unique $s_Y\in H^0(Y,\mathcal L)$ restricting to
$s_j$ on every $Y_j$. Each intersection here has its scheme-theoretic
structure inside $X_e$.
\end{lemma}

\begin{proof}
For $s=0$, the only section on the empty scheme is zero.
Assume $s>0$ and fix any scheme point $y\in Y$, including a
nonclosed point. Remove all flats not containing $y$ by taking the
complement of their closed union. Trivialize $\mathcal L$ on a
smaller neighborhood of $y$. Choose an affine product chart through
$y$ and restrict further to a principal open subset inside these
chosen neighborhoods. Its ambient coordinate ring is a localization
of a polynomial ring $A=k[t_1,\ldots,t_e]$.

For a fixed coordinate $i$, any two remaining flats that fix this
coordinate must fix it to the same value. Otherwise their defining
ideals would give two distinct constant values for $t_i$ at the same
prime ideal, whose difference is a nonzero scalar; a prime ideal
cannot contain a nonzero scalar. Translate $t_i$ by this common
value whenever that coordinate is fixed by some remaining flat.
The surviving flat ideals are now generated by subsets of the
variables, before localization. Call them $I_1,\ldots,I_q$.
Their intersections as schemes use the ideals $I_j+I_h$.

For ideals generated by monomials, membership is checked monomial
by monomial, since the monomials form a vector-space basis of $A$.
Sums correspond to unions of the sets of contained monomials, and
intersections correspond to intersections of these sets. The
set-theoretic distributive law therefore gives
\[
 (I_1\cap\cdots\cap I_{r-1})+I_r
 =\bigcap_{j<r}(I_j+I_r)\qquad(2\leq r\leq q).
\]
Localization preserves finite ideal intersections and sums, hence these equalities.

For any two ideals $I,J$ in a ring, residues $a+I,b+J$ agree modulo
$I+J$ exactly when $a-b=u+v$ for some $u\in I$, $v\in J$.
Then $a-u=b+v$ is a common lift. Two common lifts differ by
$I\cap J$. Thus
\[
 A/(I\cap J)\xrightarrow{\sim}
 \{(a+I,b+J):a-b\in I+J\}.
\]
Induct on $r$. Lift the first $r-1$ compatible residues to $a$.
For the last residue $b+I_r$, pairwise compatibility says
$a-b\in\bigcap_{j<r}(I_j+I_r)$.
The distributive identity identifies this ideal with
$(\bigcap_{j<r}I_j)+I_r$. The two-ideal result gives a common lift
for all $r$ residues, unique modulo $\bigcap_{j\leq r}I_j$.
It follows that the local compatible tuple determines exactly one
section on the reduced union, whose ideal is $\bigcap_j I_j$.

The construction applies on a neighborhood of every scheme point
$y$, so these neighborhoods form an open cover. On the intersection
of two neighborhoods the two constructed sections have the same
restrictions to every flat. Local uniqueness makes
them equal. The sheaf gluing axiom therefore gives a unique
global section of $\mathcal L$. This is the section in the statement.
\end{proof}

\subsection{Line interpolation and supported decompositions}
The closure and point-stratum definitions of
Definition~\ref{exinc:grid-definition} apply to any finite Cartesian grid
$\Omega=\prod_{i=1}^dS_i\subset E^d$ over any field $E$, using polynomials
over $E$. For no factors the grid is a singleton, the polynomial space
is $E$, and the point stratum of a subset is that subset.

\begin{lemma}\label{exloc:locator-closure}
\label{exloc:locator-scalar-extension}\label{exloc:locator-size}
Let $E$ be a field, let $d\geq1$, and let $S_i\subset E$ be nonempty
finite sets of sizes $n_i$. For any $\ell\geq0$ and $W\subseteq\Omega$,
$\operatorname{cl}_\ell(W)$ is unchanged by extension of $E$.
If $0\leq\ell<\min_i n_i$, then, with $N=\prod_i n_i$,
\[
 |\operatorname{cl}_\ell(W)|\leq\frac{N}{(\ell+1)^d}|W|.
\]
In particular $|W|<(\ell+1)^d$ implies proper closure.
\end{lemma}
\begin{proof}
Write $e_z$ for evaluation at $z$ on the space of box-degree-$\ell$
polynomials. A point belongs to the closure precisely when
$e_z\in\langle e_w:w\in W\rangle_E$. This row-span condition is
unchanged by extension of the field, by Lemma~\ref{cert:scalar-extension}.
For the bound put $D=(\ell+1)^d$ and
$L=\langle e_w:w\in W\rangle_E$, so $\dim L\leq|W|$.
On every subgrid $\prod_i T_i$ with $|T_i|=\ell+1$, tensor Lagrange
interpolation shows that its $D$ evaluation rows are a basis. At most
$\dim L$ of these rows belong to $L$. Average over all such subgrids.
Each grid point occurs with probability $D/N$, giving
$(D/N)|\operatorname{cl}_\ell(W)|\leq\dim L\leq|W|$.
\end{proof}

\begin{lemma}\label{exloc:maximal-flat-indexing}
Let $d\geq1$, let $S_i\subset E$ have size at least two, and let
$T\subsetneq\prod_i S_i$ be $\ell$-closed. For $J\subseteq[d]$ put
\[
 S_J=\prod_{i\in J}S_i,\qquad
 F_J=\{b\in S_{[d]\setminus J}:S_J\times\{b\}\subseteq T\},
 \qquad B_J=F_J^\circ.
\]
Then $F_J$ is $\ell$-closed on its complementary grid. The pairs
$(J,b)$ with $b\in B_J$ index exactly the maximal full coordinate flats
in $T$. Embedding the factors by $a\mapsto[a:1]$, their projective
union has grid trace $T$. Moreover $B_{[d]}=\varnothing$.
\end{lemma}
\begin{proof}
If $b\notin F_J$, choose $a\in S_J$ with $(a,b)\notin T$ and a
box-degree-$\ell$ polynomial $P$ vanishing on $T$ but not at $(a,b)$.
The polynomial $P(a,-)$ vanishes on $F_J$ and separates $b$.
For $i\notin J$, the full $i$-line through $b$ lies in $F_J$ precisely
when the flat $S_J\times\{b\}$ extends to a flat with free set
$J\cup\{i\}$. Since all factor sizes exceed one, every strict
containment of full flats strictly enlarges the free set. Thus
$F_J^\circ$ indexes maximal flats. Every point of $T$ extends to one
by successively freeing coordinates. Each projective flat has exactly
its original grid trace, proving the union claim. Properness excludes
the full orientation $J=[d]$.
\end{proof}
\begin{theorem}\label{exloc:line-interpolation}
Let $E\subseteq k=\overline{\mathbb F}_2$ be a finite subfield and let
$d\geq1$ be an integer. For every $i\in[d]$, let
$S_i\subseteq E$ be a finite set of size $n_i\geq2$, and put
$\Omega=\prod_{i=1}^d S_i$. Let $\ell\geq0$ be an integer and let
$K,\epsilon,\eta$ be real numbers with $K\geq1$, $\epsilon>0$,
and $\eta>0$. Put $m=\min_{i\in[d]}n_i$ and assume
$\max_{i\in[d]}n_i\leq Km$.
For every nonempty $I\subseteq[d]$, put
$\Omega_I=\prod_{i\in I}S_i$ and $m_I=\min_{i\in I}n_i$.
Embed $\Omega_I$ in the affine chart of $\PP_k^{|I|}$ by its
coordinates in increasing index order. Assume that, for every such
$I$, every $\ell$-closed $A\subseteq\Omega_I$, and every integral
projective curve $C\subseteq\PP_k^{|I|}$,
\begin{equation}\label{exloc:incidence-premise}
 \sum_{z\in C(k)\cap A^\circ}\mult_z C\leq\eta m_I\deg C.
\end{equation}
Closure is computed over $E$; it is unchanged over $k$ by
Lemma~\ref{exloc:locator-scalar-extension}, with $L=k$ and the
specified integer $\ell\geq0$. Degree and multiplicity are those of
the degree and multiplicity conventions of Section~\ref{exinc:chapter}.
Assume additionally
\[
 \eta\leq\frac{\epsilon}{2dc_dK},\qquad
 m\geq\frac{4c_d}{\epsilon},
\]
where $c_d$ is given by \eqref{exloc:arrangement-constants}.
For every proper $\ell$-closed subset $T\subsetneq\Omega$, every tuple
of integers $(r_i)_{i=1}^d$ with $\epsilon n_i\leq r_i\leq n_i$
for all $i$, and every function $f:T\to E$ whose restriction to
each complete $i$-line contained in $T$ belongs to
$\operatorname{RS}(S_i,r_i)$, there exists
$P\in E[X_1,\ldots,X_d]$ such that $\deg_{X_i}P<r_i$ for every
$i\in[d]$ and $P(x)=f(x)$ for every $x\in T$.
\end{theorem}
\begin{proof}
Let $U_E\leq E^T$ be the kernel of all dual line-constraint matrices, and
$V_E\leq E^T$ the image of evaluation of $\bigotimes_iE[X_i]_{<r_i}$.
Then $V_E\subseteq U_E$. Both spaces are defined by matrices over $E$;
Lemma~\ref{cert:scalar-extension} reduces $U_E=V_E$ to $U_k=V_k$.
Fix $f\in U_k$.

Form $B_J=F_J^\circ$ by Lemma~\ref{exloc:maximal-flat-indexing}. For nonempty
$B_J$, put $I=[d]\setminus J$, $a=|I|\geq1$; the full orientation is absent.
The incidence premise for the $\ell$-closed $F_J$ gives
\[
 \mu_{B_J}(C)\leq\Lambda\deg C\quad\text{for every integral }C\subseteq\mathbb P^a_k,
 \qquad\Lambda=\eta m_I\leq\eta Km.
\]
Theorem~\ref{exag:point-interpolation}, with $k=\overline{\mathbb F}_2$
and these $a,B_J,\Lambda$, gives interpolation in total degree $\lceil a\Lambda\rceil$.
After separate homogenization it gives individual degree
\[
 R=\max\{1,\lceil dK\eta m\rceil\}\leq dK\eta m+1
\]
on every indexing set. Empty indexing sets impose no condition.

On each maximal flat $S_J\times\{b\}$, all line constraints imply
\[
 f|_{S_J\times\{b\}}\in\bigotimes_{i\in J}\operatorname{RS}_k(S_i,r_i)
\]
because the simultaneous directional kernels intersect in the tensor product of the factor kernels (choose a complement to each code); for $J=\varnothing$ this is its single scalar value. Homogenize the polynomial in every free coordinate to degree $r_i-1$.
On the fixed coordinates, use $Y_i^{r_i-1}$, which is nonzero at
$[b_i:1]$. This defines a section of
$\mathcal O_{X_d}(r_1-1,\ldots,r_d-1)$ on each projective flat,
with the specified values in the affine trivialization $Y_i=1$.
The multihomogeneous monomial basis gives these section formulas.
On every nonempty intersection their grid values agree. Since
\[
 r_i\leq n_i\quad\Longrightarrow\quad
 \bigotimes_{i\in J'}k[X_i]_{<r_i}\longrightarrow k^{\prod_{i\in J'}S_i}
 \text{ is injective},
\]
they agree as sections on the intersection flat, using
Lemma~\ref{rs:dimension-distance} in each free coordinate. Lemma~\ref{exloc:scheme-gluing}
therefore gives a section on the reduced union $Y$.

The quantitative hypotheses give, for every $i$,
\[
 \begin{aligned}
 c_d(R+1)-1
 &\leq c_d(dK\eta m+2)-1
 \leq\epsilon m/2+2c_d-1\\
 &\leq\epsilon m-1\leq r_i-1.
 \end{aligned}
\]
Apply Theorem~\ref{exloc:anisotropic-extension} over $k$, with both dimensions
equal to $d$, this $R$, and $v_i=r_i-1$. It extends the section to $X_d$;
dehomogenization gives an element of $\bigotimes_i k[X_i]_{<r_i}$ restricting
to $f$. Thus $U_k=V_k$, and finite-matrix descent yields $U_E=V_E$.
\end{proof}

For a code $C_i\leq E^{S_i}$ and an $i$-line $L\subseteq\Omega$, write
$C_i[L]\leq E^\Omega$ for the copy of $C_i$ supported on $L$. Identify
$E^T$ with the subspace of $E^\Omega$ supported on $T$.
\begin{theorem}\label{exloc:supported-lines}
Under all grid, incidence, and numerical hypotheses of
Theorem~\ref{exloc:line-interpolation}, let $r_i$ be integers with
$\epsilon n_i\leq r_i\leq n_i$, and suppose
$C_i^\perp=\operatorname{RS}_E(S_i,r_i)$. For every proper
$\ell$-closed $T\subsetneq\Omega$,
\[
 A(C_1,\ldots,C_d)\cap E^T=\sum_{i=1}^d\sum_{\substack{L\text{ an }i\text{-line}\\L\subseteq T}}C_i[L].
\]
Every $x$ in this space has $x=\sum_i x_i$, $x_i\in A_i$, with
$\sum_i n_i\ell_i(x_i)\leq d|T|$. Both conclusions persist under
independent nonzero coordinate multipliers in the factors.
\end{theorem}
\begin{proof}
Write $A=A(C_1,\ldots,C_d)$ and denote the right-hand internal-line
sum by $D_T\subseteq A\cap E^T$. Its orthogonal complement inside $E^T$ is
exactly the set of functions whose restriction to every internal
$i$-line belongs to $C_i^\perp$. By
Theorem~\ref{exloc:line-interpolation}, this is the restriction to $T$
of $\bigotimes_i C_i^\perp=A^\perp$
(Lemma~\ref{cert:product-check}). Thus
\[
 D_T^\perp=\pi_T(A^\perp),\qquad
 D_T=(\pi_T(A^\perp))^\perp=A\cap E^T.
\]
The last equality follows directly by pairing a zero extension from
$E^T$ with elements of $A^\perp$. Group an internal-line decomposition
by direction. Disjointness of parallel lines bounds their number by
$|T|/n_i$ in direction $i$, proving the cost bound. The tensor product
of the specified invertible diagonal maps preserves supports, each
line, and its nonzero status, and transports both subspaces and the
decomposition.
\end{proof}

\subsection{Uniform product expansion and its variants}
\begin{proof}[Proof of Theorem~\ref{exloc:finite-grid-pe}]
The constants $\eta,\theta,\rho$ have already been chosen from
$t,K,\epsilon$. We now fix arbitrary data $d,(m_i),E_0,E,(k_i)$
from the theorem statement; none of the constants will be changed.
Every nonempty coordinate subtuple has power ratio at most $K$, and
its length ratio satisfies
\[
 \frac{\max_{i\in I}n_i}{\min_{i\in I}n_i}
 =\frac{\max_{i\in I}q_i+1}{\min_{i\in I}q_i+1}
 \leq\frac{\max_{i\in I}q_i}{\min_{i\in I}q_i}\leq K.
\]
Theorem~\ref{exinc:main}, with the fixed dimension bound $t$
and field $\overline{\mathbb F}_2$, gives for every nonempty subtuple of size
$e\leq d$ and every integer $0\leq\ell\leq\theta m_I$,
\[
 \mu_{A^\circ}(C)\leq P_e(t,K)\theta^{\gamma_e}m_I\deg C
 \leq P_t(t,K)\theta^{\gamma_t}m_I\deg C
 \leq\eta m_I\deg C,
\]
since $\gamma_e\geq\gamma_t$, $P_e(t,K)\leq P_t(t,K)$, and $0<\theta\leq1/2$.
These are all complementary-subgrid incidence premises of
Theorem~\ref{exloc:line-interpolation}.

Use locator closure over $E_0$, with algebraic closure $\overline{\mathbb F}_2$.
Put
\[
 m=\min_i n_i,\quad N=\prod_i n_i,\quad \ell=\lfloor\theta m\rfloor,
 \qquad C_i^0=\operatorname{RS}_{E_0}(S_i,k_i),\quad r_i=n_i-k_i.
\]
Then $\ell\leq\theta m_I$ on every nonempty complementary grid. By the RS dual formula,
\[
 v_{i,a}=\prod_{b\in S_i\setminus\{a\}}(a-b)^{-1},\quad
 D_i^0=\operatorname{diag}(v_{i,a})C_i^0,
 \qquad (D_i^0)^\perp=\operatorname{RS}_{E_0}(S_i,r_i),
 \quad\epsilon n_i\leq r_i\leq(1-\epsilon)n_i.
\]
Set $D_i=D_i^0\otimes_{E_0}E$. The dual identity extends as a kernel
identity; support isometry reduces the desired expansion to $(D_i)$.

First assume $m\geq4c_t/(\epsilon\theta)$. Then $1\leq \ell<n_i$ and
$m\geq4c_t/\epsilon\geq4c_d/\epsilon$. Moreover
\[
 \eta=\frac{\epsilon}{2tc_tK}\leq\frac{\epsilon}{2dc_dK},
\]
so all numerical hypotheses of Theorem~\ref{exloc:line-interpolation} hold.
Fix a nonzero word $x\in A(D_1,\ldots,D_d)$.
Its support is a subset of the original finite grid even when its
nonzero values lie in the larger field $E$. Thus its locator closure
can be computed using the finite-field evaluation matrix over $E_0$.
For this $x$, put
$W=\supp x$ and $T=\operatorname{cl}_\ell(W)$. In the small-support range,
\[
 |W|<\frac{\theta^dN}{2K^d},\quad N\leq K^dm^d
 \quad\Longrightarrow\quad |W|<(\theta m)^d<(\ell+1)^d
 \quad\Longrightarrow\quad T\subsetneq\Omega.
\]
Theorem~\ref{exloc:supported-lines}, with $E=E_0$, the displayed
RS-dual dimensions, and this $E_0$-defined closed set $T$, proves
\[
 A(D_1^0,\ldots,D_d^0)\cap E_0^T
 =\sum_{i,L\subseteq T}D_i^0[L].
\]
The displayed equality concerns subspaces specified by finite matrices
over $E_0$. Tensoring those matrices with $E$ preserves their kernels,
images, sums, and intersections by Lemma~\ref{cert:scalar-extension}.
This is why the fact that $x$ might have values outside $E_0$ causes
no difficulty. Extend this finite-subspace identity before decomposing $x$:
\begin{gather*}
 A(D_1,\ldots,D_d)\cap E^T
 =\sum_{i,L\subseteq T}D_i[L],\qquad
 x=\sum_i x_i,\\
 \sum_i n_i\ell_i(x_i)\leq d|T|
 \leq\frac{dN}{(\ell+1)^d}|W|\leq d(K/\theta)^d|W|.
\end{gather*}
Here Lemma~\ref{cert:scalar-extension} preserves images, intersections, and finite sums;
Lemma~\ref{exloc:locator-size} supplies the closure bound over $E_0$.
By definition of the directional sum, $x$ has at least one directional
decomposition. In each direction there are exactly $N/n_i$ parallel
lines, so every decomposition has cost at most
$\sum_i n_i(N/n_i)=dN$. In the other support range this gives
\[
 |x|\geq\frac{\theta^dN}{2K^d},\qquad
 \sum_i n_i\ell_i(x_i)\leq dN
 \quad\Longrightarrow\quad
 \frac{\theta^d}{2dK^d}\sum_i n_i\ell_i(x_i)\leq|x|.
\]
If instead $m<4c_t/(\epsilon\theta)$, then every $n_i\leq Km$.
There is no need to use the geometric interpolation theorem in this case.
Every nonzero scalar word has at least one nonzero coordinate, and any
directional decomposition has
\[
 \sum_i n_i\ell_i(x_i)\leq dN
 <d\left(\frac{4c_tK}{\epsilon\theta}\right)^d,\qquad |x|\geq1.
\]
The zero word uses zero summands. Thus all ranges are covered by
\[
 \rho_{d;t}=\frac1{2d}\left(\frac{\epsilon\theta}{4c_tK}\right)^d
 \leq\frac{\theta^d}{2dK^d},\qquad
 0<\alpha:=\frac{\epsilon\theta}{4c_tK}\leq1.
\]
For $1\leq d<t$, one has
\[
 \frac{\rho_{d+1;t}}{\rho_{d;t}}=\frac{d}{d+1}\alpha\leq1.
\]
Including $t=1$, when the minimum has one term, this gives
\[
 \min_{1\leq d\leq t}\rho_{d;t}=\rho_{t;t}=\rho.
\]
This constant was fixed before $(m_i),E_0,E,(k_i)$ and all subtuple choices.
The argument uses incidence only over $\overline{\mathbb F}_2$ and extends
the supported-subspace identity to the arbitrary field $E$; it therefore
proves the claimed field uniformity. Complementary dimensions remain in
$[\epsilon n_i,(1-\epsilon)n_i]$. The RS dual formula and
Lemma~\ref{cert:monomial-invariance} give every independent actual-dual,
permutation, multiplier, and nonempty-subtuple variant with the same $\rho$.
\end{proof}

\section{Equivariant Reed--Solomon codes and compatible cosheaves}
\label{sec:equivariance}\label{exdiag:chapter}

Fix integers $t,h\ge1$, real numbers $K,\epsilon,\beta$ with
\[
 K\ge1,\qquad 0<\epsilon<\frac1{2^h+1}<\beta\le1-\epsilon,
\]
and pairwise distinct integers $a_i\ge0$ for $i\in[t]$ such that
$2^{\max_i a_i-\min_i a_i}\le K$. Write
\begin{equation}\label{exloc:reciprocal-parameters}
 s_h=2^h,\quad D=s_h+1,\quad
 Q_{\mathrm{end}}(h,\epsilon,\beta)
 =\max\left\{2,\frac{1+D\epsilon}{1-D\epsilon},
                  \frac{D(s_h-\beta)}{D\beta-1}\right\}.
\end{equation}
Both denominators are positive. The product-expansion constant
$\rho=\rho(t,K,\epsilon)$ is the one fixed in
\eqref{exloc:uniform-pe-parameters}, before any local field sizes
are chosen.

\begin{theorem}
\label{exdiag:actual-direction-system}
Let $r\ge1$ satisfy $r+a_i\ge h$ and
$q_i=2^{r+a_i}\ge Q_{\mathrm{end}}(h,\epsilon,\beta)$ for all
$i\in[t]$. Put $n_i=q_i+1$ and
\[
 e_i=\left\lfloor\frac{q_i-1}{D}\right\rfloor,\qquad
 d_i=q_i-1-s_h e_i.
\]
Choose a finite subfield $E_0\subseteq\overline{\mathbb F}_2$
containing every $\mathbb F_{q_i^2}$, a finite extension $E/E_0$,
and a subset $L\subseteq[t]$. At every sufficiently deep level
of the tower in Theorem~\ref{exar:tower}, there is a cosheaf
$\mathcal F_\nu$ over $E$ on $X_\nu$ with compatible base codes
in the sense of Definition~\ref{def:compatible-base-codes} and
endpoint dimensions
\begin{equation}\label{exdiag:endpoint-dimensions}
 (m_{i,0},m_{i,1})=
 \begin{cases}
 (e_i+1,d_i+1),&i\in L,\\
 (n_i-e_i-1,n_i-d_i-1),&i\notin L.
 \end{cases}
\end{equation}
For $i\in L$ both endpoint rates lie in $[\epsilon,\beta]$;
for $i\notin L$ they lie in $[1-\beta,1-\epsilon]$.
Every endpoint code and its dual code have
relative distance at least $\epsilon$.

For every face $f$, every nonempty $J\subseteq\codir(f)$,
every field extension $E'/E$, and every independent choice
\[
 D_i\in\left\{E'\otimes_E C_{f,i},
                   (E'\otimes_E C_{f,i})^\perp\right\}
       \quad(i\in J),
\]
the tuple $(D_i)_{i\in J}$ has weighted product expansion at
least $\rho$. Here $C_{f,i}=\operatorname{im}E_{f,i}$ is the
actual code in the compatible chart at $f$, and orthogonal complements
use the coordinate pairing on $(E')^{\Omega_{f,i}}$.
The field $E$, all endpoint dimensions, and $\rho$ are
independent of the tower index $\nu$.
\end{theorem}

The proof uses representations of the local endpoint stabilizers,
followed by tensor products over $E$ and descent to the quotient.
We first construct a pair of evaluation codes whose characters
on the common edge stabilizer agree.

\subsection{Reciprocal projective evaluation representations}

Let $q=2^\ell$ with $\ell\ge1$, and let $E\supseteq\mathbb F_q$
be a field. For $j\in\mathbb Z$, put
\begin{equation}\label{exloc:homogeneous-action}
 \mathcal H_j=\{F:\mathbb F_q^2\setminus\{0\}\to E:
       F(av)=a^jF(v)\ (a\in\mathbb F_q^\times)\},\qquad
 (gF)(v)=F(g^{-1}v)
 \quad(g\in\mathrm{SL}_2(\mathbb F_q)).
\end{equation}
Choose $r_x\in\mathbb F_q^2\setminus\{0\}$ representing every
$x\in\mathbb P^1(\mathbb F_q)$. Evaluation at these representatives
identifies $\mathcal H_j$ with $E^{\mathbb P^1(\mathbb F_q)}$.

\begin{lemma}\label{exloc:homogeneous-representation}
\label{exloc:infinity-character}
The action in \eqref{exloc:homogeneous-action} is an $E$-linear
representation, monomial in the evaluation coordinates. Its
space and action depend only on $j\bmod(q-1)$. If
\[
 B_q=\left\{\begin{pmatrix}c&z\\0&c^{-1}\end{pmatrix}:
          c\in\mathbb F_q^\times,\ z\in\mathbb F_q\right\},
 \qquad
 \chi_j\begin{pmatrix}c&z\\0&c^{-1}\end{pmatrix}=c^{-j},
\]
then evaluation at $v_\infty=(1,0)$ is a $B_q$-equivariant
map $\mathcal H_j\to E_{\chi_j}$.
\end{lemma}
\begin{proof}
The inverse of evaluation is the map sending a coordinate list
$(u_x)_x$ to $F(ar_x)=a^ju_x$; the expression $ar_x$ is unique.
For $g,g'\in\mathrm{SL}_2(\mathbb F_q)$,
\[
 (gF)(av)=a^j(gF)(v),\qquad
 g(g'F)(v)=F((gg')^{-1}v).
\]
If $g^{-1}r_x=c_{g,x}r_{g^{-1}x}$, the coordinate formula is
$(gF)(r_x)=c_{g,x}^{j}F(r_{g^{-1}x})$, with $c_{g,x}\ne0$.
This proves the representation and monomial assertions. Since
$a^{q-1}=1$ for $a\in\mathbb F_q^\times$, congruent exponents
give identical homogeneity conditions. Finally, the upper-left
entry of a product in $B_q$ is the product of the upper-left
entries, so $\chi_j$ is a character. For $g\in B_q$ with that
entry equal to $c$, one has $g^{-1}v_\infty=c^{-1}v_\infty$,
and hence $(gF)(v_\infty)=c^{-j}F(v_\infty)$.
\end{proof}

For integers $\ell\ge h$, put
\begin{equation}\label{exloc:reciprocal-pair}
 \begin{gathered}
 q=2^\ell,\quad n=q+1,\quad M=q-1,\quad u=2^{\ell-h},\\
 e=\lfloor M/D\rfloor,\quad R=M-De,\quad d=M-s_he=e+R,\\
 U=\{(P(r_x))_x:P\in E[X,Y]_e\},\qquad
 V=\{(P(r_x^{[u]}))_x:P\in E[X,Y]_d\}.
 \end{gathered}
\end{equation}
Here $E[X,Y]_j$ is the homogeneous degree-$j$ space and
$(X,Y)^{[u]}=(X^u,Y^u)$. The functions underlying $V$ raise
the arguments, not the polynomial coefficients, to the $u$th
power. Thus both evaluation maps in \eqref{exloc:reciprocal-pair}
are $E$-linear.

\begin{theorem}\label{exloc:reciprocal-prs}
If $\ell\ge h$, $q\ge Q_{\mathrm{end}}(h,\epsilon,\beta)$,
and $E\supseteq\mathbb F_q$ is finite, then the spaces $U,V$
in \eqref{exloc:reciprocal-pair} are invariant in
$\mathcal H_e,\mathcal H_{ud}$, respectively. Their evaluation
characters at infinity are $\chi_e$ and $\chi_{-e}$.
Both are monomial images of full projective Reed--Solomon codes,
and
\[
 \dim U=e+1,\qquad\dim V=d+1,\qquad
 \epsilon\le\frac{e+1}{n}\le\frac{d+1}{n}\le\beta.
\]
The four codes $U,V,U^\perp,V^\perp$ have relative distance at
least $\epsilon$. For fixed $h$, both endpoint rates tend to
$1/(2^h+1)$ as $\ell\to\infty$.
\end{theorem}
\begin{proof}
Euclidean division gives $0\le R\le s_h$ and $d=e+R\ge e\ge0$.
The two fractional terms in \eqref{exloc:reciprocal-parameters}
give
\[
 \frac{q-1}{D}\ge\epsilon(q+1),\qquad
 \frac{q+s_hD}{D}\le\beta(q+1).
\]
Consequently
\[
 e+1\ge\epsilon n,\qquad
 d+1=\frac{q+s_h+s_hR}{D}
       \le\frac{q+s_hD}{D}\le\beta n<n.
\]
In particular $0\le e\le d<q$.

Substitution by $g^{-1}$ preserves the homogeneous degree-$e$
polynomial space, giving invariance of $U$. For $V$, use
\[
 (g^{-1}v)^{[u]}=(g^{[u]})^{-1}v^{[u]},\qquad
 P\longmapsto P\circ(g^{[u]})^{-1}:
 E[X,Y]_d\xrightarrow{\sim}E[X,Y]_d.
\]
These identities prove invariance with homogeneity exponent $ud$.
Since $s_hu=q$ and $M=q-1$,
\[
 ud=u(M-s_he)\equiv-qe\equiv-e\pmod M.
\]
Lemma~\ref{exloc:homogeneous-representation} therefore gives
the two stated characters. The congruence concerns exponents
of rational scalar actions; no equality of polynomial degrees
is needed.

Coordinatewise Frobenius induces a permutation $\pi_u$ of
$\mathbb P^1(\mathbb F_q)$. Write
$r_x^{[u]}=c_xr_{\pi_u(x)}$ with $c_x\in\mathbb F_q^\times$.
Then $P(r_x^{[u]})=c_x^dP(r_{\pi_u(x)})$, which identifies
$V$ with a monomial image of the degree-$d$ projective code.
The degree-$e$ assertion for $U$ follows directly from its definition.

Choose a finite overfield containing $E$ and $\mathbb F_{q^2}$.
The norm-grid identification of Proposition~\ref{exloc:norm-bridge},
applied to $e,d<q$, and the Reed--Solomon dual formula
\ref{exloc:rs-dual} give over this overfield
\[
 \begin{array}{c|cccc}
 C&U&V&U^\perp&V^\perp\\ \hline
 \dim C&e+1&d+1&q-e&q-d\\
 d(C)&n-e&n-d&e+2&d+2.
 \end{array}
\]
These are identities over $E$ by the scalar-extension lemma
\ref{exloc:scalar-extension}: generator ranks, kernel ranks, and
the ranks detecting a codeword with a prescribed support do not
change on extension. In particular, for every generator matrix $G$
over $E$ and every field extension $E'/E$,
\[
 E'\otimes_E\ker G^{\mathsf T}
   =\ker\bigl((E'\otimes_EG)^{\mathsf T}\bigr).
\]
Thus the first two relative distances
are at least $1-\beta\ge\epsilon$ and the latter two are at least
$(e+1)/n\ge\epsilon$. Finally $R$ is bounded independently of
$q$, and $e=(q-1-R)/D$, $d=e+R$. Division by $q+1$ proves
the two asserted limits.
\end{proof}

\begin{corollary}\label{exloc:all-endpoint-tuples}
Use the persistent parameters and the integer $r$ of
Theorem~\ref{exdiag:actual-direction-system}. Form $U_i,V_i$
over $E_0$ from \eqref{exloc:reciprocal-pair} with
$\ell=r+a_i$. For every extension $E'/E_0$, every nonempty
$J\subseteq[t]$, and independently chosen
\[
 D_i\in\{E'\otimes_{E_0}U_i,E'\otimes_{E_0}V_i,
            (E'\otimes_{E_0}U_i)^\perp,
            (E'\otimes_{E_0}V_i)^\perp\}\quad(i\in J),
\]
the tuple $(D_i)_{i\in J}$ has weighted product expansion at least
$\rho(t,K,\epsilon)$. The same holds after arbitrary independent
monomial changes in the coordinate spaces.
\end{corollary}
\begin{proof}
Theorem~\ref{exloc:reciprocal-prs} applies over $E_0$ in each
direction because $\mathbb F_{q_i}\subseteq E_0$,
$r+a_i\ge h$, and $q_i\ge Q_{\mathrm{end}}$.
The endpoint dimensions lie in
$[\epsilon n_i,\beta n_i]\subseteq[\epsilon n_i,(1-\epsilon)n_i]$.
Their dual codes have dimensions in the same latter interval.
Proposition~\ref{exloc:norm-bridge} identifies each endpoint,
over $E_0\supseteq\mathbb F_{q_i^2}$, with a weighted full
Reed--Solomon code on the norm-one grid $\mu_{q_i+1}$.
The dual formula \ref{exloc:rs-dual} does the same for its actual
dual code. These identifications extend to $E'$ by
Lemma~\ref{exloc:scalar-extension}. The inherited offsets in $J$
remain distinct and satisfy
\[
 2^{\max_{i\in J}a_i-\min_{i\in J}a_i}\le K.
\]
Theorem~\ref{exloc:finite-grid-pe}, with this subtuple and these
independent dimensions, gives the previously fixed constant
$\rho(t,K,\epsilon)$. Lemma~\ref{exloc:monomial-isometry}
preserves that constant under the indicated monomial maps.
\end{proof}

\subsection{Induced encoders and diagrams on one tree}

Let $K_*$ be a group, $I_*\le K_*$ a finite-index subgroup,
and $W$ a finite-dimensional $E$-representation of $I_*$.
We use the induced representation
\[
 \operatorname{Ind}_{I_*}^{K_*}W
 =\{F:K_*\to W:F(ki)=i^{-1}F(k)\ (k\in K_*,i\in I_*)\},
 \qquad (k_0F)(k)=F(k_0^{-1}k).
\]
The geometric coordinate set is $K_*/I_*$, consisting of left
cosets. If $U$ is a representation of $K_*$ and $a:U\to W$
is $I_*$-equivariant, define
\[
 A:U\longrightarrow\operatorname{Ind}_{I_*}^{K_*}W,
 \qquad(Au)(k)=a(k^{-1}u).
\]
All representations used here are linear over $E$; no
semisimplicity hypothesis is imposed.

\begin{lemma}
\label{exdiag:induced-encoder}
The map $A$ is well-defined, linear, and $K_*$-equivariant.
For coset representatives $k_\omega$, evaluation is an isomorphism
\[
 \operatorname{Ind}_{I_*}^{K_*}W\xrightarrow{\sim}
                    \bigoplus_{\omega\in K_*/I_*}W.
\]
The pairing
\begin{equation}\label{exdiag:induced-pairing}
 \langle F,F'\rangle
  =\sum_{\omega\in K_*/I_*}F'(k_\omega)(F(k_\omega)),
 \quad F'\in\operatorname{Ind}_{I_*}^{K_*}W^*,
\end{equation}
is independent of the representatives and is nondegenerate and
$K_*$-invariant. Here $W^*$ carries the contragredient action.
The pairing gives a $K_*$-equivariant isomorphism
\[
 \operatorname{Ind}_{I_*}^{K_*}(W^*)
 \xrightarrow{\sim}(\operatorname{Ind}_{I_*}^{K_*}W)^*.
\]
For every invariant subspace $D\leq\operatorname{Ind}_{I_*}^{K_*}W$,
its orthogonal complement
\[
 D^\perp=\{F'\in\operatorname{Ind}_{I_*}^{K_*}(W^*):
       \langle d,F'\rangle=0\text{ for all }d\in D\}
\]
is invariant under the contragredient action.
\end{lemma}
\begin{proof}
For $i\in I_*$ and $k,k_0\in K_*$,
\[
 (Au)(ki)=a(i^{-1}k^{-1}u)=i^{-1}(Au)(k),\qquad
 A(k_0u)(k)=a(k^{-1}k_0u)=(k_0Au)(k).
\]
These verify the first assertion. The inverse of coset evaluation
sends $(w_\omega)_\omega$ to $F(k_\omega i)=i^{-1}w_\omega$;
every element of $K_*$ has a unique such expression once the
representatives are fixed.

Replacing $k_\omega$ by $k_\omega i_\omega$ changes the paired
coordinates to $i_\omega^{-1}w$ and $i_\omega^{-1}\lambda$.
By the contragredient action,
$(i_\omega^{-1}\lambda)(i_\omega^{-1}w)=\lambda(w)$.
Thus \eqref{exdiag:induced-pairing} is independent of the
representatives. In any representative coordinates it is a
direct sum of nondegenerate pairings. Left translation permutes
the cosets and makes just such paired changes of representatives,
so the pairing is invariant. Finally, for an invariant subspace
$D$, $\lambda\in D^\perp$, $k\in K_*$, and $d\in D$,
$(k\lambda)(d)=\lambda(k^{-1}d)=0$, proving invariance of $D^\perp$.
\end{proof}

\begin{proposition}\label{exdiag:reciprocal-endpoints}
Let $L$ be a local field with residue field identified with
$\mathbb F_q$, where $q=2^\ell$ satisfies the hypotheses of
Theorem~\ref{exloc:reciprocal-prs}. Use the groups
$G,K_0,K_1,I$ and the matrix $s$ from
Lemma~\ref{exdiag:lattice-tree}. Let $U_0=U$ and $U_1=V$
be the homogeneous-function representations in that theorem.
Give $U_0$ the $K_0$ action through residue reduction and $U_1$
the $K_1$ action through $k\mapsto s^{-1}ks$ followed by reduction.
There is a common one-dimensional $I$-module $W$ such that
$a_b:U_b\to W$, evaluation at $(1,0)$, is $I$-equivariant for
$b=0,1$. Both induced encoders are injective, with images the
stated endpoint codes in suitable scalar edge coordinates.
\end{proposition}
\begin{proof}
For $i\in I$, choose an integral determinant-one lift whose
diagonal residues are $(c,c^{-1})$. Let $i$ act on $W=E$ by
$c^{-e}$. This is independent of the projective lift: the only
ambiguity is multiplication by $-I_2$, which reduces to $I_2$
in characteristic two. It is a character because the residue
matrix is upper triangular.

Conjugation by $s$ exchanges the two diagonal residues by
Lemma~\ref{exdiag:lattice-tree}. The evaluation characters at
the two endpoints are therefore
\[
 \chi_e(c)=c^{-e},\qquad
 \chi_{-e}(c^{-1})=c^{-e}.
\]
Both maps have the required same target action. Identify
$K_b/I$ with $\mathbb P^1(\mathbb F_q)$ in the coordinates of
that endpoint. For each $x$, choose a determinant-one residue
matrix taking $(1,0)$ to $r_x$, and lift it to $K_b$ using the
surjectivity in Lemma~\ref{exdiag:lattice-tree}. If this lift is
$k_x$, then, in the homogeneous-function model,
\[
 (A_bF)(k_x)=a_b(k_x^{-1}F)=F(k_x(1,0))=F(r_x).
\]
For $b=1$ the matrix in this formula is its matrix in
$s$-coordinates. The evaluation maps are injective by
Theorem~\ref{exloc:reciprocal-prs}. Different representatives
multiply individual scalar coordinates by nonzero elements,
which gives the claimed coordinate freedom.
\end{proof}

A $G$-equivariant diagram on a tree assigns a vector space
$\mathcal F(c)$ to every vertex or edge $c$, a map
$a_{e,v}:\mathcal F(v)\to\mathcal F(e)$ to every $v<e$,
and invertible transports
$T_{g,c}:\mathcal F(c)\to\mathcal F(gc)$ satisfying
\[
 T_{1,c}=I,\quad T_{gg',c}=T_{g,g'c}T_{g',c},\quad
 T_{g,e}a_{e,v}=a_{ge,gv}T_{g,v}.
\]
These are equivariant cosheaves on the one-dimensional face
poset, with identity maps on equal faces.

\begin{theorem}\label{exdiag:tree-diagram}
Let $G,T,K_0,K_1,I,v_0,v_1,e_*$ be as in
Lemma~\ref{exdiag:lattice-tree}. Suppose $U_b$ is a representation
$\rho_b$ of $K_b$, $W$ is a representation $\rho_e$ of $I$,
and $a_b:U_b\to W$ is $I$-equivariant with injective induced
encoder $A_b$, for $b=0,1$. Then these data define a
$G$-equivariant diagram on $T$. Its vertex stalks have dimensions
$\dim U_b$ at type $b$ and its edge stalks have dimension $\dim W$.
At every vertex, simultaneous separate changes of coordinates
in its incident edge stalks identify the aggregate incidence
map with $A_b$.
\end{theorem}
\begin{proof}
Choose frames $p_v\in G$ with $p_vv_b=v$ for every type-$b$
vertex and $p_e\in G$ with $p_ee_*=e$ for every edge.
These exist by Lemma~\ref{exdiag:lattice-tree}. Give a type-$b$
vertex the stalk $U_b$ and an edge the stalk $W$. Define
\begin{equation}\label{exdiag:incidence-formula}
 a_{e,v}=a_b\rho_b(p_e^{-1}p_v),\qquad
 T_{g,v}=\rho_b(p_{gv}^{-1}gp_v),\qquad
 T_{g,e}=\rho_e(p_{ge}^{-1}gp_e).
\end{equation}
The type-$b$ endpoint of $e$ is both $p_ev_b$ and $p_vv_b$,
so $p_e^{-1}p_v\in K_b$. The other two group arguments fix
$v_b$ and $e_*$, respectively, so all maps in
\eqref{exdiag:incidence-formula} are defined.

For a vertex or edge $c$, the identity
\[
 p_{gg'c}^{-1}gg'p_c
    =(p_{gg'c}^{-1}gp_{g'c})(p_{g'c}^{-1}g'p_c)
\]
gives the transport composition law. The identity transport is
obtained by setting $g=1$. For incidence equivariance, let
$i=p_{ge}^{-1}gp_e\in I$. Then
\[
 \begin{aligned}
 a_{ge,gv}T_{g,v}
 &=a_b\rho_b(p_{ge}^{-1}gp_v)
  =a_b\rho_b(i)\rho_b(p_e^{-1}p_v)\\
 &=\rho_e(i)a_b\rho_b(p_e^{-1}p_v)
  =T_{g,e}a_{e,v},
 \end{aligned}
\]
where the third equality uses precisely the $I$-equivariance
of $a_b$.

Fix a type-$b$ vertex $v$. Its incident edges are
$e_k=p_vk e_*$, indexed by $kI\in K_b/I$.
After choosing coset representatives, write
$p_{e_k}=p_vk i_k$ with $i_k\in I$. The coordinate formula is
\[
 a_{e_k,v}=\rho_e(i_k^{-1})a_b\rho_b(k^{-1}).
\]
Changing the coordinates in the edge stalk of $e_k$ by
$\rho_e(i_k)$ turns this map into the $kI$ coordinate of $A_b$.
These choices are made once for every incident edge and do not
mix distinct edge stalks. They identify the entire star with
the claimed injective encoder.
\end{proof}

\begin{corollary}\label{exdiag:one-color-dual}
In Theorem~\ref{exdiag:tree-diagram}, write
$D_b=\operatorname{im}A_b\subseteq\operatorname{Ind}_I^{K_b}W$.
The representations $D_b^\perp\subseteq\operatorname{Ind}_I^{K_b}W^*$,
the common edge space $W^*$, and the maps
$a_b^\perp(\lambda)=\lambda(1)$ define a $G$-equivariant diagram.
Its endpoint star code is the orthogonal complement
$D_b^\perp$ in the paired edge coordinates. When $\dim W=1$,
its endpoint dimension is $q+1-\dim U_b$.
\end{corollary}
\begin{proof}
Lemma~\ref{exdiag:induced-encoder} shows that $D_b^\perp$ is
an invariant subspace of the indicated induced representation. For $i\in I$ and $k\in K_b$,
\[
 a_b^\perp(i\lambda)=\lambda(i^{-1})=i\lambda(1),\qquad
 (A_b^\perp\lambda)(k)
       =a_b^\perp(k^{-1}\lambda)=\lambda(k).
\]
Thus $a_b^\perp$ is $I$-equivariant and its induced map is exactly
the inclusion of $D_b^\perp$. Apply
Theorem~\ref{exdiag:tree-diagram} with these two endpoint modules
and $W^*$. In paired representative coordinates, the pairing is
\eqref{exdiag:induced-pairing}, which is the ordinary coordinate
pairing when $W$ is one-dimensional. Its nondegeneracy gives
the dimension formula.
\end{proof}

\subsection{Tensor products and simultaneous full-star charts}

For $i\in[t]$, let $G_i$ act on a typed tree $T_i$ and let
$\mathcal F_i$ be an equivariant diagram from
Theorem~\ref{exdiag:tree-diagram}, over the same field $E$.
Put $G=\prod_iG_i$ and $\mathcal T=\prod_iT_i$. For a face
$f=(f_i)_i$, and for $f\le g$, define
\begin{equation}\label{exdiag:external-product-formulas}
 \mathcal F(f)=\bigotimes_{i=1}^t\mathcal F_i(f_i),\qquad
 R_{gf}=\bigotimes_{i=1}^tR^i_{g_i f_i},\qquad
 T_{\gamma,f}=\bigotimes_{i=1}^tT^i_{\gamma_i,f_i}.
\end{equation}
Here $R^i$ is the incidence map of the tree diagram for a strict
inclusion, and the identity for an equality.

\begin{proposition}\label{exdiag:external-product}
The formulas \eqref{exdiag:external-product-formulas} define a
$G$-equivariant cosheaf on $\mathcal T$.
\end{proposition}
\begin{proof}
For $f\le g\le h$, a single tree coordinate contains no chain
of two strict inclusions. Hence
$R^i_{h_i g_i}R^i_{g_i f_i}=R^i_{h_i f_i}$, since at least one
of the first two maps is the identity. Tensoring these equations
gives $R_{hg}R_{gf}=R_{hf}$, and $R_{ff}=I$.
The transport composition law and its compatibility with
$R_{gf}$ follow by tensoring the corresponding equations in
each factor. These equalities hold on pure tensors, which span
the stalks, and therefore are equalities of linear maps.
\end{proof}

\begin{theorem}\label{exdiag:tensor-star}
Suppose in addition that every edge representation $W_i$ is
one-dimensional. Fix an anchor face $f$ of $\mathcal T$ and put
$J=\codir(f)$. For $i\in J$, let $\Omega_i$ be the incident
edges at the vertex $f_i$ and let $U_i$ be its message space.
There are injective scalar encoders
$A_i:U_i\to E^{\Omega_i}$ and simultaneous isomorphisms
\[
 \Psi_{f,g}:\mathcal F(g)\xrightarrow{\sim}
       \bigotimes_{i\in\codir(g)}U_i\qquad(g\ge f)
\]
with the following property. If $g\ge f$ and $g<g'$ adds
an edge $\omega_j$ in direction $j\in\codir(g)$, then
\begin{equation}\label{exdiag:tensor-star-equation}
 \Psi_{f,g'}R_{g'g}\Psi_{f,g}^{-1}
 =\bigl(\operatorname{ev}_{\omega_j}\circ A_j\bigr)
       \otimes\bigotimes_{i\in\codir(g)\setminus\{j\}}I_{U_i},
\end{equation}
with tensor factors in increasing direction order and the scalar
output removed. Each coordinate isomorphism acts on its own
face stalk separately.
\end{theorem}
\begin{proof}
For every active direction of $f$, choose a basis of its fixed
edge stalk and use that basis in every coface. For a missing
direction $i\in J$, every coface has either the same vertex
$f_i$ or one of its incident edges. Apply
Theorem~\ref{exdiag:tree-diagram} at that one vertex, choosing
simultaneously coordinates in all its incident edge stalks.
Since $\dim W_i=1$, choose also a scalar basis of $W_i$.
The aggregate map in these coordinates is an injective map
$A_i:U_i\to E^{\Omega_i}$.

A coface $g$ selects some of these incident edges and leaves
all other missing directions equal to their original vertices.
Tensor the coordinates just chosen for its factors. All edge
factors are scalar and can be removed by the canonical tensor
identification, leaving exactly the displayed target of
$\Psi_{f,g}$. The incidence $g<g'$ changes only the $j$th
factor; in the chosen coordinates its map is
$\operatorname{ev}_{\omega_j}A_j$. This proves
\eqref{exdiag:tensor-star-equation}.
The coordinates in direction $j$ were chosen at $f_j$ once,
independently of the edges selected in the other directions.
The same equation consequently holds for every incidence in
the entire upward star simultaneously. No coordinate map
combines vectors from different face stalks.
\end{proof}

\subsection{Descent to the arithmetic quotients}

\begin{theorem}
\label{exdiag:descent}
Let $\mathcal F$ be the equivariant cosheaf of
Proposition~\ref{exdiag:external-product}. Suppose
$\Gamma\le\prod_iG_i$ acts freely on faces, preserves all
vertex types, and has a finite quotient $X=\Gamma\backslash\mathcal T$
satisfying Definition~\ref{def:cubical-complex}. Suppose also that
for every upstairs face the quotient map is an isomorphism
from its entire upward coface poset to that of its image.
Then $\mathcal F$ descends to a cosheaf $\mathcal F_X$ on $X$.
If the edge spaces are scalar, the simultaneous charts of
Theorem~\ref{exdiag:tensor-star} descend with all their incidence
identities and separate face blocks.
\end{theorem}
\begin{proof}
For $f\in X$, define $\mathcal F_X(f)$ to be the classes of
pairs $(\widetilde f,u)$, where $\widetilde f$ is a lift of $f$
and $u\in\mathcal F(\widetilde f)$, under
\[
 (\widetilde f,u)\sim
       (\gamma\widetilde f,T_{\gamma,\widetilde f}u)
             \quad(\gamma\in\Gamma).
\]
The transport identities make this an equivalence relation.
Fixing a lift $\widetilde f$ identifies its stalk bijectively
with this set of classes: any other lift is $\gamma\widetilde f$
for a unique $\gamma$, because the action on faces is free.
Transfer the vector-space operations through this bijection.
A different lift changes the coordinates by an invertible linear
transport, so these operations are independent of the lift.

For $f\le g$, choose a lift $\widetilde f$. The full-star
hypothesis provides a unique lift $\widetilde g\ge\widetilde f$.
Set
\[
 R^X_{gf}[(\widetilde f,u)]
     =[(\widetilde g,R_{\widetilde g\widetilde f}u)].
\]
Replacing $\widetilde f$ by $\gamma\widetilde f$ replaces the
unique coface by $\gamma\widetilde g$. The identity
\[
 R_{\gamma\widetilde g,\gamma\widetilde f}T_{\gamma,\widetilde f}
       =T_{\gamma,\widetilde g}R_{\widetilde g\widetilde f}
\]
then proves well-definedness. For $f\le g\le h$, the entire
chain lifts above a fixed $\widetilde f$ by the assumed poset
isomorphism, and the composition identity upstairs proves
$R^X_{hg}R^X_{gf}=R^X_{hf}$. Identity maps descend in the same way.

For a fixed anchor $f$, choose its lift once and identify every
stalk above $f$ with the stalk at its unique coface lift.
Apply the one simultaneous upstairs chart at this lift of $f$.
The descended incidence maps in these coordinates are the
upstairs maps. Thus every chart identity holds simultaneously,
and all coordinate maps remain confined to their respective
face stalks.
\end{proof}

\begin{lemma}
\label{exdiag:quotient-charts}
Under the hypotheses of Theorem~\ref{exdiag:descent}, suppose
$X$ is $(n_1,\ldots,n_t)$-regular and every $W_i$ is scalar.
If the type-$b$ vertex space in direction $i$ has dimension
$m_{i,b}$, then $\mathcal F_X$ has compatible base-code data
with exactly these endpoint dimensions.
\end{lemma}
\begin{proof}
Fix $f\in X$, an anchor vertex $x_f\le f$, and a lift
$\widetilde f$. Its lower cube has a unique vertex
$\widetilde x_f$ over $x_f$. For a missing direction $i$,
$(\widetilde x_f)_i=\widetilde f_i$. The full-star isomorphism
at $\widetilde x_f$ therefore identifies the edges incident
to this tree vertex with $E_i(x_f)=\Omega_{f,i}$.
A product coface selecting directions $T\subseteq\codir(f)$
and incident edges $\omega_T$ contains $\widetilde f$ and
those lifted edges at $\widetilde x_f$. Its quotient is
therefore $\theta_{f,x_f}(T,\omega_T)$ by the characterization
in Lemma~\ref{lem:cubical-stars}.

Choose the simultaneous chart at $\widetilde f$ from
Theorem~\ref{exdiag:tensor-star}, descend it by
Theorem~\ref{exdiag:descent}, and choose bases of the message
spaces. Its encoders become
\[
 E_{f,i}:E^{m_{i,b_i(f)}}\hookrightarrow E^{\Omega_{f,i}}.
\]
Its stalk isomorphisms are the maps $\Psi_{f,g}$ in
Definition~\ref{def:compatible-base-codes}, and
\eqref{exdiag:tensor-star-equation}, with the geometric indexing
just verified, is precisely \eqref{eq:compatible-charts}.
The same fixed choices work for all cofaces of $f$. Repeating
this construction separately for each anchor is permitted by
that definition, which requires simultaneous incidence identities
within each star but no additional agreement between charts
at distinct anchors.
\end{proof}

\begin{proof}[Proof of Theorem~\ref{exdiag:actual-direction-system}]
The residue fields at the walking places in
Theorem~\ref{exar:tower} have cardinalities $q_i$.
Choose an isomorphism of each with the subfield
$\mathbb F_{q_i}\subseteq E_0$; finite-field uniqueness permits
these choices independently. A single such $E_0$ exists, for
example $E_0=\mathbb F_{2^d}$ for
$d=\operatorname{lcm}_{i\in[t]}2(r+a_i)$ inside
$\overline{\mathbb F}_2$.
For each $i$, the hypotheses of
Proposition~\ref{exdiag:reciprocal-endpoints} hold with
$\ell=r+a_i$, and it provides representations $U_{i,0},U_{i,1}$,
a scalar common edge representation $W_i$, and injective induced
encoders. If $i\in L$, use the resulting diagram of
Theorem~\ref{exdiag:tree-diagram}. If $i\notin L$, use its
diagram for the dual codes from
Corollary~\ref{exdiag:one-color-dual}. Their dimensions are
\eqref{exdiag:endpoint-dimensions}, and their rates and relative
distances follow from Theorem~\ref{exloc:reciprocal-prs}.

Take the external tensor product over $E$ in
Proposition~\ref{exdiag:external-product}. At every sufficiently
deep arithmetic level, Theorem~\ref{exar:tower} supplies freeness,
type preservation, the cubical intersection axiom,
$(n_1,\ldots,n_t)$-regularity, and the full upward-poset isomorphisms.
Theorem~\ref{exdiag:descent} and
Lemma~\ref{exdiag:quotient-charts} therefore produce the required
cosheaf and its compatible charts on the actual quotient.

It remains to verify product expansion for the codes in these
actual coordinates, including their dual codes. Choose
projective representatives over $\mathbb F_{q_i}$.
The corresponding endpoint codes $U_i,V_i$ and their
dual codes are defined over $E_0$. Every encoder image $C_{f,i}$
is a monomial image of the extension of one of these four codes:
\[
 C_{f,i}=M_{f,i}(E\otimes_{E_0}D_i),\qquad
 D_i\in\{U_i,V_i,U_i^\perp,V_i^\perp\}.
\]
Indeed changes of edge frames act by the scalar character of
$W_i$ or $W_i^*$, changes of coset representatives act by a
permutation and nonzero scalar in each coordinate, and changes
of message basis preserve the image. These are exactly the
coordinate changes used in the tree and tensor-star proofs.

For an arbitrary field extension $E'/E$, scalar extension of
matrices and their kernels gives
\[
 \begin{aligned}
 E'\otimes_E C_{f,i}&=M_{f,i}(E'\otimes_{E_0}D_i),\\
 (E'\otimes_E C_{f,i})^\perp
   &=M_{f,i}^{-T}(E'\otimes_{E_0}D_i^\perp).
 \end{aligned}
\]
The inverse transpose is taken with respect to the actual
coordinate pairings. In particular, if the direction already
uses a diagram for the dual codes, taking the dual code again
recovers the original finite-dimensional code.
Corollary~\ref{exloc:all-endpoint-tuples} applies to every
nonempty selected set of directions over $E'/E_0$, with every
independent endpoint and dual-code choice.
Lemma~\ref{exloc:monomial-isometry} applies to the monomial
matrices $M_{f,i}$ and $M_{f,i}^{-T}$ and preserves its constant
$\rho(t,K,\epsilon)$. This proves the complete uniform expansion
assertion. All representation and field choices preceded the
congruence tower, so none depends on $\nu$.
\end{proof}

\section{Simultaneous Realization and the Main Theorem}\label{real:section}
\subsection{Order of Parameter Selection}
Fix integers $t\geq4$, $2\leq k\leq t-2$, and $h\geq1$,
a set $L\subseteq[t]$ of size $k$, and reals $K,\epsilon,\tau$ with
\begin{equation}\label{real:parameters}
 K\geq1,\qquad 0<\epsilon<\frac1{2^h+1}<\tau<\frac12,
 \qquad G_{t,k}(\tau)>0.
\end{equation}
Choose distinct integers $a_i\geq0$ with
$2^{\max_i a_i-\min_i a_i}\leq K$, and put $a_-=\min_i a_i$.
In this order, set
\begin{equation}\label{real:endpoint-threshold}
 \begin{gathered}
 \rho=\rho(t,K,\epsilon),\qquad
 \mathcal L=\mathcal L(t,k,\rho,K),\qquad
 \mathcal Q=\mathcal Q(t,k,\rho,K),\\
 Q_{\rm end}=Q_{\mathrm{end}}(h,\epsilon,\tau),
 \end{gathered}
\end{equation}
using \eqref{exloc:uniform-pe-parameters}, \eqref{geo:abstract-constants}, and
\eqref{exloc:reciprocal-parameters}. These are fixed positive
numbers, and $0<\rho\leq1$ by its defining formula.
Choose an integer $r\geq1$ such that
\begin{equation}\label{real:degree-selection}
 r+a_-\geq h,\qquad 2^{r+a_-}\geq Q_{\rm end},\qquad
 2^{r+a_-}>16\mathcal L^2.
\end{equation}
Such an $r$ exists because all quantities on the right are already
fixed. In particular the product-expansion constant has been fixed
before choosing the local lengths. Put
\begin{equation}\label{eq:real-local-data}
 \ell_i=r+a_i,\quad q_i=2^{\ell_i},\quad n_i=q_i+1,\quad
 n=\max_i n_i,\quad s=\lcm_i(2\ell_i),\quad F=\F_{2^s}.
\end{equation}
The local data remain fixed as the congruence-tower index $\nu$ grows.

\begin{lemma}\label{real:admissible-data}
For every $t\geq4$ and $2\leq k\leq t-2$ there are choices
satisfying \eqref{real:parameters} and \eqref{real:degree-selection}.
\end{lemma}
\begin{proof}
Take $\tau=(8t2^t)^{-1}$, choose $h$ with $2^h>1/\tau$,
and set $\epsilon=1/(2(2^h+1))$, $L=[k]$, $a_i=i-1$,
$K=2^{t-1}$. Lemma~\ref{bin:rate-positive} gives
$G_{t,k}(\tau)>1/2$, and all other inequalities in
\eqref{real:parameters} follow immediately. Then take $r$
sufficiently large in \eqref{real:degree-selection}.
\end{proof}

\subsection{Geometry and Coefficients on the Same Tower}
\begin{theorem}\label{real:main}
The parameters just chosen give a family $(X_\nu,\mathcal F_\nu)$
satisfying Theorem~\ref{thm:abstract}. Its endpoint proportions have
low and high form as in Corollary~\ref{bin:rate-numbers}, with
$\epsilon\leq\alpha_i\leq\beta_i\leq\tau$, and its spectral
parameter satisfies $\lambda\mathcal L<1/2$.
\end{theorem}
\begin{proof}
Since $2\ell_i\mid s$, the field $F$ contains
$\F_{q_i^2}$, and hence $\F_{q_i}$, for every $i$.
This is the standard finite-field subfield criterion
\cite[Chapter V, Section 1]{Grillet2007}.
Apply Theorem~\ref{exloc:reciprocal-prs} in direction $i$ with
parameters $(h,\ell_i,\epsilon,\tau,F)$. Its size premise follows
from $q_i\geq2^{r+a_-}\geq Q_{\rm end}$. Denote its two polynomial
degrees by $e_i,d_i$ and put
\[
 \alpha_i=(e_i+1)/n_i,\qquad \beta_i=(d_i+1)/n_i.
\]
The theorem gives $\epsilon\leq\alpha_i\leq\beta_i\leq\tau$.

Apply Theorem~\ref{exar:tower} to $t,r,(a_i)_i$ and discard the
fixed initial segment before full radius-$2t$ cubical balls embed.
It gives an $(n_1,\ldots,n_t)$-regular tower $X_\nu$ whose parity
masses grow, with full upward stars identified with the corresponding
stars in the product of local trees. Identify the residue fields
with $\F_{q_i}\subset F$. The actions are through the prescribed
local determinant-one groups and preserve directional parity.

Use the reciprocal endpoint diagrams in directions $i\in L$ and
the diagrams for their dual codes in the other directions.
Theorem~\ref{exdiag:actual-direction-system}, applied over the common
field $F$, descends their external tensor product to a cosheaf
$\mathcal F_\nu$ on each quotient. The full-star identification is
exactly the geometric premise needed for its simultaneous charts.
The endpoint proportions are
\[
 (\mu_{i,0},\mu_{i,1})=
 \begin{cases}(\alpha_i,\beta_i),&i\in L,\\
 (1-\alpha_i,1-\beta_i),&i\notin L.
 \end{cases}
\]
All lie strictly between zero and one. The same theorem and
Theorem~\ref{exloc:finite-grid-pe} give universal product expansion
at least $\rho$ for every nonempty tuple in every star and for
the tuple of its dual codes. The distinct exponents $\ell_i$
and ratio $\max q_i/\min q_i\leq K$ verify the norm-grid premises;
all message dimensions lie in $[\epsilon n_i,(1-\epsilon)n_i]$.

The arithmetic theorem gives $\lambda\leq2/\sqrt{q_-}$.
By \eqref{real:degree-selection},
\[
 \lambda\mathcal L<\frac12,\qquad
 c=\frac{1-\lambda\mathcal L}{\mathcal Q}>
 c_0:=\frac1{2\mathcal Q}>0.
\]
It also gives $\max n_i/\min n_i\leq K$.
Corollary~\ref{bin:rate-numbers} yields
$\Gamma_k\geq G_{t,k}(\tau)>0$. Thus all geometric, coefficient,
dimension, and spectral hypotheses of Theorem~\ref{thm:abstract}
hold simultaneously. The field and all local dimensions are fixed
while $V_\nu\to\infty$.
\end{proof}

\subsection{Proofs of the Main Results}
\begin{proof}[Proof of Theorem~\ref{thm:simultaneous}]
Choose admissible parameters by Lemma~\ref{real:admissible-data}
and apply Theorem~\ref{real:main}.
\end{proof}
\begin{proof}[Proof of Theorem~\ref{thm:main-qltc}]
Apply Theorem~\ref{thm:abstract} to any one family from
Theorem~\ref{thm:simultaneous}. Its positive rate, relative distance,
and soundness lower bounds give $R,\Delta,\sigma$. The larger of
its two fixed locality bounds is a valid $w$.
\end{proof}
More explicitly, in \eqref{comb:output-constants} replace $c$ by
$c_0=(2\mathcal Q)^{-1}$, and denote the resulting $v$ by $v_0$.
The construction has
\begin{equation}\label{eq:realized-parameters}
 \begin{gathered}
 \frac{K_\nu}{N_\nu}\geq\frac{G_{t,k}(\tau)}{E_{t,k,k}(\tau)},\qquad
 \frac{d_\nu}{N_\nu}\geq\frac{c_0\xi}{sn^{t-k}},\\
 \mathsf H_\nu\succeq
 \frac{[2(1-\tau)]^{t-k}v_0}
 {sn^{t-k}(E_{t,k,k-1}(\tau)+E_{t,k,k+1}(\tau))}
             \frac{\mathsf D_{\mathscr Q_\nu}}{N_\nu}.
 \end{gathered}
\end{equation}
Indeed Corollary~\ref{comb:weaken} allows $c_0<c$, and
Corollary~\ref{bin:rate-numbers} supplies the two dimension ratios.

\subsection{A Fixed Numerical Instantiation}
\label{fixed:section}\label{fixed:data}
Take
\[
 t=4,\quad k=2,\quad L=\{1,2\},\quad h=4,\quad K=2^3,\quad
 \epsilon=2^{-5},\quad\tau=2^{-3},\quad a_i=i-1,\quad r=2^{15}.
\]
The local lengths and common coefficient field are
\[
 n_i=2^{r+i-1}+1\quad(i\in[4]),\qquad
 n=n_4,\qquad s=\lcm_{0\leq a\leq3}2(r+a),\qquad F=\F_{2^s}.
\]
In particular, every $\F_{q_i^2}$ is a subfield of $F$.

\begin{corollary}\label{fixed:main}
With these fixed choices, the codes satisfy
\[
 \frac{K_\nu}{N_\nu}\geq2^{-5},\qquad
 \frac{d_\nu}{N_\nu}\geq2^{-2^{18}},\qquad
 \mathsf H_\nu\succeq2^{-2^{20}}
                         \frac{\mathsf D_{\mathscr Q_\nu}}{N_\nu}.
\]
Both their check weights and their qubit incidences are at most
$2^{2^{17}}$.
\end{corollary}
\begin{proof}
Appendix~\ref{app:constants} verifies \eqref{real:parameters} and
\eqref{real:degree-selection}, and gives
\[
 c_*:=2^{-2^{14}}\leq c_0=(2\mathcal Q)^{-1},\qquad
 n<2^{r+4},\qquad s<2^{68}.
\]
In the abstract constants, put
\[
 z=\frac{e_3(n_1,n_2,n_3,n_4)}{2e_2(n_1,n_2,n_3,n_4)}.
\]
Then $r_2=1$, $r_3=z$, $B=2^{12}n^4$, $T=2^{49}n^{20}$,
and $\xi=B^{-1}$. For $m=\min_i n_i$, the identity
\[
 3e_3(n_1,n_2,n_3,n_4)
 =\sum_{i<j}n_i n_j\sum_{a\notin\{i,j\}}n_a
 \geq2m e_2(n_1,n_2,n_3,n_4)
\]
gives $z\geq m/3\geq1$. At $c=c_*$, therefore,
\[
 u_c=u_d=\min\{(6n^2)^{-1},c_*z\}
       \geq\frac{c_*}{2^3n^2},
\]
\[
 v=\min\left\{\frac{u_c}{2^{49}n^{20}},
                    \frac{c_*z}{2^{12}n^4}\right\}
       \geq\frac{c_*}{2^{52}n^{22}}.
\]
The last inequality uses $0<c_*\leq1$ and $n\geq1$.

The coefficient formulas at $\tau=2^{-3}$ give
\[
 G_{4,2}(\tau)=(1-2\tau)^4-2\tau^2\geq2^{-2},\qquad
 E_{4,2,2}(\tau)=4\tau^2+16\tau+4<2^3,
\]
\[
 [2(1-\tau)]^2\geq2,\qquad
 E_{4,2,1}(\tau)+E_{4,2,3}(\tau)=16\tau^2+20\tau+4<2^3.
\]
Thus the rate ratio in \eqref{eq:realized-parameters} is at least
$2^{-5}$, and its soundness dimension ratio is at least $2^{-2}$.
Corollary~\ref{comb:weaken}, \eqref{eq:realized-parameters}, and
Lemma~\ref{bin:locality} now give
\[
 \frac{d_\nu}{N_\nu}\geq\frac{c_*}{2^{12}sn^6},\qquad
 \mathsf H_\nu\succeq\frac{c_*}{2^{54}sn^{24}}
                      \frac{\mathsf D_{\mathscr Q_\nu}}{N_\nu},
\]
and both locality bounds are at most $2^4sn^3$.
Finally, for $r=2^{15}$,
\[
 \begin{aligned}
 2^{14}+12+68+6(r+4)&<2^{18},\\
 2^{14}+54+68+24(r+4)&<2^{20},\\
 4+68+3(r+4)&<2^{17}.
 \end{aligned}
\]
Together with the bounds on $n$ and $s$, these prove all the claims.
\end{proof}

\section*{AI Disclosure}



The general direction of using non-Abelian cubical complexes and robust tensor codes to design asymptotically good quantum LTCs was human. This project spans nearly a year. It was initially given to WG when he joined UIUC, and he has engaged deeply with this research problem and worked hard on it exploring several interactions of expanders and codes. 

All versions of ChatGPT Pro from 5 to 6 were used in this research. Only versions 5.6 and 6 Pro materially helped orchestrating some of the sophisticated expander-code interactions used to achieve our results. The authors take responsibility for results, proofs, and correctness of the work. 

In light of the OpenAI rumors of having solved several major TCS problems, we decided to release this paper at an earlier point. ChatGPT 6 Pro was used as an editorial assistant. 

\printbibliography

@book{Hartshorne1977,
  author = {Hartshorne, Robin},
  title = {Algebraic Geometry},
  series = {Graduate Texts in Mathematics},
  volume = {52},
  publisher = {Springer},
  year = {1977}
}

@book{Eisenbud1995,
  author = {Eisenbud, David},
  title = {Commutative Algebra with a View Toward Algebraic Geometry},
  series = {Graduate Texts in Mathematics},
  volume = {150},
  publisher = {Springer},
  year = {1995}
}

@book{Fulton1998,
  author = {Fulton, William},
  title = {Intersection Theory},
  edition = {2},
  series = {Ergebnisse der Mathematik und ihrer Grenzgebiete (3)},
  volume = {2},
  publisher = {Springer},
  year = {1998}
}

@article{MustataSchwede2014,
  author = {Musta{\c t}{\u a}, Mircea and Schwede, Karl},
  title = {A {Frobenius} variant of {Seshadri} constants},
  journaltitle = {Mathematische Annalen},
  volume = {358},
  year = {2014},
  pages = {861--878},
  doi = {10.1007/s00208-013-0976-4},
  url = {https://doi.org/10.1007/s00208-013-0976-4}
}

@article{DiPasqualeNguyenSeceleanu2023,
  author = {DiPasquale, Michael and Nguyễn, Thái Thành and Seceleanu, Alexandra},
  title = {Duality for asymptotic invariants of graded families},
  journaltitle = {Advances in Mathematics},
  volume = {430},
  year = {2023},
  eid = {109208},
  doi = {10.1016/j.aim.2023.109208},
  url = {https://doi.org/10.1016/j.aim.2023.109208}
}

@article{DS2002,
  author = {Derksen, Harm and Sidman, Jessica},
  title = {{Castelnuovo--Mumford} regularity by approximation},
  journaltitle = {Advances in Mathematics},
  volume = {188},
  year = {2004},
  pages = {104--123},
  doi = {10.1016/j.aim.2003.10.001},
  url = {https://sites.lsa.umich.edu/hderksen/wp-content/uploads/sites/614/2018/05/A.I.a.18.pdf}
}

@article{ChardinPhilippon1999,
  author = {Chardin, Marc and Philippon, Patrice},
  title = {R{\'e}gularit{\'e} et interpolation},
  journaltitle = {Journal of Algebraic Geometry},
  volume = {8},
  number = {3},
  year = {1999},
  pages = {471--481},
  url = {https://webusers.imj-prg.fr/~marc.chardin/publications/textes/11.RI.pdf}
}

@article{ChardinPhilippon2002,
  author = {Chardin, Marc and Philippon, Patrice},
  title = {Erratum to \enquote{R{\'e}gularit{\'e} et interpolation}},
  journaltitle = {Journal of Algebraic Geometry},
  volume = {11},
  number = {3},
  year = {2002},
  pages = {599--600},
  url = {https://webusers.imj-prg.fr/~marc.chardin/publications/textes/11.Err.pdf}
}

@incollection{Hochster1977,
  author = {Hochster, Melvin},
  title = {{Cohen--Macaulay} rings, combinatorics, and simplicial complexes},
  booktitle = {Ring Theory {II}},
  series = {Lecture Notes in Pure and Applied Mathematics},
  volume = {26},
  publisher = {Marcel Dekker},
  location = {New York},
  year = {1977},
  pages = {171--223},
  url = {https://sites.lsa.umich.edu/hochster/wp-content/uploads/sites/1337/2026/03/Cohen-Macaulay-Rings-Combinatorics-and-Simplicial-Complexes.pdf}
}

@book{Voight2021,
  author = {Voight, John},
  title = {Quaternion Algebras},
  series = {Graduate Texts in Mathematics},
  volume = {288},
  publisher = {Springer},
  year = {2021},
  url = {https://jvoight.github.io/quat-book.pdf}
}

@article{JordanLivne2000,
  author = {Jordan, Bruce W. and Livn{\'e}, Ron},
  title = {The {Ramanujan} property for regular cubical complexes},
  journaltitle = {Duke Mathematical Journal},
  volume = {105},
  year = {2000},
  pages = {85--103},
  doi = {10.1215/S0012-7094-00-10514-5},
  url = {https://doi.org/10.1215/S0012-7094-00-10514-5}
}

@incollection{Livne2001,
  author = {Livn{\'e}, Ron},
  title = {Communication networks and {Hilbert} modular forms},
  booktitle = {Applications of Algebraic Geometry to Coding Theory, Physics and Computation},
  series = {NATO Science Series II},
  volume = {36},
  publisher = {Springer},
  year = {2001},
  pages = {255--270},
  doi = {10.1007/978-94-010-1011-5_13},
  url = {https://doi.org/10.1007/978-94-010-1011-5_13}
}

@incollection{Blasius2006,
  author = {Blasius, Don},
  title = {{Hilbert} modular forms and the {Ramanujan} conjecture},
  booktitle = {Noncommutative Geometry and Number Theory},
  series = {Aspects of Mathematics},
  volume = {E37},
  publisher = {Vieweg},
  year = {2006},
  pages = {35--56},
  doi = {10.1007/978-3-8348-0352-8_2},
  url = {https://doi.org/10.1007/978-3-8348-0352-8_2}
}

@inproceedings{DLV2024,
  author = {Dinur, Irit and Lin, Ting-Chun and Vidick, Thomas},
  title = {Expansion of high-dimensional cubical complexes: With application to quantum locally testable codes},
  booktitle = {Proceedings of FOCS 2024},
  year = {2024},
  pages = {379--385},
  doi = {10.1109/FOCS61266.2024.00031},
  url = {https://doi.org/10.1109/FOCS61266.2024.00031}
}

@book{Grillet2007,
  author = {Grillet, Pierre Antoine},
  title = {Abstract Algebra},
  edition = {2},
  series = {Graduate Texts in Mathematics},
  volume = {242},
  publisher = {Springer},
  year = {2007}
}

@article{Dinur2007,
  author = {Dinur, Irit},
  title = {The {PCP} theorem by gap amplification},
  journaltitle = {Journal of the ACM},
  volume = {54},
  number = {3},
  year = {2007},
  eid = {12},
  doi = {10.1145/1236457.1236459},
  url = {https://doi.org/10.1145/1236457.1236459}
}

@article{BGHSV2006,
  author = {Ben-Sasson, Eli and Goldreich, Oded and Harsha, Prahladh and Sudan, Madhu and Vadhan, Salil},
  title = {Robust {PCPs} of proximity, shorter {PCPs}, and applications to coding},
  journaltitle = {SIAM Journal on Computing},
  volume = {36},
  number = {4},
  year = {2006},
  pages = {889--974},
  doi = {10.1137/S0097539705446810},
  url = {https://doi.org/10.1137/S0097539705446810}
}

@inproceedings{PanteleevKalachev2022,
  author = {Panteleev, Pavel and Kalachev, Gleb},
  title = {Asymptotically good quantum and locally testable classical {LDPC} codes},
  booktitle = {Proceedings of STOC 2022},
  year = {2022},
  pages = {375--388},
  doi = {10.1145/3519935.3520017},
  url = {https://doi.org/10.1145/3519935.3520017}
}

@article{DELLM2022,
  author = {Dinur, Irit and Evra, Shai and Livn{\'e}, Ron and Lubotzky, Alexander and Mozes, Shahar},
  title = {Good locally testable codes},
  journaltitle = {Annals of Mathematics},
  volume = {203},
  number = {2},
  year = {2026},
  pages = {511--553},
  doi = {10.4007/annals.2026.203.2.3},
  url = {https://doi.org/10.4007/annals.2026.203.2.3},
  addendum = {Conference version: \enquote{Locally testable codes with constant rate, distance, and locality}. In: \emph{Proceedings of STOC 2022}. 2022, pp.~357--374. \textsc{doi}: \href{https://doi.org/10.1145/3519935.3520024}{\nolinkurl{10.1145/3519935.3520024}}}
}

@article{AharonovEldar2015,
  author = {Aharonov, Dorit and Eldar, Lior},
  title = {Quantum locally testable codes},
  journaltitle = {SIAM Journal on Computing},
  volume = {44},
  number = {5},
  year = {2015},
  pages = {1230--1262},
  doi = {10.1137/140975498},
  url = {https://doi.org/10.1137/140975498}
}

@inproceedings{EldarHarrow2017,
  author = {Eldar, Lior and Harrow, Aram W.},
  title = {Local {Hamiltonians} whose ground states are hard to approximate},
  booktitle = {Proceedings of FOCS 2017},
  year = {2017},
  pages = {427--438},
  doi = {10.1109/FOCS.2017.46},
  url = {https://doi.org/10.1109/FOCS.2017.46}
}

@article{AharonovAradVidick2013,
  author = {Aharonov, Dorit and Arad, Itai and Vidick, Thomas},
  title = {Guest column: The quantum {PCP} conjecture},
  journaltitle = {ACM SIGACT News},
  volume = {44},
  number = {2},
  year = {2013},
  pages = {47--79},
  doi = {10.1145/2491533.2491549},
  url = {https://doi.org/10.1145/2491533.2491549}
}

@inproceedings{KalachevPanteleev2025,
  author = {Kalachev, Gleb and Panteleev, Pavel},
  title = {Maximally extendable product codes are good coboundary expanders},
  booktitle = {Proceedings of FOCS 2025},
  year = {2025},
  doi = {10.1109/FOCS63196.2025.00079},
  url = {https://doi.org/10.1109/FOCS63196.2025.00079}
}

@article{ReedSolomon1960,
  author = {Reed, Irving S. and Solomon, Gustave},
  title = {Polynomial codes over certain finite fields},
  journaltitle = {Journal of the Society for Industrial and Applied Mathematics},
  volume = {8},
  number = {2},
  year = {1960},
  pages = {300--304},
  doi = {10.1137/0108018},
  url = {https://doi.org/10.1137/0108018}
}

@article{LPS1988,
  author = {Lubotzky, Alexander and Phillips, Ralph and Sarnak, Peter},
  title = {{Ramanujan} graphs},
  journaltitle = {Combinatorica},
  volume = {8},
  number = {3},
  year = {1988},
  pages = {261--277},
  doi = {10.1007/BF02126799},
  url = {https://doi.org/10.1007/BF02126799}
}

@inproceedings{HLMRZ2025,
  author = {Hsieh, Jun-Ting and Lubotzky, Alexander and Mohanty, Sidhanth and Reiner, Assaf and Zhang, Rachel Yun},
  title = {Explicit lossless vertex expanders},
  booktitle = {Proceedings of FOCS 2025},
  year = {2025},
  doi = {10.1109/FOCS63196.2025.00046},
  url = {https://doi.org/10.1109/FOCS63196.2025.00046}
}

@article{Gallager1962,
  author = {Gallager, Robert G.},
  title = {Low-density parity-check codes},
  journaltitle = {IRE Transactions on Information Theory},
  volume = {8},
  number = {1},
  year = {1962},
  pages = {21--28},
  doi = {10.1109/TIT.1962.1057683},
  url = {https://doi.org/10.1109/TIT.1962.1057683}
}

@article{Tanner1981,
  author = {Tanner, R. Michael},
  title = {A recursive approach to low complexity codes},
  journaltitle = {IEEE Transactions on Information Theory},
  volume = {27},
  number = {5},
  year = {1981},
  pages = {533--547},
  doi = {10.1109/TIT.1981.1056404},
  url = {https://doi.org/10.1109/TIT.1981.1056404}
}

@article{SipserSpielman1996,
  author = {Sipser, Michael and Spielman, Daniel A.},
  title = {Expander codes},
  journaltitle = {IEEE Transactions on Information Theory},
  volume = {42},
  number = {6},
  year = {1996},
  pages = {1710--1722},
  doi = {10.1109/18.556667},
  url = {https://doi.org/10.1109/18.556667}
}

@article{AroraSafra1998,
  author = {Arora, Sanjeev and Safra, Shmuel},
  title = {Probabilistic checking of proofs: A new characterization of {NP}},
  journaltitle = {Journal of the ACM},
  volume = {45},
  number = {1},
  year = {1998},
  pages = {70--122},
  doi = {10.1145/273865.273901},
  url = {https://doi.org/10.1145/273865.273901}
}

@article{ALMSS1998,
  author = {Arora, Sanjeev and Lund, Carsten and Motwani, Rajeev and Sudan, Madhu and Szegedy, Mario},
  title = {Proof verification and the hardness of approximation problems},
  journaltitle = {Journal of the ACM},
  volume = {45},
  number = {3},
  year = {1998},
  pages = {501--555},
  doi = {10.1145/278298.278306},
  url = {https://doi.org/10.1145/278298.278306}
}

@article{GoldreichSudan2006,
  author = {Goldreich, Oded and Sudan, Madhu},
  title = {Locally testable codes and {PCPs} of almost-linear length},
  journaltitle = {Journal of the ACM},
  volume = {53},
  number = {4},
  year = {2006},
  pages = {558--655},
  doi = {10.1145/1162349.1162351},
  url = {https://doi.org/10.1145/1162349.1162351}
}

@article{BenSassonSudan2008,
  author = {Ben-Sasson, Eli and Sudan, Madhu},
  title = {Short {PCPs} with polylog query complexity},
  journaltitle = {SIAM Journal on Computing},
  volume = {38},
  number = {2},
  year = {2008},
  pages = {551--607},
  doi = {10.1137/050646445},
  url = {https://doi.org/10.1137/050646445}
}

@article{BenSassonSudan2006,
  author = {Ben-Sasson, Eli and Sudan, Madhu},
  title = {Robust locally testable codes and products of codes},
  journaltitle = {Random Structures \& Algorithms},
  volume = {28},
  number = {4},
  year = {2006},
  pages = {387--402},
  doi = {10.1002/rsa.20120},
  url = {https://doi.org/10.1002/rsa.20120}
}

@article{CalderbankShor1996,
  author = {Calderbank, A. Robert and Shor, Peter W.},
  title = {Good quantum error-correcting codes exist},
  journaltitle = {Physical Review A},
  volume = {54},
  number = {2},
  year = {1996},
  pages = {1098--1105},
  doi = {10.1103/PhysRevA.54.1098},
  url = {https://doi.org/10.1103/PhysRevA.54.1098}
}

@article{Steane1996,
  author = {Steane, Andrew M.},
  title = {Multiple-particle interference and quantum error correction},
  journaltitle = {Proceedings of the Royal Society of London A},
  volume = {452},
  number = {1954},
  year = {1996},
  pages = {2551--2577},
  doi = {10.1098/rspa.1996.0136},
  url = {https://doi.org/10.1098/rspa.1996.0136}
}

@article{Gottesman2014,
  author = {Gottesman, Daniel},
  title = {Fault-tolerant quantum computation with constant overhead},
  journaltitle = {Quantum Information and Computation},
  volume = {14},
  number = {15--16},
  year = {2014},
  pages = {1338--1372},
  doi = {10.26421/QIC14.15-16-5},
  url = {https://doi.org/10.26421/QIC14.15-16-5}
}

@article{TillichZemor2014,
  author = {Tillich, Jean-Pierre and Z{\'e}mor, Gilles},
  title = {Quantum {LDPC} codes with positive rate and minimum distance proportional to the square root of the blocklength},
  journaltitle = {IEEE Transactions on Information Theory},
  volume = {60},
  number = {2},
  year = {2014},
  pages = {1193--1202},
  doi = {10.1109/TIT.2013.2292061},
  url = {https://doi.org/10.1109/TIT.2013.2292061}
}

@inproceedings{LeverrierTillichZemor2015,
  author = {Leverrier, Anthony and Tillich, Jean-Pierre and Z{\'e}mor, Gilles},
  title = {Quantum expander codes},
  booktitle = {Proceedings of FOCS 2015},
  year = {2015},
  pages = {810--824},
  doi = {10.1109/FOCS.2015.55},
  url = {https://doi.org/10.1109/FOCS.2015.55}
}

@inproceedings{HastingsHaahODonnell2021,
  author = {Hastings, Matthew B. and Haah, Jeongwan and O'Donnell, Ryan},
  title = {Fiber bundle codes: Breaking the {$n^{1/2}\operatorname{polylog}(n)$} barrier for quantum {LDPC} codes},
  booktitle = {Proceedings of STOC 2021},
  year = {2021},
  pages = {1276--1288},
  doi = {10.1145/3406325.3451005},
  url = {https://doi.org/10.1145/3406325.3451005}
}

@article{PanteleevKalachev2022Distance,
  author = {Panteleev, Pavel and Kalachev, Gleb},
  title = {Quantum {LDPC} codes with almost linear minimum distance},
  journaltitle = {IEEE Transactions on Information Theory},
  volume = {68},
  number = {1},
  year = {2022},
  pages = {213--229},
  doi = {10.1109/TIT.2021.3119384},
  url = {https://doi.org/10.1109/TIT.2021.3119384}
}

@inproceedings{LeverrierZemor2022,
  author = {Leverrier, Anthony and Z{\'e}mor, Gilles},
  title = {Quantum {Tanner} codes},
  booktitle = {Proceedings of FOCS 2022},
  year = {2022},
  pages = {872--883},
  doi = {10.1109/FOCS54457.2022.00117},
  url = {https://doi.org/10.1109/FOCS54457.2022.00117}
}

@inproceedings{AnshuBreuckmannNirkhe2023,
  author = {Anshu, Anurag and Breuckmann, Nikolas P. and Nirkhe, Chinmay},
  title = {{NLTS Hamiltonians} from good quantum codes},
  booktitle = {Proceedings of STOC 2023},
  year = {2023},
  pages = {1090--1096},
  doi = {10.1145/3564246.3585114},
  url = {https://doi.org/10.1145/3564246.3585114}
}

@article{LeverrierLondeZemor2022,
  author = {Leverrier, Anthony and Londe, Vivien and Z{\'e}mor, Gilles},
  title = {Towards local testability for quantum coding},
  journaltitle = {Quantum},
  volume = {6},
  year = {2022},
  eid = {661},
  doi = {10.22331/q-2022-02-24-661},
  url = {https://doi.org/10.22331/q-2022-02-24-661}
}

@article{CrossEtAl2024,
  author = {Cross, Andrew and He, Zhiyang and Natarajan, Anand and Szegedy, Mario and Zhu, Guanyu},
  title = {Quantum locally testable code with constant soundness},
  journaltitle = {Quantum},
  volume = {8},
  year = {2024},
  eid = {1501},
  doi = {10.22331/q-2024-10-18-1501},
  url = {https://doi.org/10.22331/q-2024-10-18-1501}
}

@inproceedings{KaufmanLubotzky2014,
  author = {Kaufman, Tali and Lubotzky, Alexander},
  title = {High dimensional expanders and property testing},
  booktitle = {Proceedings of ITCS 2014},
  year = {2014},
  pages = {501--506},
  doi = {10.1145/2554797.2554842},
  url = {https://doi.org/10.1145/2554797.2554842}
}

@article{Kalachev2026,
  author = {Kalachev, Gleb V.},
  title = {High-dimensional expansion of product codes is stronger than robust and agreement testability},
  journaltitle = {Intelligent Systems. Theory and Applications},
  volume = {30},
  number = {2},
  year = {2026},
  pages = {101--119},
  url = {https://www.intsysmagazine.ru/en/issues/2026/2/article/1}
}

@article{RungtanapiromStixVdovina2019,
  author = {Rungtanapirom, Nithi and Stix, Jakob and Vdovina, Alina},
  title = {Infinite series of quaternionic 1-vertex cube complexes, the doubling construction, and explicit cubical {Ramanujan} complexes},
  journaltitle = {International Journal of Algebra and Computation},
  volume = {29},
  number = {6},
  year = {2019},
  pages = {951--1007},
  doi = {10.1142/S0218196719500371},
  url = {https://doi.org/10.1142/S0218196719500371}
}

@misc{stacks-project,
  author       = {The {Stacks project authors}},
  title        = {The Stacks project},
  howpublished = {\url{https://stacks.math.columbia.edu}},
  year         = {2026},
}
\appendix
\section{CSS conventions and binary conversion}
\label{app:css}\label{sec:css-appendix}

\subsection{Stabilizers, logical operators, and syndrome spaces}

Fix binary matrices $H_X\in\mathbb F_2^{m_X\times N}$ and
$H_Z\in\mathbb F_2^{m_Z\times N}$ with
$H_XH_Z^{\mathsf T}=0$, where $N\geq1$ and $m_X+m_Z>0$.
Use the CSS codespace $\mathscr Q\subseteq
\mathscr H_N=(\mathbb C^2)^{\otimes N}$, listed tester $\mathsf H$,
and distance operator $\mathsf D_{\mathscr Q}$ defined in
Section~\ref{sec:css-preliminaries}. Write $U_X=\operatorname{row}H_X$
and $U_Z=\operatorname{row}H_Z$. The action on a computational
basis vector is
\[
 X^xZ^z|a\rangle=(-1)^{z\cdot a}|a+x\rangle,
 \qquad x,z,a\in\mathbb F_2^N.
\]
In particular,
\begin{equation}\label{bin:pauli-product}
 (X^xZ^z)(X^{x'}Z^{z'})
       =(-1)^{z\cdot x'}X^{x+x'}Z^{z+z'}.
\end{equation}
All Pauli phases are irrelevant to their support and commutation
syndrome.

\begin{lemma}\label{bin:css-algebra}
The integer
$K=N-\operatorname{rank}H_X-\operatorname{rank}H_Z$
is nonnegative, and $\dim_{\mathbb C}\mathscr Q=2^K$.
\end{lemma}
\begin{proof}
The CSS equation gives $U_X\subseteq U_Z^\perp$, so
$\dim U_X+\dim U_Z\leq N$. For labels in $U_X\times U_Z$,
the sign in \eqref{bin:pauli-product} is one. Consequently
\[
 S=\{X^xZ^z:x\in U_X,\ z\in U_Z\}
\]
is a group of commuting self-adjoint involutions, generated by
the listed checks. Different label pairs give different matrices:
the image of $|a\rangle$ determines $x$, and its signs for all
$a$ then determine $z$. No nonzero label is a scalar identity.
Thus $|S|=2^{\dim U_X+\dim U_Z}$ and $-I\notin S$.

The group average
\[
 \Pi=|S|^{-1}\sum_{g\in S}g
\]
is self-adjoint and satisfies $\Pi^2=\Pi$, since each group
element occurs $|S|$ times in the double sum of products.
For every $h\in S$, multiplication permutes $S$, so $h\Pi=\Pi$.
Conversely the average fixes every vector fixed by $S$.
Thus $\Pi$ is the orthogonal projector onto $\mathscr Q$.

Every Pauli with nonzero label has trace zero. If $x\ne0$,
its matrix has zero diagonal. If $x=0$ and $z\ne0$, choose
$j$ with $z_j=1$; the diagonal entries for $a$ and $a+e_j$
cancel in pairs. Only the identity has trace $2^N$. It follows
that
\[
 \dim\mathscr Q=\operatorname{Tr}_{\mathbb C}\Pi
 =2^N/|S|=2^K.
\]
\end{proof}

A Pauli is a nontrivial logical operator if it commutes with all
checks and its restriction to $\mathscr Q$ is not a scalar
operator. Its minimum possible weight, with value $+\infty$
when no such Pauli exists, is the logical distance.

\begin{lemma}\label{bin:css-distance}
With the convention $\min\varnothing=+\infty$, the logical
distance is
\[
 d=\min\left\{
 \min_{z\in\ker H_X\setminus\operatorname{im}H_Z^{\mathsf T}}|z|_H,
 \min_{x\in\ker H_Z\setminus\operatorname{im}H_X^{\mathsf T}}|x|_H
 \right\}.
\]
\end{lemma}
\begin{proof}
By \eqref{bin:pauli-product}, $X^xZ^z$ commutes with every
listed check exactly when $H_Xz=0$ and $H_Zx=0$.
Its label is a stabilizer label exactly when
$x\in\operatorname{im}H_X^{\mathsf T}$ and
$z\in\operatorname{im}H_Z^{\mathsf T}$; such a Pauli acts
as the identity on $\mathscr Q$.

Conversely, let a commuting Pauli $P$ have a label outside this
stabilizer-label space. Both $P$ and $P^{-1}$ preserve
$\mathscr Q$, so its restriction is unitary. For every $g\in S$,
the label of $Pg$ is nonzero. The trace calculation in
Lemma~\ref{bin:css-algebra} gives
\[
 \operatorname{Tr}(P|_{\mathscr Q})
 =\operatorname{Tr}(P\Pi)
 =|S|^{-1}\sum_{g\in S}\operatorname{Tr}(Pg)=0.
\]
A scalar unitary on the nonzero space $\mathscr Q$ has nonzero
trace. Thus $P$ is a nontrivial logical Pauli. We have shown that
nontriviality is equivalent to at least one of $x,z$ lying outside
its indicated stabilizer image. Since
\[
 \operatorname{wt}(X^xZ^z)
 =|\operatorname{supp}x\cup\operatorname{supp}z|
 \geq\max\{|x|_H,|z|_H\},
\]
the displayed minimum is a lower bound on logical distance.
Each vector in either difference realizes the same weight as a
pure $Z$ or pure $X$ logical Pauli, proving equality. If $K=0$,
each kernel and its contained image have equal dimension, so
both differences are empty. The one-dimensional codespace then
has no nonscalar logical Pauli, consistent with the convention.
\end{proof}

For a binary pair $\eta=(\eta_X,\eta_Z)\in
\mathbb F_2^{m_X}\times\mathbb F_2^{m_Z}$, let
$\mathscr E_\eta$ be the joint check eigenspace whose eigenvalues
are $(-1)^{(\eta_X)_j}$ and $(-1)^{(\eta_Z)_j}$. Call $\eta$
consistent when this space is nonzero. For any syndrome realized
by a Pauli, define
\[
 w(\eta)=\min\{|\operatorname{supp}x\cup\operatorname{supp}z|:
             H_Xz=\eta_X,\ H_Zx=\eta_Z\}.
\]

\begin{lemma}
\label{bin:css-syndrome-decomposition}
Every consistent syndrome is realized by a Pauli. If $P$ has
syndrome $\eta$, then $\mathscr E_\eta=P\mathscr Q$.
The consistent syndrome spaces form an orthogonal decomposition
of $\mathscr H_N$, and
\[
 \mathsf D_{\mathscr Q}|_{\mathscr E_\eta}
       =w(\eta)I,\qquad
 \mathsf H|_{\mathscr E_\eta}
       =\frac{|\eta_X|_H+|\eta_Z|_H}{m_X+m_Z}I.
\]
\end{lemma}
\begin{proof}
List the commuting check involutions as $S_1,\ldots,S_m$,
where $m=m_X+m_Z$. For each $\eta\in\mathbb F_2^m$, the product
\[
 P_\eta=\prod_{j=1}^m\frac{I+(-1)^{\eta_j}S_j}{2}
\]
is the orthogonal projector onto $\mathscr E_\eta$: its factors
are commuting orthogonal projectors onto the required individual
eigenspaces. Different binary lists give orthogonal projectors
since at a differing coordinate the product contains
$(I+S_j)(I-S_j)=0$. Expanding over the two choices in every
coordinate yields $\sum_\eta P_\eta=I$.

Suppose $\eta$ is consistent. For
$a\in\ker H_X^{\mathsf T}$, the product of the $X$-checks
selected by $a$ is the identity. On a nonzero vector of
$\mathscr E_\eta$ its eigenvalue is $(-1)^{a\cdot\eta_X}$,
so $a\cdot\eta_X=0$. Hence
$\eta_X\in(\ker H_X^{\mathsf T})^\perp=\operatorname{im}H_X$.
The identical argument for products of $Z$-checks gives
$\eta_Z\in\operatorname{im}H_Z$. Choose $z,x$ solving the two
syndrome equations. Then $P=X^xZ^z$ has the prescribed
commutation signs. Commuting checks past $P$ and $P^{-1}$ gives
$P\mathscr Q\subseteq\mathscr E_\eta$ and
$P^{-1}\mathscr E_\eta\subseteq\mathscr Q$, proving equality.

For $0\leq\ell\leq N$, the Pauli neighborhood $\mathscr Q_{\leq\ell}$
is therefore exactly
\[
 \mathscr Q_{\leq\ell}
 =\bigoplus_{\substack{\eta\text{ consistent}\\w(\eta)\leq\ell}}
                                \mathscr E_\eta.
\]
The projector $\Pi_{\leq\ell}$ onto this neighborhood is identity
on precisely these summands. On $\mathscr E_\eta$, the difference
$\Pi_{\leq\ell}-\Pi_{\leq\ell-1}$ is therefore identity exactly
when $\ell=w(\eta)$ and zero otherwise. Summing these differences
with coefficient $\ell$ proves the asserted distance-operator formula.
Finally a listed check projector $(I-S_j)/2$ acts as its syndrome
bit on $\mathscr E_\eta$. Averaging all $m$ such terms proves
the tester formula, including dependent and zero rows.
\end{proof}

\subsection{Trace-dual coordinates}

Let $F=\mathbb F_{2^s}$, $s\geq1$, and write
\[
 \operatorname{tr}(a)=\operatorname{Tr}_{F/\mathbb F_2}(a)
                    =\sum_{j=0}^{s-1}a^{2^j}.
\]
This is $\mathbb F_2$-linear. Its values lie in $\mathbb F_2$
because $a^{2^s}=a$ implies $\operatorname{tr}(a)^2=\operatorname{tr}(a)$.
The defining polynomial is nonzero and has degree
$2^{s-1}<|F|$, so it cannot vanish on all of $F$. Thus
$\operatorname{tr}:F\to\mathbb F_2$ is nonzero, including when
$s=1$.

\begin{lemma}\label{bin:pairing-dual}
The pairing $(a,b)\mapsto\operatorname{tr}(ab)$ is symmetric,
bilinear, and nondegenerate. Every ordered binary basis
$\beta=(\beta_1,\ldots,\beta_s)$ has a unique dual basis
$\beta^\vee$ satisfying
$\operatorname{tr}(\beta_i\beta_j^\vee)=\delta_{ij}$.
For every $F$-linear map $A:F^a\to F^b$,
\begin{equation}\label{bin:binary-transpose}
 R_\beta(A)^{\mathsf T}=R_{\beta^\vee}(A^{\mathsf T}).
\end{equation}
\end{lemma}
\begin{proof}
Bilinearity and symmetry follow from field multiplication and
linearity of the trace. Choose $c\in F$ with
$\operatorname{tr}(c)=1$. For every $a\ne0$,
$\operatorname{tr}(a(c/a))=1$, proving nondegeneracy. Consequently
the Gram matrix $G=(\operatorname{tr}(\beta_i\beta_j))_{ij}$
is invertible. Indeed, a vector in its kernel would specify an
element pairing to zero with every basis element and hence
with every element of $F$. The columns of $G^{-1}$ give the
unique dual basis.

For $x,y\in F^a$, paired coordinate vectors satisfy
\[
 \operatorname{coord}_\beta(x)^{\mathsf T}
 \operatorname{coord}_{\beta^\vee}(y)
       =\operatorname{tr}\left(\sum_i x_iy_i\right).
\]
For $x\in F^a$ and $y\in F^b$, apply this identity to $(Ax,y)$
and $(x,A^{\mathsf T}y)$. The field dot products agree, giving
\[
 \operatorname{coord}_\beta(x)^{\mathsf T}
 R_\beta(A)^{\mathsf T}\operatorname{coord}_{\beta^\vee}(y)
 =\operatorname{coord}_\beta(x)^{\mathsf T}
 R_{\beta^\vee}(A^{\mathsf T})\operatorname{coord}_{\beta^\vee}(y).
\]
Both coordinate vectors range independently over their entire
binary spaces, which proves the matrix identity. The argument
also covers an empty input or output coordinate space.
\end{proof}

The transpose in \eqref{bin:binary-transpose} is the algebraic
dual map in the dual $F$-bases. An arbitrary change of the
original $F$-basis therefore induces its inverse transpose on
dual vectors. The statement makes no assertion that an
unpaired use of the original coordinates represents this dual.

\subsection{Proof of the binary conversion proposition}

\begin{proof}[Proof of Proposition~\ref{bin:restriction-and-soundness}]
Use the field, three based block spaces, maps $A,B$, and constants
in that proposition. Restriction of scalars preserves composition,
so
\[
 H_XH_Z^{\mathsf T}
 =R_\beta(A)R_\beta(B)=R_\beta(AB)=0.
\]
An $F$-space of dimension $d$ has binary dimension $sd$,
as seen by multiplying an $F$-basis by the elements of $\beta$.
Applied to the image of each map, this gives
$\operatorname{rank}_{\mathbb F_2}R_\beta(A)
=s\operatorname{rank}_F A$, and likewise for $B$.
Lemma~\ref{bin:css-algebra} then proves the claimed dimension
formula. The hypotheses give $N=sD_2>0$ and
$m_X+m_Z=s(D_1+D_3)>0$.

Use $\beta$ on the original spaces and $\beta^\vee$ on their
algebraic duals. Lemma~\ref{bin:pairing-dual} gives the actual
binary subspace identifications
\begin{align*}
 \ker H_X&=\operatorname{coord}_\beta(\ker A),&
 \operatorname{im}H_Z^{\mathsf T}
       &=\operatorname{coord}_\beta(\operatorname{im}B),\\
 \ker H_Z&=\operatorname{coord}_{\beta^\vee}(\ker B^{\mathsf T}),&
 \operatorname{im}H_X^{\mathsf T}
       &=\operatorname{coord}_{\beta^\vee}(\operatorname{im}A^{\mathsf T}).
\end{align*}
For either paired coordinate map on the middle space, a nonzero
block contributes between one and $sM_2$ nonzero binary
coordinates. Consequently
\begin{equation}\label{bin:block-coordinate-comparison}
 |x|_b\leq|\overline x|_H\leq sM_2|x|_b.
\end{equation}
The lower comparison and Lemma~\ref{bin:css-distance} give
$d\geq\min\{\mu_c,\mu_h\}$. This includes $K=0$, when both
kernel-minus-image sets are empty.

For a fixed $x\in C^2$, a minimum block-weight correction to
$\ker A$ exists: the possible weights form a nonempty subset
of the finite set of integers from zero to the number of blocks.
Applying the upper comparison to such a correction gives
\[
 \operatorname{dist}_H(\overline x,\ker H_X)
       \leq sM_2\operatorname{dist}_b(x,\ker A).
\]
Every nonzero output block contributes at least one binary
coordinate. Hence
\[
 |H_X\overline x|_H\geq|Ax|_b
 \geq\frac{\varepsilon_c}{sM_2}
                  \operatorname{dist}_H(\overline x,\ker H_X).
\]
Using the dual coordinates and the hypothesis for $B^{\mathsf T}$
gives in precisely the same manner
\[
 |H_Z\overline y|_H
 \geq\frac{\varepsilon_h}{sM_2}
                  \operatorname{dist}_H(\overline y,\ker H_Z).
\]
Here the correction lies in the dual middle space, whose block
dimensions are also bounded by $M_2$; its output is in $(C^1)^*$,
where only the lower weight comparison is used. Thus no upper
bound on either syndrome-block size is needed.

Fix a consistent syndrome $\eta=(\eta_X,\eta_Z)$. By
Lemma~\ref{bin:css-syndrome-decomposition}, both syndrome fibers
are nonempty. Let
\[
 w_z=\min_{H_Xz=\eta_X}|z|_H,\qquad
 w_x=\min_{H_Zx=\eta_Z}|x|_H,\qquad
 \varepsilon_* =\min\{\varepsilon_c,\varepsilon_h\}.
\]
Each fiber is a coset of its entire kernel. The preceding two
expansion inequalities therefore imply
\[
 |\eta_X|_H+|\eta_Z|_H
 \geq\frac{\varepsilon_*}{sM_2}(w_z+w_x)
 \geq\frac{\varepsilon_*}{sM_2}w(\eta).
\]
For the last inequality, independent minimum solutions $x,z$
give a Pauli with support at most $w_x+w_z$; the minimum over
all Pauli solutions is no larger. On $\mathscr E_\eta$, the
tester eigenvalue is its syndrome weight divided by
$s(D_1+D_3)$, while the distance eigenvalue is $w(\eta)$.
Thus
\[
 \mathsf H|_{\mathscr E_\eta}
 \succeq\frac{\varepsilon_*}{s^2M_2(D_1+D_3)}
                 \mathsf D_{\mathscr Q}|_{\mathscr E_\eta}.
\]
The syndrome spaces are an orthogonal decomposition, so this
inequality holds on the entire Hilbert space. Substituting
$N=sD_2$ gives exactly the soundness constant stated in the
proposition.

Finally each scalar entry of a field matrix becomes an
$s\times s$ binary multiplication matrix. A zero entry gives
a zero block, and each row and column of any such block has
at most $s$ nonzero entries. Thus field row and column weights
at most $w_F$ give binary row and column weights at most $sw_F$.
Rows of $H_X$ are rows of $R_\beta(A)$, while rows of $H_Z$
are columns of $R_\beta(B)$. This proves check weight at most
$sw_F$. A qubit occurs in at most $sw_F$ checks from each list,
giving at most $2sw_F$ occurrences altogether.
\end{proof}

\section{Auxiliary arithmetic arguments}\label{app:arithmetic}

We retain the notation $F,B,\mathcal M,S,H=B^1$ of
Section~\ref{sec:arithmetic}. The field $F$ is totally real,
the quaternion algebra $B$ is totally definite, and
$\mathcal M$ is a maximal $\mathcal O_F$-order. At a finite
place $v$, write $\mathcal M_v=\mathcal M\otimes_{\mathcal O_F}
\mathcal O_v$ and $H_v=\{x\in B_v:\operatorname{nrd}(x)=1\}$.
The additive and multiplicative adelic topologies below are the
restricted-product topologies of~\cite[Sections~27.4,27.6]{Voight2021}.

\subsection{Integral norm-one groups and principal levels}

\begin{lemma}\label{exar:integral-local-groups}
For every finite place $v$, the group
$C_v=H_v\cap\mathcal M_v$ is compact and open in $H_v$.
For every integer $c\ge1$,
\[
 C_v(c)=\{x\in C_v:x-1\in\mathfrak m_v^c\mathcal M_v\}
\]
is an open normal subgroup of finite index in $C_v$.
Moreover $\bigcap_{c\ge1}C_v(c)=\{1\}$.
The principal groups in~\eqref{exar:principal-subgroups}
are finite-index subgroups of $H(F)\cap\mathcal M_S$;
deeper principal levels are normal in shallower ones.
\end{lemma}
\begin{proof}
The completed order is a free $\mathcal O_v$-module of rank
four, since it is a lattice in the four-dimensional $F_v$-space
$B_v$ and $\mathcal O_v$ is a discrete valuation ring.
It is compact and open in $B_v$: in a lattice basis this is
the assertion that $\mathcal O_v^4$ is compact and open in
$F_v^4$. These standard local-field properties are recalled
in~\cite[Chapter~12]{Voight2021}.

For an order element $x$, reduced trace and norm are integral
\cite[Corollaries~10.3.3 and~10.3.6]{Voight2021}.
Hence the standard involution
$\overline x=\operatorname{trd}(x)-x$ preserves the order.
If $\operatorname{nrd}(x)=1$, then
$x^{-1}=\overline x$, so $C_v$ is a group.
The reduced norm is a continuous quadratic polynomial,
making $H_v$ closed in $B_v$. Thus $C_v$ is compact,
and it is open in the relative topology of $H_v$.

Reduction modulo $\mathfrak m_v^c\mathcal M_v$ is a group
homomorphism from $C_v$ into the units of the finite ring
$\mathcal M_v/\mathfrak m_v^c\mathcal M_v$.
Its kernel is $C_v(c)$. The ring is finite because its additive
group is a direct sum of four copies of
$\mathcal O_v/\mathfrak m_v^c$, each of cardinality $|k_v|^c$.
This proves normality and finite index. Its congruence
conditions are open in lattice coordinates. In the same
coordinates, $\bigcap_c\mathfrak m_v^c\mathcal M_v=0$,
which proves the intersection assertion.

For a global ideal $\mathfrak n$ supported outside $S$, combine
these finitely many reduction homomorphisms at the primes
dividing $\mathfrak n$. Their common kernel on
$H(F)\cap\mathcal M_S$ is exactly $\Gamma(\mathfrak n)$.
The finite target proves finite index. Every deeper principal
group is the intersection with further such kernels and is
therefore normal in every shallower principal group.
\end{proof}

\subsection{Compactness and finite norm-one quotients}

Let $\mathbb A_F$ be the adele ring of $F$. Write
$\mathbb B^\times$ for the restricted product of $B_v^\times$
with respect to $\mathcal M_v^\times$ at finite places, with
the finitely many infinite factors included. With normalized
local absolute values, put
\[
 \mathbb B^{(1)}
 =\left\{(x_v)_v\in\mathbb B^\times:
                  \prod_v|\operatorname{nrd}(x_v)|_v=1\right\},
 \qquad
 H(\mathbb A_F)=\{x\in\mathbb B^\times:
                         \operatorname{nrd}(x_v)=1\ \forall v\}.
\]
The norm map $\mathbb B^\times\to\mathbb A_F$ is continuous.
Indeed, it is a polynomial in each local algebra coordinate,
and at almost all finite places it maps the reference integral
unit group into $\mathcal O_v^\times$. This verifies continuity
on every basic restricted-product neighborhood. Multiplication
is locally polynomial, and inversion is
$x^{-1}=\overline x/\operatorname{nrd}(x)$.
The same description verifies the topological group operations.

\begin{lemma}
\label{exar:norm-one-compact}
There is a compact set $C_H\subset H(\mathbb A_F)$ such that
\[
                       H(\mathbb A_F)=H(F)C_H.
\]
\end{lemma}
\begin{proof}
Total definiteness implies that $B$ is a division algebra:
at a real embedding the reduced norm of every nonzero element
is a positive sum of four real squares, so its reduced norm
in $F$ is nonzero, and its inverse is
$\overline x/\operatorname{nrd}(x)$.
Fujisaki's lemma
\cite[Main Theorem~27.6.14(a), proof Steps~3--4]{Voight2021}
therefore applies to this division algebra over the global
field $F$. Its compact representative construction supplies
a compact set $C\subset\mathbb B^{(1)}$ satisfying
$\mathbb B^{(1)}=B^\times C$.

The continuous norm image $\operatorname{nrd}(C)$ is compact
in $\mathbb A_F$. Diagonal $F$ is closed and discrete there
\cite[27.4.5]{Voight2021}; its intersection with a compact set
is consequently finite. Thus
\[
 D=\operatorname{nrd}(C)\cap\operatorname{nrd}(B^\times)
\]
is finite. For each $c\in D$, choose $b_c\in B^\times$
with $\operatorname{nrd}(b_c)=c$, and define
\[
 C_c=\{z\in C:\operatorname{nrd}(z)=c\},\qquad
 C_H=\bigcup_{c\in D}b_c^{-1}C_c.
\]
Each $C_c$ is closed in $C$, since $\mathbb A_F$ is Hausdorff.
Hence $C_H$ is compact, and multiplicativity of the norm
shows $C_H\subset H(\mathbb A_F)$.

For $h\in H(\mathbb A_F)$, write $h=bz$ with $b\in B^\times$
and $z\in C$. Its norm identity gives
$c=\operatorname{nrd}(z)=\operatorname{nrd}(b)^{-1}\in D$.
Then
\[
 h=(bb_c)(b_c^{-1}z),\qquad
 bb_c\in H(F),\quad b_c^{-1}z\in C_H.
\]
This proves the required factorization. The choice of $b_c$
uses only norms already attained by rational quaternions.
\end{proof}

\begin{proof}[Proof of Proposition~\ref{exar:finite-quotient}]
Write $H_S=\prod_{v\in S}H_v$. At $v\notin S$ let
$U_v=C_v$, except at a prime dividing $\mathfrak n$, where
take
$U_v=C_v(\operatorname{ord}_v\mathfrak n)$.
The product $U=\prod_{v\notin S,\ v\text{ finite}}U_v$
is a compact open subgroup of the finite adelic group away
from $S$. Each infinite group $H(F_v)$ is the unit sphere
of Hamilton's quaternions and is compact. For a vertex-parity
pattern $b\in\{0,1\}^t$, let $K_b\le H_S$ be the compact
open stabilizer of a vertex of that type.
Lemma~\ref{exdiag:lattice-tree} identifies the type-$b$
vertex set with $H_S/K_b$.

Extend $g\in H_S$ by the identity at all other places.
This induces a map
\[
 \Gamma(\mathfrak n)\backslash H_S/K_b
 \longrightarrow
 H(F)\backslash H(\mathbb A_F)/
          (K_b\times U\times H(F_\infty)).
\]
It is well-defined: a rational element of $\Gamma(\mathfrak n)$
has its outside-$S$ coordinates in $U$, and its infinite
coordinates in $H(F_\infty)$.
It is injective. In fact, an adelic equality
$\widetilde g'=h\widetilde g k$ with $h\in H(F)$ and
$k\in K_b\times U\times H(F_\infty)$ gives
$1=h_vk_v$ at every finite $v\notin S$.
Thus $h_v\in U_v$ at all those places, precisely the defining
conditions for $h\in\Gamma(\mathfrak n)$.
At $S$ the equality is then $g'=h_Sgk_S$, with $k_S\in K_b$.

The target double-coset set is finite.
Indeed, Lemma~\ref{exar:norm-one-compact} supplies compact
representatives for the left quotient, and each right
double coset is open because its right subgroup is open.
Their open cover of the compact quotient has a finite subcover.
There are only $2^t$ vertex-parity patterns, so the quotient
has finitely many vertices. Every vertex belongs to finitely
many cells of the locally finite product of trees, and every
cell has a vertex; hence it has finitely many cell orbits.

For completeness, its cell stabilizers are finite.
The diagonal additive embedding $B\hookrightarrow
B\otimes_F\mathbb A_F$ is closed and discrete
\cite[27.6.2]{Voight2021}. The rational elements whose
$S$-coordinates lie in a fixed compact set, whose outside
finite coordinates lie in $U$, and whose infinite coordinates
lie in $H(F_\infty)$ form a finite set. Apply this to the
compact stabilizer of a product cell. A torsion-free group
has no nontrivial finite subgroup, so the cell stabilizers
are trivial.

Finally, determinant one preserves every directional parity.
The vertices of a cube have distinct full parity vectors,
so a cell stabilizer cannot permute its vertices. The free
proper action without inversions gives a cubical covering:
a sufficiently small neighborhood inside the star of a vertex
has disjoint translates, and its face incidences are preserved.
This proves all assertions of the proposition.
\end{proof}

\subsection{The torsion-free starting level}

\begin{lemma}\label{exar:free}
For the split place $\ell_0\mid p_0$ chosen in
Section~\ref{sec:arithmetic}, with $p_0$ odd and unramified in
$F$, the group $\Gamma_0=\Gamma(\ell_0)$ is torsion free.
Its action on $\mathcal T$ identifies it with its projective
image.
\end{lemma}
\begin{proof}
Choose an integral splitting at $\ell_0$, write
$v=\operatorname{ord}_{\ell_0}$, and extend $v$ to matrices
by the minimum of the entry valuations, with $v(0)=+\infty$.
Unramifiedness gives $v(p_0)=1$.
Every nonidentity local matrix in the principal level has
the form $g=I_2+\pi^sA$, where $s\ge1$ and $v(A)=0$.

If $u\ge1$ is not divisible by $p_0$, then $v(u)=0$.
In the binomial expansion of $g^u-I_2$, the first term has
valuation $s$ and every later term has valuation at least
$2s>s$. Choosing a unit entry of $A$ proves
$v(g^u-I_2)=s$.
In the expansion of $g^{p_0}-I_2$, the first term has
valuation $s+1$. For $2\le j<p_0$, the coefficient
$\binom{p_0}{j}$ has valuation one, so the corresponding
term has valuation at least $js+1>s+1$.
The last term has valuation at least $p_0s>s+1$ since
$p_0\ge3$ and $s\ge1$. Thus
$v(g^{p_0}-I_2)=s+1$.
Iteration is valid at each new positive valuation and gives
\[
 v(g^{u p_0^a}-I_2)=s+a<\infty
 \qquad(u\ge1,\ p_0\nmid u,\ a\ge0).
\]
Every positive integer is of this form; hence $g$ is not torsion.
The embedding $B\hookrightarrow B_{\ell_0}$ is injective,
so this proves torsion-freeness of $\Gamma_0$.

A determinant-one matrix acting trivially on the local lattice
tree is scalar. To see this directly, it stabilizes
$\mathcal O^2$, and hence is integral. Stabilizing
$\mathcal O e_1+\mathcal O\pi^a e_2$ for every $a\ge0$
forces its lower-left entry to be zero; interchanging the
basis vectors forces the upper-right entry to be zero.
Repeating with the basis $e_1+e_2,e_2$ forces the diagonal
entries to be equal. Thus it is $I_2$ or $-I_2$.
The faithful embedding into a split completion then shows
that a rational quaternion in the kernel of the action is
$1$ or $-1$. The latter does not lie in the identity level
modulo the odd prime $\ell_0$. The action is therefore faithful.
\end{proof}

\section{Regularity tools for coordinate-flat arrangements}\label{app:regularity}
Throughout this appendix, $k=\overline{\mathbb F}_2$. For an integer
$N\geq1$, the polynomial ring $S=k[z_1,\ldots,z_N]$ has its standard grading. For a finitely generated
graded $S$-module with minimal resolution
$F_i=\bigoplus_jS(-j)^{\beta_{i,j}}$, use
\[
 \operatorname{reg}M=\max\{j-i:\beta_{i,j}\ne0\},\qquad
 \operatorname{reg}0=-\infty.
\]
Minimal graded free resolutions and this definition are as in
\cite[\S20.5]{Eisenbud1995}. For a closed subscheme
$Z\subseteq\mathbb P^{N-1}_k$, write $I_Z$ for its saturated homogeneous
ideal and $\mathcal I_Z$ for its ideal sheaf. Thus
$I_Z:(z_1,\ldots,z_N)^\infty=I_Z$; the empty scheme has $I_Z=S$.

\begin{lemma}\label{exloc:external-regularity}
\label{exloc:external-ideal-regularity}
For $N\geq2$, a nonzero proper saturated ideal $I_Z\subset S$, and an
integer $a$,
\[
 \operatorname{reg}I_Z=\operatorname{reg}(S/I_Z)+1,
 \qquad \operatorname{reg}I_Z\leq a
 \iff H^i(\mathbb P^{N-1},\mathcal I_Z(a-i))=0\quad(i>0).
\]
If these conditions hold, then $I_Z$ is generated in degrees at most
$a$, and
\[
 S_n\twoheadrightarrow H^0(Z,\mathcal O_Z(n))\quad(n\geq a-1),
 \qquad H^j(Z,\mathcal O_Z(n))=0\quad(j>0, n\geq a-j-1).
\]
\end{lemma}
\begin{proof}
The exact-sequence inequalities of \cite[Corollary 20.19]{Eisenbud1995},
applied to $0\to I_Z\to S\to S/I_Z\to0$, give the first identity:
$\operatorname{reg}S=0$, $\operatorname{reg}(S/I_Z)\geq0$, and
$\operatorname{reg}I_Z\geq1$.
For the second, apply \cite[Exercise 20.20(b)]{Eisenbud1995} with
$\mathcal F=\mathcal I_Z$ and its section module
$\bigoplus_nH^0(\mathbb P^{N-1},\mathcal I_Z(n))=I_Z$.
The section-module identification is exactly saturation; a section
represented on the standard opens belongs to $I_Z$ after multiplication
by a power of each coordinate, hence after multiplication by a power
of the irrelevant ideal. The source's associated-point condition
holds: $\mathcal I_Z$ is a torsion-free ideal in the local domain of
projective space, so it has no closed associated points when $N-1\geq1$.

The same criterion with $a'\geq a$ gives
$H^i(\mathcal I_Z(n))=0$ for $n\geq a-i$. The generator bound follows
from the degree-zero term of the minimal resolution. Apply the ideal
sequence at twist $n$. Its $H^1(\mathcal I_Z(n))$ proves restriction
surjectivity. For $j>0$, the groups on either side of
$H^j(Z,\mathcal O_Z(n))$ are $H^j(\mathbb P^{N-1},\mathcal O(n))$
and $H^{j+1}(\mathcal I_Z(n))$. The latter vanishes in the asserted
range. The former vanishes in intermediate degrees; in top degree
$n\geq a-(N-1)-1\geq-(N-1)$ because $a\geq1$, so it also vanishes
\cite[III, Theorem 5.1]{Hartshorne1977}.
\end{proof}

For the products used here, the cohomology calculations are
\begin{equation}\label{exloc:external-projective-cohomology}
 h^0(\mathbb P^1,\mathcal O(n))=\max(n+1,0),\qquad
 h^1(\mathbb P^1,\mathcal O(n))=\max(-n-1,0),
\end{equation}
and, for projective $k$-schemes $U,V$ and coherent sheaves $\mathcal F,\mathcal G$,
\begin{equation}\label{exloc:external-kunneth}
 H^j(U\times V,\mathcal F\boxtimes\mathcal G)
 =\bigoplus_{p+q=j}H^p(U,\mathcal F)\otimes_kH^q(V,\mathcal G).
\end{equation}
The first is \cite[III, Theorem 5.1]{Hartshorne1977}; the second is~\cite[Lemma~33.29.1]{stacks-project}. Its quasi-compactness and affine-diagonal hypotheses
hold for projective schemes. Finite reduced schemes have only degree-zero
cohomology, and a coherent sheaf on a projective scheme of dimension
$h$ has no cohomology above $h$ \cite[III, Theorems 3.5 and 2.7]{Hartshorne1977}.

\begin{lemma}\label{ped:projective-parameter-power}
Let $N,D\geq1$. If $f_1,\ldots,f_N\in S_D$ have no common projective
zero, then
\[
 (z_1,\ldots,z_N)^{N(D-1)+1}\subseteq(f_1,\ldots,f_N).
\]
\end{lemma}
\begin{proof}
The homogeneous Nullstellensatz gives
$\mathfrak m^s\subseteq(f_1,\ldots,f_N)$ for some $s\geq1$.
The powers $z_1^s,\ldots,z_N^s$ are a regular sequence, as follows
successively from their distinct variables. Localize at $\mathfrak m$.
The ideal generated by the $N$ forms contains this regular sequence of
length $N$, so \cite[Corollary 17.7]{Eisenbud1995} makes
$f_1,\ldots,f_N$ a regular sequence in $S_{\mathfrak m}$.
Its positive Koszul homology therefore vanishes there. Each homology
module is a finite graded $S$-module, and a nonzero such module cannot
vanish at $\mathfrak m$: its homogeneous annihilator is proper and
contained in $\mathfrak m$. Hence the Koszul complex is globally exact.
The exact Koszul resolution \cite[Corollary 17.5]{Eisenbud1995} gives Hilbert series
\[
 \frac{(1-z^D)^N}{(1-z)^N}=(1+z+\cdots+z^{D-1})^N.
\]
All graded pieces above $N(D-1)$ vanish, which is the stated inclusion.
\end{proof}

\begin{lemma}\label{exloc:external-approximation}
Let $I\subseteq I_j\subsetneq S$ and $A_j\subsetneq S$ be homogeneous
ideals, indexed by a nonempty finite set. Suppose
\[
 A_jI_j\subseteq I,\qquad
 (z_1,\ldots,z_N)^\tau\subseteq\sum_j A_j,\qquad
 \operatorname{reg}(S/I_j)\leq b,
\]
where $\tau\geq1$ and $b\geq0$ are integers. Then
\[
 \operatorname{reg}(S/I)\leq b+1+(\tau-1)N.
\]
\end{lemma}
\begin{proof}
Apply \cite[Theorem 3.5 and Definition 3.1]{DS2002} to
$M=S/I$ and the quotient maps $M\to M_j=S/I_j$.
The ideals $A_j$ annihilate their kernels, and their sum contains
$\mathfrak m^\tau$, so this is an approximation system of degree $\tau$.
The source has $n+1=N$ variables. Set its parameters
$t=\tau$ and $r=b+1+(\tau-1)N$. Its target bound is
$r-(t-1)(n+1)-1=b$, and its generator bound for $M$ is
$r-(t-1)(n+1)=b+1$. The latter holds because $M$ is generated in degree
zero. The field $k$ is infinite, as required by that theorem.
\end{proof}

\begin{lemma}\label{exloc:initial-betti-comparison}
\label{exloc:sr-regularity}
For a homogeneous ideal $I\subset S$ and a term order,
$\operatorname{reg}(S/I)\leq\operatorname{reg}(S/\operatorname{in}I)$.
For an integer $s\geq0$, if a squarefree monomial ideal $J$ has no
squarefree standard monomial of degree greater than $s$, then $\operatorname{reg}(S/J)\leq s$.
\end{lemma}
\begin{proof}
The homogeneous Gr\"obner degeneration of
\cite[Theorem 15.17]{Eisenbud1995} has special fiber $S/\operatorname{in}I$
and general fiber isomorphic to $S/I$. In each internal degree its
Koszul complex is a finite complex of finite free modules over the
parameter ring. Matrix ranks can drop at specialization, so homology
dimensions can only increase. Since these dimensions are the graded
Betti numbers, the first inequality follows.

For the second, the case $J=S$ is immediate. Otherwise let $\Delta$ be the simplicial complex whose faces are
the squarefree standard monomials of $J$. Hochster's formula, with its
multigrading as in the proof of \cite[Theorem 5.1, pp.~194--198]{Hochster1977},
is
\[
 \beta_{i,j}(S/J)=\sum_{|W|=j}
 \dim_k\widetilde H^{j-i-1}(\Delta_W;k).
\]
All faces have cardinality at most $s$, so every $\Delta_W$ has
dimension at most $s-1$. A nonzero summand therefore has $j-i\leq s$.
This includes variables belonging to $J$ (absent vertices) and degree-zero
Betti numbers. Taking the maximum proves the claim.
\end{proof}

\section{Verification of the Fixed Constants}\label{app:constants}
Throughout this appendix, take $t=4$, $k=2$, $K=2^3$,
$\epsilon=2^{-5}$, $\tau=2^{-3}$, and $h=4$, as in
Section~\ref{fixed:section}. Write
\[
 \rho=\rho(4,2^3,2^{-5}),\qquad \rho_*=2^{-2^{11}}.
\]

\begin{lemma}\label{fixed:constants}
The constants satisfy
\[
 \rho>\rho_*,\qquad
 \max\{\mathcal L(4,2,\rho,2^3),\mathcal Q(4,2,\rho,2^3)\}
       <2^{64}\rho_*^{-6}.
\]
\end{lemma}
\begin{proof}
By \eqref{exloc:arrangement-constants},
\[
 c_4=1+2^8\bigl(2^3(2^4-1)-1\bigr)<2^{15}.
\]
For the incidence recurrence \eqref{exinc:constant-recurrence},
$P_2=K+1<2^4$ and $B_4(K)<2^7$. Its next two steps give
\[
 \begin{array}{c|cccc}
 e&S_e&L_e&H_e&P_e\\\hline
 3&<2^{13}&<2^{18}&<2^{24}&<2^{31}\\
 4&<2^{43}&<2^{50}&<2^{59}&<2^{66}
 \end{array}
\]
Indeed $S_e=2^{3e}P_{e-1}$, while $4\cdot3!<2^5$,
$4\cdot4!<2^7$, $3^3<2^5$, and $4^4=2^8$; use
$1+2^a\leq2^{a+1}$ for every integer $a\geq0$.
Since $\gamma_4=1/4$, \eqref{exloc:uniform-pe-parameters} yields
\[
 \eta>2^{-26},\qquad
 \theta=\min\{2^{-1},(\eta/P_4)^4\}>2^{-4(26+66)},
\]
\[
 \rho=2^{-3}\left(\frac{2^{-5}\theta}{4c_4K}\right)^4
       >2^{-3-4(25+4(26+66))}>2^{-2^{11}}.
\]
Every bound has been obtained before fixing the local lengths.

Put $D_\rho=\rho^{-1}\geq1$. For $j\in\{2,3\}$ and
$0\leq\ell\leq j$, \eqref{ten:filling-constants} gives
\[
 \left(\prod_{a=\ell}^j\kappa_{4-a,j-a}\right)^{-1}
 \leq (K+1)^3D_\rho^6<2^{12}D_\rho^6.
\]
For $j=2$, the exponent sums of $D_\rho$ are $6,3,1$ and
those of $K+1$ are $3,1,0$, as $\ell=0,1,2$.
For $j=3$, the corresponding sums are $4,3,2,1$ and all zero.
The six numerator factors in \eqref{geo:Btilde} satisfy
\[
 \begin{gathered}
 K^{j-\ell+1}\leq2^{12},\quad
 \binom j\ell\leq2^2,\quad
 \binom{4-\ell}{j-\ell}\leq2^3,\\
 4-\ell\leq2^2,\quad 2^{j-\ell}\leq2^3,\quad
 (j-\ell)!\leq2^3.
 \end{gathered}
\]
Consequently $\widehat b_{j\ell}<2^{37}D_\rho^6$.
Also
\[
 \widehat\beta_{j\ell}
 =\binom4\ell2^{3-\ell}K^j\leq2^3\cdot2^3\cdot2^9=2^{15}.
\]
There are at most four summands in each sum defining
$\widehat L_j$ and $\widehat Q_j$. Therefore
\[
 \mathcal L<2^{39}D_\rho^6,\qquad
 \mathcal Q<2^{54}D_\rho^6,
\]
which imply the asserted common bound because $D_\rho<\rho_*^{-1}$.
\end{proof}

\begin{lemma}\label{fixed:size}
The choices in Section~\ref{fixed:section} satisfy
\eqref{real:parameters} and \eqref{real:degree-selection}. Moreover
\[
 n<2^{r+4},\qquad s<2^{68},\qquad
 0<c_*:=2^{-2^{14}}<(2\mathcal Q)^{-1}.
\]
\end{lemma}
\begin{proof}
The endpoint and rate premises follow from
\[
 0<2^{-5}<\frac1{2^4+1}<2^{-3}<\frac12,\qquad
 G_{4,2}(2^{-3})\geq2^{-2}>0.
\]
With $D=2^4+1$ and $s_h=2^4$ in
\eqref{exloc:reciprocal-parameters}, direct substitution gives
\[
 \frac{1+D\epsilon}{1-D\epsilon}<2^2,\qquad
 \frac{D}{D\tau-1}<2^4,\qquad s_h-\tau<2^4.
\]
Thus $Q_{\mathrm{end}}(4,2^{-5},2^{-3})<2^8$.
The offsets have gap three, giving the ratio bound $2^3=K$.
Since $r=2^{15}$, both $r\geq h$ and $2^r\geq Q_{\rm end}$ hold.
Lemma~\ref{fixed:constants} gives
\[
 16\mathcal L^2<2^{4+2(64+6\cdot2^{11})}<2^{2^{15}}=2^r,
\]
\[
 (2\mathcal Q)^{-1}>2^{-1-64-6\cdot2^{11}}>2^{-2^{14}}=c_*.
\]
This proves the strict spectral premise and the bound on $c_*$.

The exact local data satisfy
\[
 n=2^{r+3}+1<2^{r+4},\qquad
 s=\lcm_{0\leq a\leq3}2(r+a)
 \leq\prod_{a=0}^3 2(r+a)<(2^{17})^4=2^{68}.
\]
For the last strict inequality, use $2(r+3)<2^{17}$.
The field degree remains the indicated least common multiple;
the upper bound on $s$ is used only in the code-parameter estimates.
\end{proof}

\end{document}